%% file: main.tex
\documentclass[a4paper,USenglish,runningheads]{llncs}

\usepackage{amsfonts, amssymb, amsmath, comment, xcolor}
\usepackage{mathtools}
\usepackage{environ}
\usepackage{tabularx, booktabs, orcidlink}
\usepackage{academicons}

\usepackage[utf8]{inputenc}
\usepackage[T1]{fontenc}
\usepackage{academicons}
\usepackage{balance}
\usepackage[firstpage]{draftwatermark}
\usepackage{etoolbox}
\AtEndEnvironment{proof}{\phantom{}\qed}

\usepackage{caption, subcaption}
\usepackage[noend]{algpseudocode}
\usepackage{algorithm}
\usepackage{tikz}
\usetikzlibrary{arrows,automata,positioning,backgrounds,calc,patterns}
\usetikzlibrary{arrows.meta}
\usetikzlibrary{decorations.pathreplacing, shapes} %
\usepackage{mathpartir}
\usepackage{framed}
\usepackage{ulem}
\usepackage{adjustbox}
\usepackage{thmtools}
\usepackage{hyperref}
\usepackage{cleveref}
\usepackage{thm-restate}
\usepackage{wrapfig}
\usepackage{listings}
\usepackage{xparse}
\usepackage{soulutf8}

\let\so\soSoul

\usepackage[framemethod=tikz]{mdframed}

\usepackage[scr=boondoxupr]{mathalpha}

\input{commands/vldb}
\tikzset{transaction state/.style={draw=black!0}}
\tikzset{
	arrow/.pic={\path[tips,every arrow/.try,->,>=#1] (0,0) -- +(.1pt,0);},
	pics/arrow/.default={triangle 90},
	lab dis/.store in=\LabDis,
	lab dis=0.3,
	->-/.style args={at #1 with label #2}{decoration={
			markings,
			mark=at position #1 with {\arrow{>}; \node at (0,\LabDis) {#2};}},postaction={decorate}},
	-<-/.style args={at #1 with label #2}{decoration={
			markings,
			mark=at position #1 with {\arrow{<}; \node at (0,\LabDis)
				{#2};}},postaction={decorate}},
	-*-/.style={decoration={
			markings,
			mark=at position #1 with {\fill (0,0) circle (1.5pt);}},postaction={decorate}}
}

\definecolor{pinegreen}{rgb}{0.0, 0.47, 0.44}

\definecolor{enriqueColor}{HTML}{b3e6ff}

\newcommand\oldver[1]{}

\newcommand{\key}{x}
\newcommand{\Keys}{\mathsf{Keys}}
\newcommand{\Events}{\mathcal{E}}

\definecolor{coColor}{HTML}{b37200}
\definecolor{soColor}{HTML}{800080}
\definecolor{poColor}{HTML}{C41E3A}
\definecolor{wrColor}{HTML}{008000}
\definecolor{orColor}{HTML}{0026ff}
\definecolor{x1Color}{HTML}{003399}
\definecolor{x2Color}{HTML}{ff8000}
\definecolor{x3Color}{HTML}{008040}

\newcommand{\tup}[1]{{({#1})}}

\newcommand{\ebegin}{\mathtt{begin}}
\newcommand{\eabort}{\mathtt{abort}}
\newcommand{\ecommit}{\mathtt{commit}}

\newcommand{\yvar}{{y}}

\newcommand{\tr}{t}

\newcommand{\hist}{{h}}
\newcommand{\po}{\textcolor{poColor}{\mathsf{po}}}
\newcommand{\so}{\textcolor{soColor}{\mathsf{so}}}
\newcommand{\co}{\textcolor{coColor}{\mathsf{co}}}
\newcommand{\pco}{\textcolor{coColor}{\mathsf{pco}}}

\newcommand{\wro}{\textcolor{wrColor}{\mathsf{wr}}}

\newcommand{\readOp}[1]{\mathsf{reads}({#1})}
\newcommand{\tlogs}[1]{\mathsf{tr}({#1})}
\newcommand{\transC}[1]{\mathsf{cmtt}({#1})}
\newcommand{\transNotPending}[1]{\mathsf{complete}({#1})}

\newcommand{\writeOp}[1]{\mathsf{writes}({#1})}
\newcommand{\writeVar}[2]{{#1}\ \mathsf{writes}\ {#2}}
\newcommand{\deleteVar}[2]{{#1}\ \mathsf{deletes}\ {#2}}
\newcommand{\readVar}[2]{{#1}\ \mathsf{reads}\ {#2}}
\renewcommand{\hist}{{h}}

\newcommand{\btrue}{\mathsf{true}}
\newcommand{\bfalse}{\mathsf{false}}

\newcommand{\init}{\textup{\textbf{\texttt{init}}}}
\newcommand{\result}{\textup{\textbf{\texttt{res}}}}

\newcommand{\localWrite}{\textup{\textbf{\texttt{localWr}}}}

\algrenewcommand\algorithmicindent{1.0em}%
\algnewcommand\algorithmicswitch{\textbf{switch}}
\algnewcommand\algorithmiccase{\textbf{case}}
\algnewcommand\algorithmicassert{\texttt{assert}}
\algnewcommand\Assert[1]{\State \algorithmicassert(#1)}%
\algdef{SE}[SWITCH]{Switch}{EndSwitch}[1]{\algorithmicswitch\ #1\ \algorithmicdo}{\algorithmicend\ \algorithmicswitch}%
\algdef{SE}[CASE]{Case}{EndCase}[1]{\algorithmiccase\ #1\ :}{\algorithmicend\ \algorithmiccase}%
\algtext*{EndSwitch}%
\algtext*{EndCase}%
\algtext*{EndFor}%
\algnewcommand\Let{\State\textbf{let }}
\algnewcommand\ReturnAlgorithmic{\State\ReturnNameAlgorithmic}
\algnewcommand\ReturnNameAlgorithmic{\textbf{return }}
\algnewcommand\Break{\State\textbf{break }}
\algnewcommand\Continue{\State\textbf{continue }}
\algnewcommand\InputAlgorithmic{\Statex\textbf{Procedure }} 

\algnewcommand\OutputAlgorithmic{\Statex\textbf{Output: }}
\algblockdefx[BlockName]{ForTo1}{EndForTo1}%
[3][i]{\textbf{for} #1 $\gets$ #2 \textbf{to} #3 \textbf{do}}%

\lstdefinelanguage{MyLang}{%
	keywords = { delete, do, each, else, export, finally, for, foreach, function,
		if, in, let, of, return, void, while, with, yield, elements, read, write,
		insert, remove, add, AddItem, DeleteItem, Push, Pop, Enroll, Tweet, Timeline,
		NewsFeed, begin, end, break, throw},
	morecomment = [l]{//},
	morecomment = [s]{/*}{*/},
	morestring  = [b]',
	morestring  = [b]",
	sensitive   = true,
}

\lstdefinelanguage{Java10}{
	language      = Java,
	morekeywords  ={ var },
}

\SetWatermarkAngle{0}
\SetWatermarkText{\raisebox{14.5cm}{
\hspace{0.1cm}
\includegraphics{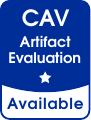}
\hspace{9cm}
}}

\begin{document}

\input{commands/complexity-defs}

\input{commands/isolation-levels}

\input{commands/transactions}

\input{commands/sql}

\input{commands/utility}

\title{On the Complexity of Checking Mixed Isolation Levels for SQL Transactions}%

\author{Ahmed Bouajjani\inst{1}\orcidlink{0000-0002-2060-3592} \and
Constantin Enea\inst{2}\orcidlink{0000-0003-2727-8865} \and Enrique Rom\'an-Calvo\inst{1}\orcidlink{0009-0005-7539-2330}}%
\authorrunning{Ahmed Bouajjani \and
Constantin Enea \and Enrique Rom\'an-Calvo}

\institute{Université Paris Cité, CNRS, IRIF \\ \email{\{abou, calvo\}@irif.fr}
\and LIX, École Polytechnique, CNRS and \\ Institut Polytechnique de Paris \\
\email{cenea@lix.polytechnique.fr}}%

\maketitle

\vspace{-7mm}
\begin{abstract}
	Concurrent accesses to databases are typically grouped in transactions which define units of work that should be isolated from other concurrent computations and resilient to failures. Modern databases provide different levels of isolation for transactions that correspond to different trade-offs between consistency and throughput. Quite often, an application can use transactions with different isolation levels at the same time. In this work, we investigate the problem of testing isolation level implementations in databases, i.e., checking whether a given execution composed of multiple transactions adheres to the prescribed isolation level semantics. We particularly focus on transactions formed of SQL queries and the use of multiple isolation levels at the same time. We show that many restrictions of this problem are NP-complete and provide an algorithm which is exponential-time in the worst-case, polynomial-time in relevant cases, and practically efficient. %
\end{abstract}

	\sloppy

	\input{sections/introduction}

	\input{sections/histories}

	\input{sections/multi-isolation}
	\input{sections/full-saturable-polynomial}
	\input{sections/observable-np}

	\input{sections/csob-algorithm}

	\input{sections/experiments}

	\input{sections/related-work}

	\input{sections/acknowledgements}

\newpage

	\bibliography{bibliography/main,bibliography/dblp,bibliography/acmart,bibliography/misc}

	\appendix
	\counterwithin{figure}{section}
	\counterwithin{table}{section}

	\newpage
	\input{appendices/operational-semantics}

	\newpage
	\input{appendices/proof-theorems}

\end{document}

%% file: commands/complexity-defs.tex
\newcommand{\roundTransaction}[1]{\mathtt{round}(#1)}
\newcommand{\visConsPrefix}[4]{\mathsf{vp}_{#1}^{#2}(#3, #4)}
\newcommand{\isConsistentExtension}[3]{\textup{\textsc{isConsistentExtension}}(#1, #2, #3)}

\newcommand{\dotText}[1]{\textup{\textit{\.{#1}}}}
\newcommand{\pathTransaction}[1]{\mathtt{path}(#1)}
\newcommand{\partition}[1]{\mathtt{partition}(#1)}
\newcommand{\prefixVar}{\mathsf{p}}
\newcommand{\coprefixVar}{\co\mathsf{p}}
\newcommand{\almostPhi}[1]{\mathsf{almost}(#1)}
\newcommand{\sessionHistory}[1]{\mathtt{session}(#1)}

\newcommand{\widthHistory}[1]{\mathtt{width}(#1)}
\newcommand{\conflictsHistory}[1]{\mathtt{\#conf}(#1)}
\newcommand{\csobAlgorithm}{\textup{\textsc{exploreConsistentPrefixes}}}
\newcommand{\checksobound}{\textup{\textsc{checkConsistency}}}
\newcommand{\zeroWhere}[1][\pco]{\mathtt{0}_x^r(#1)}
\newcommand{\oneWhere}[1][\pco]{\mathtt{1}_x^r(#1)}
\newcommand{\computeCO}[2]{\textup{\textsc{saturate}}(#1, #2)}

\newcommand{\minCoprefixVar}[2][i]{\mathsf{min}\coprefixVar^{\sigma\pi}_{#1}(#2)}

%% file: commands/isolation-levels.tex
\definecolor{coColor}{HTML}{b37200}
\definecolor{soColor}{HTML}{800080}
\definecolor{poColor}{HTML}{C41E3A}
\definecolor{wrColor}{HTML}{008000}
\definecolor{orColor}{HTML}{0026ff}
\definecolor{x1Color}{HTML}{003399}
\definecolor{x2Color}{HTML}{ff8000}
\definecolor{x3Color}{HTML}{008040}

\newcommand{\CC}{\textup{\texttt{TCC}}}
\newcommand{\TCC}{\textup{\texttt{TCC}}}
\newcommand{\xCC}{\textup{\texttt{x-TCC}}}
\newcommand{\SER}{\textup{\texttt{SER}}}
\newcommand{\PRE}{\textup{\texttt{PC}}}
\newcommand{\SI}{\textup{\texttt{SI}}}
\newcommand{\RA}{\textup{\texttt{RA}}}
\newcommand{\RC}{\textup{\texttt{RC}}}
\newcommand{\PSI}{\textup{\texttt{PSI}}}
\newcommand{\SC}{\textup{\texttt{SC}}}

\newcommand{\PO}{\textcolor{poColor}{\textsf{po}}}
\newcommand{\SO}{\textcolor{soColor}{\textsf{so}}}
\newcommand{\WR}{\textcolor{wrColor}{\textsf{WR}}}
\newcommand{\WW}{\textcolor{red}{\textsf{WW}}}
\newcommand{\RW}{\textcolor{red}{\textsf{RW}}}
\newcommand{\VIS}{\textcolor{red}{\textsf{vis}}}
\newcommand{\CO}{\textcolor{coColor}{\textsf{co} }}

\newcommand{\axpre}{\textsf{Prefix}}
\newcommand{\axconf}{\textsf{Conflict}}
\newcommand{\axser}{\textsf{Serializability}}
\newcommand{\axrc}{\textsf{Read Committed}}
\newcommand{\axcc}{\textsf{Causal}}
\newcommand{\axra}{\textsf{Read Atomic}}

\newcommand{\ext}{\textsc{Ext}}
\newcommand{\session}{\textsc{Session}}
\newcommand{\transvis}{\textsc{TransVis}}
\newcommand{\prefix}{\textsc{Prefix}}
\newcommand{\noconflict}{\textsc{NoConflict}}
\newcommand{\totalvis}{\textsc{TotalVis}}

%% file: commands/transactions.tex
\newcommand{\Obj}{\mathsf{Obj}}
\newcommand{\Var}{\mathsf{Vars}}
\newcommand{\Vars}{\mathsf{Keys}}
\newcommand{\VarsFomula}[1]{\mathsf{Vars}(#1)}
\newcommand{\LVars}{\mathsf{LVars}}
\newcommand{\Val}{\mathsf{Vals}}
\newcommand{\Vals}{\mathsf{Rows}}
\newcommand{\setRows}{\mathsf{R}}
\newcommand{\mapRows}{\mathsf{U}}
\NewDocumentCommand{\keyof}{m+o}{\mathsf{key}\IfValueT{#1}{(#1)}}
\newcommand{\row}{\mathsf{r}}
\newcommand{\Iso}{\mathsf{Iso}}
\newcommand{\Columns}{\mathbb{C}}
\newcommand{\Rows}{\mathsf{Rows}}
\newcommand{\Tables}{\mathbb{T}}
\newcommand{\OId}{\mathsf{OpId}}
\newcommand{\Op}{\mathsf{Op}}
\newcommand{\xvar}{{x}}
\newcommand{\prog}{{\mathsf{P}}}
\newcommand{\histOf}[2][]{{\mathsf{hist}_{#1}({#2})}}
\newcommand{\values}{\mathsf{v}}
\newcommand{\predicate}{\mathsf{p}}
\newcommand{\predicateUpdate}{\mathsf{q}}
\newcommand{\where}[1]{\whereName(#1)}
\newcommand{\whereName}{\mathtt{WHERE}}
\newcommand{\valuewrName}[1][\wro]{\mathtt{value}_{#1}}
\newcommand{\valuewr}[3][\wro]{\valuewrName[#1](#2, #3)}
\newcommand{\setUpdate}[1]{\mathtt{SET}(#1)}
\newcommand{\Predicates}{\mathscr{P}}
\newcommand{\del}{\dagger}
\newcommand{\id}[1]{\mathtt{id}(#1)}
\newcommand{\val}[1]{\mathtt{val}(#1)}
\newcommand{\var}[1]{\mathtt{var}(#1)}
\newcommand{\sesnapshot}[1]{\so \text{-} \mathtt{snapshot}(#1)}
\newcommand{\timeTransaction}[1]{\mathtt{time}(#1)}
\newcommand{\selectOperational}{\mathsf{readFrom}}
\newcommand{\commitOperational}[1][\iota]{\mathsf{validate}_{#1}}
\newcommand{\snapshot}[1][\iota]{\mathsf{snapshot}_{#1}}
\newcommand{\choice}[1]{\mathsf{choice}\left(#1\right)}
\newcommand{\historyExecution}[1]{\mathsf{history}(#1)}
\newcommand{\isolationExecution}[1]{\mathsf{iso}(#1)}
\newcommand{\isolation}[1]{\mathsf{iso}(#1)}
\newcommand{\labelRule}[1]{\mathsf{rule}(#1)}
\newcommand{\signTransaction}[1]{\mathtt{sign}(#1)}
\newcommand{\opSignTransaction}[1]{\mathtt{opsign}(#1)}

\newcommand{\confWr}{\mathsf{WrCons}}
\newcommand{\idRelation}{\mathtt{id}}
\newcommand{\trans}[1]{\mathsf{tr}(#1)}
\newcommand{\length}[1]{\mathsf{len}(#1)}
\newcommand{\last}[1]{\textup{\texttt{last}}(#1)}
\newcommand{\lastTr}[1]{\mathsf{lastTr}({#1})}

\newcommand{\axiomConstraint}{\mathsf{constraint}}
\newcommand{\visibility}{\mathsf{vis}}
\newcommand{\visibilityRelation}[5]{\mathsf{vis}_{#1}^{#2}( #3, #4, #5)}
\newcommand{\visibilityRelationAxiom}[3]{\visibilityRelation{a}{}{#1}{#2}{#3}}
\newcommand{\visibilityRelationAxiomCo}[3]{\visibilityRelation{a}{\co}{#1}{#2}{#3}}
\newcommand{\visibilityRelationInstanceApplied}[4][\co]{\mathsf{v}(#1)(#2, #3, #4)}
\newcommand{\visibilityRelationInstance}{\mathsf{v}}
\newcommand{\visibilitySet}[1]{\mathsf{vis}(#1)}
\newcommand{\constraintAxiom}[3][\iota]{C_{#1}^{#2}(#3)}
\newcommand{\constraintAxiomIota}[2][\co]{\constraintAxiom{#1}{#2}}

%% file: commands/sql.tex
\newcommand{\events}[1]{\mathsf{events}({#1})}
\newcommand{\sessionTransaction}[1]{\mathtt{ses}({#1})}
\newcommand{\thread}[1]{\mathtt{th}({#1})}
\newcommand{\variable}[1]{\mathtt{var}({#1})}
\newcommand{\valEvent}[1]{\mathtt{val}({#1})}
\newcommand{\eqdef}{::=}
\newcommand{\tab}{\mathit{tab}}
\newcommand{\pkey}{\mathit{pkey}}
\newcommand{\pkeyVal}{\mathit{pkeyVal}}
\newcommand{\ibegin}[1]{\ibegine(#1)}
\newcommand{\ibegine}{\mathtt{begin}}
\newcommand{\iadd}{\mathtt{add}}
\newcommand{\iremove}{\mathtt{remove}}
\newcommand{\ielements}{\mathtt{elements}}
\newcommand{\icontains}{\mathtt{has}}
\newcommand{\iend}{\mathtt{commit}}
\newcommand{\icommit}{\mathtt{commit}}
\newcommand{\iabort}{\mathtt{abort}}
\newcommand{\iif}[2]{\mathtt{if}({#1}) \{ {#2} \}}
\newcommand{\iread}{\mathtt{read}}
\newcommand{\iwrite}{\mathtt{write}}
\newcommand{\iselect}[1]{\mathtt{SELECT}({#1})}
\newcommand{\iinsert}[1]{\mathtt{INSERT}({#1})}
\newcommand{\idelete}[1]{\mathtt{DELETE}({#1})}
\newcommand{\iupdate}[2]{\mathtt{UPDATE}({#1},\ {#2})}
\newcommand{\eselect}{\mathtt{SELECT}}
\newcommand{\einsert}{\mathtt{INSERT}}
\newcommand{\edelete}{\mathtt{DELETE}}
\newcommand{\eupdate}{\mathtt{UPDATE}}
\newcommand{\eend}{\mathtt{end}}
\newcommand{\KVProgs}{\mathcal{P}_{KV}}
\newcommand{\SQLProgs}{\mathcal{P}_{SQL}}
\newcommand{\DBschema}{\mathcal{S}}
\newcommand{\DBinst}{\mathcal{D}}

\newcommand{\exec}{\xi}

%% file: commands/utility.tex
\newbool{testCondition}
\setbool{testCondition}{false}

\ifbool{testCondition}{%
\definecolor{newVersionColor}{HTML}{009999}
\definecolor{appendixVersionColor}{HTML}{e600ac}
\definecolor{noAppendixVersionColor}{HTML}{6600ff}
}{
\colorlet{newVersionColor}{black}
\colorlet{appendixVersionColor}{black}
\colorlet{noAppendixVersionColor}{black}
}%

\newcommand{\nver}[1]{\textcolor{newVersionColor}{#1}}
\newcommand{\red}[1]{\textcolor{red}{#1}}
\definecolor{orcidlogocol}{HTML}{A6CE39}

\newenvironment{sketchproof}{\renewcommand{\proofname}{Sketch of a Proof}\proof}{\qed\endproof}

\newbool{appendixVerCondition}
\setbool{appendixVerCondition}{false}

\NewEnviron{appendixVer}{%
\ifbool{appendixVerCondition}{%
    \ignorespaces%
    \leavevmode%
    \unskip%
    \color{appendixVersionColor}%
    \BODY%
    \color{black}%
    \ignorespacesafterend%
}{}%
}

\NewEnviron{noappendixVer}{%
\ifbool{appendixVerCondition}{}{%
    \ignorespaces%
    \leavevmode%
    \unskip%
    \color{noAppendixVersionColor}%
    \BODY%
    \color{black}%
    \ignorespacesafterend%
}%
}

%% file: sections/introduction.tex
\vspace{-5mm}
\section{Introduction}

Concurrent accesses to databases are typically grouped in transactions which define units of work that should be isolated from other concurrent computations and resilient to failures. Modern databases provide different levels of isolation for transactions with different trade-offs between consistency and throughput. The strongest isolation level, \emph{Serializability}~\cite{DBLP:journals/jacm/Papadimitriou79b}, provides the illusion that transactions are executed atomically one after another in a serial order. Serializability incurs a high cost in throughput. For performance, databases provide weaker isolation levels, e.g., 
{\it Snapshot Isolation}~\cite{DBLP:conf/sigmod/BerensonBGMOO95} or {\it Read
Committed}~\cite{DBLP:conf/sigmod/BerensonBGMOO95}.

The concurrency control protocols used in large-scale databases to implement isolation levels are difficult to build and test. For instance, the black-box testing framework Jepsen~\cite{jepsen} found a remarkably large number of subtle problems in many production databases. %

In this work, we focus on testing the isolation level implementations in databases, and more precisely, on the problem of checking whether a given execution adheres to the prescribed isolation level semantics. 
Inspired by scenarios that arise in commercial software~\cite{DBLP:conf/sigmod/Pavlo17}, %
we consider a quite generic version of the problem where transactions are formed of SQL queries and \emph{multiple} isolation levels are used at the same time, i.e., each transaction is assigned a possibly different isolation level (the survey in~\cite{DBLP:conf/sigmod/Pavlo17} found that 32\% of the respondents use such ``heterogeneous'' configurations). Previous work~\cite{DBLP:journals/jacm/Papadimitriou79b,DBLP:journals/pacmpl/BiswasE19} studied the complexity of the problem when transactions are formed of reads and writes on a \emph{static} set of keys (variables), and all transactions have the \emph{same} isolation level. 

As a first contribution, we introduce a formal semantics for executions with SQL transactions and a range of isolation levels, including serializability, snapshot isolation, prefix consistency, and read committed. Dealing with SQL queries is more challenging than classic reads and writes of a \emph{static} set of keys (as assumed in previous formalizations~\cite{DBLP:conf/concur/Cerone0G15,DBLP:journals/pacmpl/BiswasE19}). SQL insert and delete queries change the set of locations at runtime and the set of locations returned by an SQL query depends on their values (the values are restricted to satisfy $\whereName$ clauses). 

We define an abstract model for executions, called \emph{history}, where every SQL query that inspects the database (has a $\whereName$ clause) is associated with a set of SQL queries that wrote the inspected values. This relation is called a \emph{write-read} relation (also known as read-from). This is similar to associating reads to writes in defining memory models. We consider two classes of histories depending on the ``completeness'' of the write-read relation. To define a formal semantics of isolation levels, we need a complete write-read relation in the sense that for instance, an SQL select is associated with a write \emph{for every} possible row (identified by its primary key) in the database, even if that row is \emph{not} returned by the select because it does not satisfy the $\whereName$ clause. Not returning a row is an observable effect that needs to be justified by the semantics. Such \emph{full} histories can not be constructed by interacting with the database in a black-box manner (a desirable condition in testing) when only the outputs returned by queries can be observed. Therefore, we introduce the class of \emph{client} histories where the write-read concerns only rows that are \emph{returned} by a query. The consistency of a client history is defined as the existence of an extension of the write-read to a full history which satisfies the semantics.
The semantics on full histories combines axioms from previous work~\cite{DBLP:journals/pacmpl/BiswasE19} in a way that is directed by SQL queries that inspect the database and the isolation level of the transaction they belong to. This axiomatic semantics is validated by showing that it is satisfied by a standard operational semantics inspired by real implementations.

We study the complexity of checking if a full or client history is consistent, it satisfies the prescribed isolation levels. This problem is more complex for client histories, which record less dependencies and need to be extended to full ones. 

For full histories, we show that the complexity of consistency checking matches previous results in the reads and writes model when all transactions have the same isolation level~\cite{DBLP:journals/pacmpl/BiswasE19}: polynomial time for the so-called saturable isolation levels,
and NP-complete for stronger levels like Snapshot Isolation or Serializability. The former is a new result 
that generalizes the work of~\cite{DBLP:journals/pacmpl/BiswasE19} and exposes the key ideas for achieving polynomial-time complexity, %
while the latter is a consequence of the previous results.

We show that consistency checking becomes NP-complete for client histories even for saturable isolation levels. It remains NP-complete regardless of the expressiveness of $\whereName$ clauses (for this stronger result we define another class of histories called \emph{partial-observation}). The problem is NP-complete even if we bound the number of sessions. In general, transactions are organized in \emph{sessions}~\cite{DBLP:conf/pdis/TerryDPSTW94}, an abstraction of the sequence of transactions performed during the execution of an application (the counterpart of threads in shared memory). This case is interesting because it is polynomial-time in the read/write model~\cite{DBLP:journals/pacmpl/BiswasE19}.

As a counterpart to these negative results, we introduce an algorithm for checking consistency of client histories which is exponential-time in the worst case, but polynomial time in relevant cases. Given a client history as input, this algorithm combines an enumeration of extensions towards a full history with a search for a total commit order that satisfies the required axioms. The commit order represents the order in which transactions are committed in the database and it is an essential artifact for defining isolation levels. For efficiency, the algorithm uses a non-trivial enumeration of extensions that are \emph{not} necessarily full but contain enough information to validate consistency. The search for a commit order is a non-trivial generalization of an algorithm by Biswas et al.~\cite{DBLP:journals/pacmpl/BiswasE19} which concerned only serializability. This generalization applies to all practical isolation levels and combinations thereof. We evaluate an implementation of this algorithm on histories generated by PostgreSQL with a number of applications from BenchBase~\cite{difallah2013oltp}, e.g., the TPC-C model of a store and a model of Twitter. This evaluation shows that the algorithm is quite efficient in practice and scales well to typical workloads used in testing databases.

To summarize, we provide the first results concerning the complexity of checking the correctness of mixed isolation level implementations for SQL transactions. We introduce a formal specification for such implementations, and a first tool that can be used in testing their correctness.

%% file: sections/histories.tex
\vspace{-1mm}
\section{Histories}
\label{sec:histories}

\vspace{-1mm}
\subsection{Transactions}\label{ssec:syntax}

We model the database as a set of rows from an unbounded domain $\Vals$. Each row is associated to a unique (primary) key from a domain $\Keys$, given by the function $\mathsf{key}:\Vals\to\Keys$. We consider client programs accessing the database from a number of parallel sessions, each session being a sequence of transactions defined by the following grammar:

\begin{footnotesize}
\begin{center}
$\iota \in \Iso\quad a\in \LVars \quad \setRows \in 2^\Rows \quad \predicate \in \Vals \to \{0,1\} \quad \mapRows \in \Keys \to \Rows$ \\[-6mm]
\end{center}
\vspace{-5mm}
\begin{align*}
\mathsf{Transaction} & \eqdef  \ibegin{\iota}; \mathsf{Body}; \icommit\\
\mathsf{Body} & \eqdef  \mathsf{Instr} \ \mid\  \mathsf{Instr}; \mathsf{Body} \\
\mathsf{Instr} & \eqdef  \mathsf{InstrDB} \ \mid\  a := \mathsf{LExpr}  \mid\ \iif{ \mathsf{LCond} }{\mathsf{Instr}} \\
\mathsf{InstrDB} & \eqdef a := \iselect{\predicate}  \ \mid\ \iinsert{\setRows} \ \mid\ \idelete{\predicate} \ \mid \ \iupdate{\predicate}{\mapRows}\ \mid \iabort \\[-5mm]
\end{align*}
\end{footnotesize}

Each transaction is delimited by $\ebegin$ and $\icommit$ instructions. The $\ebegin$ instruction defines an isolation level $\iota$ for the current transaction. The set of isolation levels $\Iso$ we consider in this work will be defined later. 
The body contains standard SQL-like statements for accessing the database and standard assignments and conditionals for local computation. Local computation uses (transaction-)local variables from a set $\LVars$. We use $a$, $b$, $\ldots$ to denote local variables. Expressions and Boolean conditions over local variables are denoted with $\mathsf{LExpr}$ and $\mathsf{LCond}$, respectively. 

Concerning database accesses (sometimes called queries), we consider a simplified but representative subset of SQL: $\iselect{\predicate}$ returns the set of rows satisfying the predicate $\predicate$ and the result is stored in a local variable $a$. $\iinsert{\setRows}$ inserts the set of rows $\setRows$ or updates them in case they already exist
(this corresponds to $\mathtt{INSERT}\ \mathtt{ON}$ $\mathtt{CONFLICT}$ $\mathtt{DO}$ $\mathtt{UPDATE}$ in PostgreSQL)
, and $\idelete{\predicate}$ deletes all the rows that satisfy $\predicate$. Then, $\iupdate{\predicate}{\mapRows}$ updates the rows satisfying $\predicate$ with values given by the map $\mapRows$, i.e., every row $\row$ in the database that satisfies $\predicate$ is replaced with $\mapRows(\keyof{\row})$, and $\iabort$ aborts the current transaction. The predicate $\predicate$ corresponds to a $\whereName$ clause in standard SQL.

\vspace{-1mm}
\subsection{Histories}

\vspace{-1mm}

We define a model of the interaction between a program and a database called \emph{history} which abstracts away the local computation in the program and the internal behavior of the database. A history is a set of \emph{events} representing the database accesses in the execution grouped by transaction, along with some relations between these events which explain the output of $\eselect$ instructions. 

An event is a tuple $\langle e, \mathit{type} \rangle$ where $e$ is an \textit{identifier} and $\mathit{type}$ is one of $\ebegin$, $\ecommit$, $\eabort$, $\eselect$, $\einsert$, $\edelete$ and $\eupdate$. $\Events$ denotes the set of events.
For an event $e$ of type $\eselect$, $\edelete$, or $\eupdate$, we use $\where{e}$ to denote the predicate $\predicate$
and for an $\eupdate$ event $e$, we use $\setUpdate{e}$ to denote the map $\mapRows$. 

We call $\iread$ events the $\eselect$ events that read the database to return a set of rows, and the $\edelete$ and $\eupdate$ events that read the database checking satisfaction of some predicate $\predicate$. Similarly, we call $\iwrite$ events the $\einsert$, $\edelete$ and $\eupdate$ events that modify the database. We also say that an event is of type $\eend$ if it is either a $\ecommit$ or an $\eabort$ event.

A \emph{transaction log} $\tup{t, \iota_t, E, \po_t}$ is an identifier $t$, an \emph{isolation level} identifier $\iota_t$, and a finite set of events $E$ along with a strict total order $\po_t$ on $E$, called \emph{program order} (representing the order between instructions in the body of a transaction). The set $E$ of events in a transaction log $t$ is denoted by $\events{t}$. For simplicity, we may use the term \emph{transaction} instead of transaction log.

Isolation levels differ in the values returned by read events which are not preceded by a write on the same variable in the same transaction. %
We denote by $\readOp{t}$ the set of $\iread$ events contained in $t$. Also, if $t$ does \textit{not} contain an $\eabort$ event, the set of $\iwrite$ events in $t$ is denoted by $\writeOp{t}$. If $t$ contains an $\eabort$ event, then we define $\writeOp{t}$ to be empty. This is because the effect of aborted transactions (its set of writes) should not be visible to other transactions. The extension to sets of transaction logs is defined as usual. 

To simplify the exposition we assume that for any given key $x\in\Keys$, a transaction does not modify (insert/delete/update) a row with key $x$ more than once. Otherwise, under all isolation levels, only the last among multiple updates is observable in other transactions.

As expected, we assume that the minimal element of $\po_t$ is a $\ebegin$ event, if a $\ecommit$ or an $\eabort$ event occurs, then it is maximal in $\po_t$, and a log cannot contain both $\ecommit$ and $\eabort$.
A transaction log without $\ecommit$ or $\eabort$ is called \emph{pending}. Otherwise, it is \emph{complete}. A complete transaction log with a $\ecommit$ is \textit{committed} and \textit{aborted} otherwise.

A \emph{history} contains a set of transaction logs (with distinct identifiers)
ordered by a (partial) \emph{session order} $\so$ that represents the order
between transactions in the same session. %
It also includes a
\emph{write-read} relation $\wro$ which associates $\iwrite$ events with $\iread$ events. The $\iwrite$ events associated to a $\iread$ implicitly define the values observed (returned) by the $\iread$ (read events do \emph{not} include explicit values).
Let $T$ be a set of transaction logs.
For every key $\key \in \Vars$ we consider a write-read relation $\wro_{\key}\subseteq \writeOp{T}\times \readOp{T}$. The union of $\wro_{\key}$ for every $\key \in \Vars$ is denoted by $\wro$. We extend the relations $\wro$ and $\wro_\key$ to pairs of transactions by $\tup{t_1,t_2}\in \wro$, resp., $\tup{t_1,t_2}\in \wro_\key$, iff there exist events $w$ in $t_1$ and $r$ in $t_2, t_2 \neq t_1$ s.t. $\tup{w,r}\in \wro$, resp., $\tup{w, r}\in \wro_\key$. Analogously, we extend $\wro$ and $\wro_\key$ to tuples formed of a transaction (containing a write) and a read event. We say that the transaction $t_1$ is \emph{read} by the transaction $t_2$ when $\tup{t_1,t_2}\in \wro$. The inverse of $\wro_{\key}$ is defined as usual and denoted by $\wro_{\key}^{-1}$. We assume that $\wro_{\key}^{-1}$ is a partial function and thus, use $\wro_{\key}^{-1}(e)$ to denote the $\iwrite$ event $w$ such that $(w,e)\in \wro_{\key}$. We also use $\wro_{\key}^{-1}(e) \downarrow$ and $\wro_{\key}^{-1}(e) \uparrow$ to say that there exists a $\iwrite$ $w$ such that $(w,e)\in \wro_{\key}$ (resp. such $\iwrite$ $w$ does not exist). 

To simplify the exposition, every history includes a distinguished transaction $\init$ preceding all the other transactions in $\so$ and inserting a row for every $\key$. It represents the initial state and it is the only transaction that may insert as value $\del_\key$ (indicating that initially, no row with key $x$ is present). 

\vspace{-1mm}
\begin{definition}
\label{def:history}
A \emph{history} $\tup{T, \so, \wro}$ is a set of transaction logs $T$ along with a strict partial \emph{session order} $\so$, and a 
\emph{write-read} relation $\wro_\key\subseteq \writeOp{T}\times \readOp{T}$ for each $\key \in \Vars$ s.t.
\vspace{-1mm}
\begin{itemize}
	\item the inverse of $\wro_{\key}$ is a partial function, 
	\item $\so\cup\wro$ is acyclic (here we use the extension of $\wro$ to pairs of transactions), 
	\item if $(w, r) \in \wro_\key$, then $\valuewr{w}{\key} \neq \bot$, where

\begin{equation*}
\sloppy
\valuewr{w}{x} = \left\{\begin{array}{ll}
	\row & \text{if } w = \iinsert{\setRows} \land \row \in \setRows \land \keyof{\row}=x  \\
	\del_{\key} & \text{if } w = \idelete{\predicate} \land\ \wro_{\key}^{-1}(w) \downarrow \\ 
	&\hspace{1cm} \land\ \predicate(\valuewr{\wro_{\key}^{-1}(w) }{x}) = 1 \\[1mm]
	\mapRows(\key) & \text{if } w = \iupdate{\predicate}{\mapRows} \land \ \wro_{\key}^{-1}(w) \downarrow \\ 
	&\hspace{1cm} \land\ \predicate(\valuewr{\wro_{\key}^{-1}(w) }{x}) = 1 \\ %
	\bot & \text{otherwise} 
	\end{array}\right.
\end{equation*}
	$ $
\end{itemize}
\vspace{-4mm}
\end{definition}

\input{figures-tex/example-history}

The function $\wro_x^{-1}$ may be partial because some query may not read a key $x$, e.g., if the corresponding row does not satisfy the query predicate.

The function $\valuewr{w}{x}$ returns the row with key $x$ written by the $\iwrite$ event $w$. If $w$ is an $\einsert$, it returns the inserted row with key $x$. If $w$ is an $\iupdate{\predicate}{\mapRows}$ event, it returns the value of $\mapRows$ on key $x$ if $w$ reads a value for key $x$ that satisfies predicate $\predicate$. If $w$ is a $\idelete{p}$, it returns the special value $\del_{\key}$ if $w$ reads a value for key $x$ that satisfies $\predicate$. This special value indicates that the database does \emph{not} contain a row with key $\key$. In case no condition is satisfied, $\valuewr{w}{x}$ returns an undefined value $\bot$. We assume that the special values $\del_x$ or $\bot$ do not satisfy any predicate. 
Note that the recursion in the definition of $\valuewr{w}{\key}$ terminates because $\wro$ is an acyclic relation.

Figure~\ref{fig:example-history} shows an example of a history. For the $\eupdate$ event $w$ in $t_1$, $\valuewr{w}{x_1}=\bot$ because this event reads $x_1$ from the $\edelete$ event in $t_2$; while $\valuewr{w}{x_2}=-2$ as it reads $x_2$ from the $\einsert$ event in $\init$.

The set of transaction logs $T$ in a history $\hist=\tup{T, \so, \wro}$ is denoted by $\tlogs{\hist}$ and $\events{\hist}$ is the union of $\events{t}$ for every $t\in T$. For a history $\hist$ and an event $e$ in $\hist$, $\trans{e}$ is the transaction $t$ in $\hist$ that contains $e$. We assume that each event belongs to only one transaction. Also, $\writeOp{\hist}=\bigcup_{t\in \tlogs{h}}\writeOp{t}$ and $\readOp{\hist}=\bigcup_{t\in \tlogs{h}}\readOp{t}$.
We extend $\so$ to pairs of events by $(e_1,e_2)\in \so$ if $(\trans{e_1},\trans{e_2})\in\so$. Also, $\po=\bigcup_{t\in T} \po_t$.
We use $\hist$, $\hist_1$, $\hist_2$, $\ldots$ to range over histories. 

For a history $h$, we say that an event $r$ \emph{reads} $x$ in $h$
whenever $\wro_x^{-1}(r) \downarrow$. Also, we say that an event $w$ \emph{writes} $x$ in $h$, denoted by $\writeVar{w}{x}$, whenever $\valuewr{w}{x} \neq \bot$ and the transaction of $w$ is \emph{not} aborted. We extend the function $\mathtt{value}$ to transactions: $\valuewr{t}{x}$ equals $\valuewr{w}{x}$, where $w$ is the maximal event in $\po_t$ that writes $x$.

\vspace{-1mm}
\subsection{Classes of histories}

\vspace{-1mm}
We define two classes of histories: (1) \emph{full} histories which are required to define the semantics of isolation levels and (2) \emph{client} histories which model what is observable from interacting with a database as a black-box.

Full histories model the fact that every read query ``inspects'' an entire snapshot of the database in order to for instance, select rows satisfying some predicate. Roughly, full histories contain a write-read dependency for every read and key. There is an exception which concerns ``local'' reads. If a transaction modifies a row with key $x$ and then reads the same row, then it must always return the value written in the transaction. This holds under all isolation levels. In such a case, there would be no write-read dependency because these dependencies model interference across different transactions.
We say that a read $r$ reads a key $x$ \emph{locally} if it is preceded in the same transaction by a write $w$ that writes $x$.

\vspace{-1mm}
\begin{definition}
\label{def:full-history}
A \emph{full history} $\tup{T, \so, \wro}$ is a history where $\wro_\key^{-1}(r)$ is defined for all $\key$ and $r$, unless $r$ reads $x$ locally.
\vspace{-1mm}
\end{definition}

\input{figures-tex/example-full-client-consistent}

Client histories record less write-read dependencies compared to full histories, which is formalized by the \emph{extends} relation.

\vspace{-1mm}
\begin{definition}
	A history $\overline{h} =\tup{T, \so, \overline{\wro}}$ \emph{extends} another history $h =\tup{T, \so, \wro}$ if $\wro \subseteq \overline{\wro}$. We denote it by $h \subseteq \overline{h}$.
\vspace{-1mm}
\end{definition}

\vspace{-1mm}
\begin{definition}
\label{def:client-history}
A \emph{client history} $h =\tup{T, \so, \wro}$ is a history s.t. there is a full history $\overline{h} =\tup{T, \so, \overline{\wro}}$ with $h \subseteq \overline{h}$, and s.t for every $\key$, if $(w,r) \in \overline{\wro}_\key \setminus \wro_\key$ then $\where{r}(\valuewr[\overline{\wro}]{w}{\key}) = 0$. The history $h'$ is called a \emph{witness} of $h$.
\vspace{-1mm}
\end{definition}

Compared to a witness full history, a client history may omit write-read dependencies if the written values do \emph{not} satisfy the predicate of the read query. These values would not be observable when interacting with the database as a black-box. This includes the case when the write is a $\edelete$ (recall that the special value $\del_{\key}$ indicating deleted rows falsifies every predicate by convention). \Cref{fig:example-history} shows a full history as every query reads both $x_1$ and $x_2$. \Cref{fig:example-full-client:client} shows a client history: transactions $t_1, t_2$ does not read $x_2$ and $x_1$ resp. \Cref{fig:example-full-client:no-extension} is an extension but not a witness while \Cref{fig:example-full-client:witness} is indeed a witness of it.

%% file: figures-tex/example-history.tex
\begin{figure}[t]
    \centering
            
    \resizebox{0.9\textwidth}{!}{
        \begin{tikzpicture}[>=stealth',shorten >=1pt,auto,node distance=3cm,
            semithick, transform shape,
            B/.style = {%
            decoration={brace, amplitude=1mm,#1},%
            decorate},
            B/.default = ,  %
            ]
            \node[minimum width=8em, draw, rounded corners=2mm,outer sep=0, label={[font=\small]170:$t_1$}] (r11) at (0, -1.25) {
                \begin{tabular}{l}
                   $\iupdate{\lambda r: r \geq 1}{\lambda r: -2}$
                \end{tabular}
            };

            \node[minimum width=8em, draw, rounded corners=2mm,outer sep=0, label={[font=\small]10:$\init$}] (init) at (4, 0) {$\iinsert{\{x_1 : 0, x_2 : 1\}}$};
            \node[minimum width=8em, draw, rounded corners=2mm,outer sep=0, label={[font=\small]10:$t_2$}] (r12) at (8, -1.25) {
                \begin{tabular}{l}
                $\idelete{\lambda r: r \leq 0}$
                \end{tabular}
            };

            \path (init) edge[->, soColor, right, transform canvas={xshift=-6mm}] node [above, left, shift={(0,0.1)}] {$\so$} (r11);
    
            \path (init) edge[->, transform canvas={xshift=0mm}, wrColor] node [black, left] {$\wro_{x_1, x_2}$} (r12);

            \path (init) edge[->, soColor, right, transform canvas={xshift=6mm}] node [above, right, shift={(0,0.1)}] {$\so$} (r12);
    
            \path (init) edge[->,  transform canvas={xshift=0mm}, wrColor] node [black, right, shift={(0.3,0)}] {$\wro_{x_2}$} (r11);

            \path (r12) edge[->, wrColor, transform canvas={yshift=-1.5},] node [black, below] {$\wro_{x_1}$} (r11);

        \end{tikzpicture}  
    }
    
    \vspace{-2mm}
    
    \caption{An example of a history (isolation levels omitted for legibility). Arrows represent $\so$ and $\wro$ relations. Transaction $\init$ defines the initial state: row $0$ with key $x_1$ and row $1$ with key $x_2$. Transaction $t_2$ reads $x_1$ and $x_2$ from $\init$ and deletes row with key $x_1$ (the only row satisfying predicate $\lambda r: r \leq 0$ corresponds to key $x_1$). Transaction $t_1$ reads $x_1$ from $t_2$ and $x_2$ from $\init$, and updates only row with key $x_2$ as this is the only row satisfying predicate $\lambda r: r \geq 1$.}
    \label{fig:example-history}
\vspace{-4mm}
\end{figure}
    

%% file: figures-tex/example-full-client-consistent.tex
\begin{figure}[t]
\centering
\begin{subfigure}{0.5\textwidth}
    
\resizebox{\textwidth}{!}{
    \begin{tikzpicture}[>=stealth',shorten >=1pt,auto,node distance=3cm,
        semithick, transform shape,
        B/.style = {%
        decoration={brace, amplitude=1mm,#1},%
        decorate},
        B/.default = ,  %
        ]
        \node[minimum width=8em, draw, rounded corners=2mm,outer sep=0, label={[font=\small]170:$t_1$}] (r11) at (0, -2) {
            \begin{tabular}{l}
                $\iupdate{\lambda r: r \geq 1}{\lambda r: -2}$
            \end{tabular}
        };

        \node[minimum width=8em, draw, rounded corners=2mm,outer sep=0, label={[font=\small]10:$\init$}] (init) at (2.2, -0.5) {$\iinsert{\{x_1 : 0, x_2 : 1\}}$};
        \node[minimum width=8em, draw, rounded corners=2mm,outer sep=0, label={[font=\small]10:$t_2$}] (r12) at (4.4, -2) {
            \begin{tabular}{l}
            $\idelete{\lambda r: r \leq 0}$
            \end{tabular}
        };

        \path (init) edge[->, soColor, right, transform canvas={xshift=-6mm}] node [above, right] {$\so$} (r11);

        \path (init) edge[->, transform canvas={xshift=0mm}, wrColor] node [black, left] {$\wro_{x_1}$} (r12);

        \path (init) edge[->, soColor, right, transform canvas={xshift=6mm}] node [above] {$\so$} (r12);

        \path (init) edge[->,  transform canvas={xshift=0mm}, wrColor] node [black, right] {$\wro_{x_2}$} (r11);

    \end{tikzpicture}  
}

\caption{Client history.}
\label{fig:example-full-client:client}
\end{subfigure}
\hfill
\begin{subfigure}{0.49\textwidth}
        
\resizebox{\textwidth}{!}{
    \begin{tikzpicture}[>=stealth',shorten >=1pt,auto,node distance=3cm,
        semithick, transform shape,
        B/.style = {%
        decoration={brace, amplitude=1mm,#1},%
        decorate},
        B/.default = ,  %
        ]
        \node[minimum width=8em, draw, rounded corners=2mm,outer sep=0, label={[font=\small]170:$t_1$}] (r11) at (0, -2) {
            \begin{tabular}{l}
                $\iupdate{\lambda r: r \geq 1}{\lambda r: -2}$
            \end{tabular}
        };

        \node[minimum width=8em, draw, rounded corners=2mm,outer sep=0, label={[font=\small]10:$\init$}] (init) at (2.2, -0.5) {$\iinsert{\{x_1 : 0, x_2 : 1\}}$};
        \node[minimum width=8em, draw, rounded corners=2mm,outer sep=0, label={[font=\small]10:$t_2$}] (r12) at (4.4, -2) {
            \begin{tabular}{l}
            $\idelete{\lambda r: r \leq 0}$
            \end{tabular}
        };

        \path (init) edge[->,  dashed, transform canvas={xshift=-9mm}, wrColor] node [black, left] {$\wro_{x_1}$} (r11);

        \path (init) edge[->, soColor, right, transform canvas={xshift=-6mm}] node [above, right] {$\so$} (r11);

        \path (init) edge[->, transform canvas={xshift=0mm}, wrColor] node [black, left] {$\wro_{x_1}$} (r12);

        \path (init) edge[->, soColor, right, transform canvas={xshift=6mm}] node [above] {$\so$} (r12);

        \path (init) edge[->,  transform canvas={xshift=0mm}, wrColor] node [black, right] {$\wro_{x_2}$} (r11);

        \path (r11) edge[->, wrColor, dashed, transform canvas={yshift=-1.5},] node [black, above] {$\wro_{x_2}$} (r12);

    \end{tikzpicture}  
}

\caption{$t_2$ observes $x_2 = -2$.}
\label{fig:example-full-client:no-extension}

\end{subfigure}
\hfill
\begin{subfigure}{0.49\textwidth}
    
\resizebox{\textwidth}{!}{
    \begin{tikzpicture}[>=stealth',shorten >=1pt,auto,node distance=3cm,
        semithick, transform shape,
        B/.style = {%
        decoration={brace, amplitude=1mm,#1},%
        decorate},
        B/.default = ,  %
        ]
        \node[minimum width=8em, draw, rounded corners=2mm,outer sep=0, label={[font=\small]170:$t_1$}] (r11) at (0, -2) {
            \begin{tabular}{l}
                $\iupdate{\lambda r: r \geq 1}{\lambda r: -2}$
            \end{tabular}
        };

        \node[minimum width=8em, draw, rounded corners=2mm,outer sep=0, label={[font=\small]10:$\init$}] (init) at (2.2, -0.5) {$\iinsert{\{x_1 : 0, x_2 : 1\}}$};
        \node[minimum width=8em, draw, rounded corners=2mm,outer sep=0, label={[font=\small]10:$t_2$}] (r12) at (4.4, -2) {
            \begin{tabular}{l}
                $\idelete{\lambda r: r \leq 0}$
            \end{tabular}
        };
        \path (init) edge[->,  dashed, transform canvas={xshift=-9mm}, wrColor] node [black, left] {$\wro_{x_1}$} (r11);

        \path (init) edge[->, soColor, right, transform canvas={xshift=-6mm}] node [above, right] {$\so$} (r11);

        \path (init) edge[->,  transform canvas={xshift=0mm}, wrColor] node [black, right] {$\wro_{x_2}$} (r11);

        \path (init) edge[->, soColor, right, transform canvas={xshift=6mm}] node [above] {$\so$} (r12);

        \path (init) edge[->, transform canvas={xshift=0mm}, wrColor] node [black, left] {$\wro_{x_1}$} (r12);

        \path (init) edge[->, dashed, transform canvas={xshift=12mm}, wrColor] node [black, right] {$\wro_{x_2}$} (r12);

    \end{tikzpicture}  
}

\caption{$t_2$ observes $x_2 = 1$.}
\label{fig:example-full-client:witness}

\end{subfigure}
\vspace{-2mm}
\caption{Examples of a client history $h$ and two possible extensions. The dashed edge belongs only to the extensions. The first extension is not a witness of $h$ as $t_1$ writes $-2$ on $x_2$ and $\where{t_2}(-2) = 1$.
}
\label{fig:example-full-client}
\vspace{-4mm}
\end{figure}

%% file: sections/multi-isolation.tex
\section{Axiomatic Semantics With Different Isolation Levels}
\label{sec:multi-isolation}

We define an axiomatic semantics on histories where transactions can be assigned different isolation levels, which builds on the work of Biswas et al.~\cite{DBLP:journals/pacmpl/BiswasE19}.

\input{sections/subsec-multiiso/executions}

\subsection{Isolation Levels}

Isolation levels enforce restrictions on the commit order in an execution that depend on the session order $\so$ and the write-read relation $\wro$. An \emph{isolation level} $\iota$ for a transaction $t$
is a set of constraints called \emph{axioms}.  %
Intuitively, an axiom states that a read event $r \in t$ reads key $x$ from transaction $t_1$ if $t_1$ is the latest transaction that writes $x$ which is ``visible'' to $r$ -- latest refers to the commit order $\co$. Formally, an axiom $a$ is a predicate of the following form:

\vspace{-4mm}
\begin{equation}
\begin{array}{lllll}
  a(r) & \coloneqq & \forall \key, \tr_1, \tr_2. \tr_1 \neq \tr_2 \land \tup{\tr_1,r}\in \wro_\key \land \writeVar{\tr_2}{\key} \land \visibilityRelationAxiom{t_2}{r}{x} & \Rightarrow & \tup{\tr_2,\tr_1}\in\co 
\end{array}
\label{eq:axiom}
\end{equation}
\vspace{-4mm}

\noindent
where $r$ is a read event from $t$. 

The visibility relation of $a$ $\mathsf{vis}_a$ is described by a formula of the form:

\vspace{-5mm}
\begin{equation}
  \visibilityRelationAxiom{\tau_0}{\tau_{k+1}}{x} : \exists \tau_1, \ldots, \tau_{k}. \bigwedge_{i=1}^{k+1} (\tau_{i-1}, \tau_{i}) \in \mathsf{Rel}_i \land \confWr_{a}(\tau_0, \ldots, \tau_{k+1}, x) \label{eq:axioms-constraints}
  \end{equation}
with each $\mathsf{Rel}_i$ is defined by the grammar:
\begin{equation}
\mathsf{Rel} \eqdef \ \po\, |\, \so\, |\, \wro\, |\, \co \,|\,\mathsf{Rel} \cup \mathsf{Rel} \,|\,\mathsf{Rel} ; \mathsf{Rel} \,| \,\mathsf{Rel}^+ \,| \, \mathsf{Rel}^*  
\label{eq:axioms-constraints-rel}
\end{equation}

This formula states that $\tau_0$ (which is $\tr_2$ in Eq.\ref{eq:axiom}) is connected to $\tau_{k+1}$ (which is $r$ in Eq.\ref{eq:axiom}) by a path of dependencies that go through some intermediate transactions or events $\tau_1, \ldots, \tau_k$. Every relation used in such a path is described based on $\po,\so, \wro$ and $\co$ using union $\cup$, composition of relations $;$, 
and transitive closure operators. 
Finally, extra requirements on the intermediate transactions s.t. writing a different key $y\neq x$ are encapsulated in the predicate $\confWr_{a}(\tau_0, \ldots, \tau_k, x)$.

Each axiom $a$ uses a specific visibility relation denoted by $\visibility_a$. $\visibilitySet{\iota}$ denotes the set of visibility relations used in axioms defining an isolation level $\iota$.

\input{figures-tex/consistency-levels}

\begin{appendixVer}
  \Cref{fig:consistency_defs} shows two axioms which correspond to their homonymous isolation levels~\cite{DBLP:journals/pacmpl/BiswasE19}: \textit{Read Committed} ($\RC$) and \textit{Serializability} ($\SER$). $\SER$ states that $t_2$ is visible to $r$ if $t_2$ commits before $r$, while $\RC$ states that $t_2$ is visible to $r$ if either $(t_2, r) \in \so$ or if there exists a previous event $r'$ in $\trans{r}$ that reads $x$ from $t_2$. Similarly, \textit{Read Atomic} ($\RA$) and \textit{Prefix Consistency} ($\PRE$) are defined using their homonymous axioms while \textit{Snapshot Isolation} ($\SI$) is defined as a conjunction of both $\axpre$ and $\axconf$.
\end{appendixVer}

\begin{noappendixVer}
  \Cref{fig:consistency_defs} shows two axioms which correspond to their homonymous isolation levels~\cite{DBLP:journals/pacmpl/BiswasE19}: \textit{Read Committed} ($\RC$) and \textit{Serializability} ($\SER$). $\SER$ states that $t_2$ is visible to $r$ if $t_2$ commits before $r$, while $\RC$ states that $t_2$ is visible to $r$ if either $(t_2, r) \in \so$ or if there exists a previous event $r'$ in $\trans{r}$ that reads $x$ from $t_2$. Similarly, \textit{Read Atomic} ($\RA$) and \textit{Prefix Consistency} ($\PRE$) are defined using their homonymous axioms while \textit{Snapshot Isolation} ($\SI$) is defined as a conjunction of both $\axpre$ and $\axconf$.
\end{noappendixVer}

The \emph{isolation configuration} of a history is a mapping $\isolation{\hist}: T \to \Iso$ associating to each transaction an isolation level identifier from a set $\Iso$.

Whenever every transaction in a history has the same isolation level $\iota$, the isolation configuration of that history is denoted simply by $\iota$. 

Note that $\SER$ is stronger than $\RC$: every transaction visible to a read $r$ according to $\RC$ is also visible to $r$ according to $\SER$. This means $\SER$ imposes more constraints for transaction $t_1$ to be read by $r$ than $\RC$.
In general, for two isolation configurations $I_1$ and $I_2$, $I_1$ is \emph{stronger than} $I_2$ when for every transaction $t$, $I_1(t)$ is stronger than $I_2(t)$ 
(i.e., whenever $I_1(t)$ holds in an execution $\exec$, $I_2(t)$ also holds in $\exec$). %
The \emph{weaker than} relationship is defined similarly.

\label{eq:axioms}

Given a history $h$ with isolation configuration $\isolation{\hist}$, $h$ is called \emph{consistent} when there exists an execution $\exec$ of $h$ such that for all transactions $t$ in $\exec$, the axioms in $\isolation{\hist}(t)$ are satisfied in $\exec$ (the interpretation of an axiom over an execution is defined as expected).
For example, let $h$ be the full history in Figure~\ref{fig:example-full-client:witness}. If both $t_1, t_2$'s isolation are $\SER$, then $h$ is \emph{not} consistent, i.e., every execution $\exec = \tup{h, \co}$ violates the corresponding axioms. Assume for instance, that $(t_1, t_2) \in \co$. Then, by axiom $\SER$, as $(\init, t_2) \in \wro_{x_1}$ and $\writeVar{t_1}{x_1}$, we get that $(t_1, \init) \in \co$, which is impossible as $(\init,t_1)\in\so\subseteq \co$. However, if the isolation configuration is weaker (for example $\isolation{h}(t_2) = \RC$), then the history is consistent using $\init <_{\co} t_1 <_{\co} t_2$ as commit order.

\begin{definition}
\label{def:consistency-full}
A full history $\hist=\tup{T, \so, \wro}$ with isolation configuration $\isolation{\hist}$ is \emph{consistent} iff there is an execution $\exec$  of $h$ s.t. $\bigwedge_{t \in T, r \in \readOp{t}, a \in \isolation{h}(t)} a(r)$ holds in $\exec$; $\exec$ is called a \emph{consistent} execution of $h$.
\label{axiom-criterion}
\end{definition}

The notion of consistency on full histories is extended to client histories.
\begin{definition}
\label{def:consistency-non-full}
A client history $h = \tup{T, \so, \wro}$ with isolation configuration $\isolation{\hist}$ is \emph{consistent} iff there is a full history $\overline{h}$ with the same isolation configuration which is a witness of $h$ and consistent; $\overline{h}$ is called a \emph{consistent} witness of $h$.
\end{definition}

In general, the witness of a client history may not be consistent. In particular, there may exist several witnesses but no consistent witness.

\subsection{Validation of the semantics}
\label{ssec:operational-semantics}

To justify the axiomatic semantics defined above, we define an operational semantics inspired by real implementations and prove that every run of a program can be translated into a consistent history. 
Every instruction is associated with an increasing timestamp and it reads from a snapshot of the database defined according to the isolation level of the enclosing transaction. At the end of the transaction we evaluate if the transaction can be committed or not. We assume that a transaction can abort only if explicitly stated in the program. 
We model an optimistic approach where %
if a transaction cannot commit, the run blocks
(modelling unexpected aborts). %
We focus on three of the most used isolation levels: $\SER, \SI, \RC$. Other isolation levels can be handled in a similar manner. $ $%
\begin{appendixVer}
For each run $\rho$ we extract a full history $\historyExecution{\rho}$. We show by induction that $\historyExecution{\rho}$ is consistent at every step. The formal description of the semantics and its correctness can be found in \Cref{app:operational-semantics}.%
\end{appendixVer}
\begin{noappendixVer}
For each run $\rho$ we extract a full history $\historyExecution{\rho}$. We show by induction that $\historyExecution{\rho}$ is consistent at every step. %
\end{noappendixVer}

\begin{restatable}{theorem}{semantics}
\label{th:operational-semantics-full-and-consistent}
For every run $\rho$, $\historyExecution{\rho}$ is consistent.
\end{restatable}

%% file: sections/subsec-multiiso/executions.tex
\subsection{Executions}
\label{sec:executions}

An \emph{execution} of a program is represented using a history with a set of transactions $T$ along with a total order $\co \subseteq T \times T$ called \emph{commit order}. Intuitively, the commit order represents the order in which transactions are committed in the database.

\begin{definition}
\label{def:execution}
An \emph{execution} $\exec=\tup{h,\co}$ is a history $h = \tup{T, \so, \wro}$ along with a \emph{commit order} $\co \subseteq T \times T$, such that transactions in the same session or that are read are necessarily committed in the same order: $\so \cup \wro \subseteq \co$. $\exec$ is called an execution of $h$.
\end{definition}
For a transaction $t$, we use $t\in \exec$ to denote the fact that $t\in T$. Analogously, for an event $e$, we use $e \in \exec$ to denote that $e \in t$ and $t \in \exec$. The extension of a commit order to pairs of events or pairs of transactions and events is done in the obvious way.

%% file: figures-tex/consistency-levels.tex
\begin{figure*}[tp]
	\begin{center}
	\resizebox{\textwidth}{!}{
		\footnotesize
	\begin{subfigure}{\textwidth}
		
		\begin{tabularx}{\textwidth}{|*{3}{X|}}
			\hline & \\ [-3.5mm]
			\begin{subfigure}[b]{.33\textwidth}
				\centering
				\resizebox{.7\textwidth}{!}{
				\begin{tikzpicture}[->,>=stealth,shorten >=1pt,auto,node distance=1cm,
					semithick, transform shape]
					\node[transaction state, text=black] at (0,0)       	(t_1)           {$t_1$};
					\node[transaction state, text=black, label={above:\textcolor{black}{$\writeVar{ }{x}$}}] at (-0.5,1.5) (t_2) {$t_2$};
					\node[transaction state, text=black] at (2,0)       (o_1)           {$r$};
					\node[transaction state] at (1.5,1.5) (o_2) {$r'$};
					\path (t_1) edge[color=wrColor] node [black] {$\wro_x$} (o_1);
					\path (t_2) edge[] node {$\so \cup \wro$} (o_2);
					\path (o_2) edge[color=poColor] node {$\po^*$} (o_1);
					\path (t_2) edge[left,double equal sign distance,color=coColor] node {$\co$} (t_1);
				\end{tikzpicture}
				}
				\parbox{\textwidth}{
					$\forall x \forall t_1 \forall t_2.\ t_1\neq t_2\ \land $
					
					\hspace{8mm}$\tup{t_1,r}\in \wro_x \ \land$ 

					\hspace{8mm}$ \writeVar{t_2}{x} \ \land$ 

					\hspace{8mm}$ \tup{t_2,r}\in (\so \cup \wro);\po^*$ 
										
					\hspace{8mm} \quad $\implies \tup{t_2,t_1}\in\co$
				}
				\vspace{-1.5mm}

				\caption{$\mathsf{Read\ Committed}$}
				\label{fig:rc_def}
			\end{subfigure}   
		
			&     

			\begin{subfigure}[b]{.33\textwidth}
				\centering
				\resizebox{.6\textwidth}{!}{

				\begin{tikzpicture}[->,>=stealth,shorten >=1pt,auto,node distance=1cm,
					semithick, transform shape]
					\node[transaction state, text=black] at (0,0)       (t_1)           {$t_1$};
					\node[transaction state] at (2,0)       (r)           {$r$};
					\node[transaction state, text=black,label={above:\textcolor{black}{$\writeVar{ }{\xvar}$}}] at (-.5,1.5) (t_2) {$t_2$};
					\path (t_1) edge[wrColor] node {$\wro_x$} (r);
					\path (t_2) edge[bend left] node {$\so \cup \wro$} (r);
					\path (t_2) edge[left, double ,coColor] node {$\co$} (t_1);
				\end{tikzpicture}
				}
				\parbox{\textwidth}{
					$\forall x \forall t_1 \forall t_2.\ t_1\neq t_2\ \land$
					
					\hspace{8mm}$\tup{t_1,\trans{r}}\in \wro_x \ \land$ 
	
					\hspace{8mm}$ \writeVar{t_2}{x}\ \land$ 
					
					\hspace{8mm}$\tup{t_2,\trans{r}}\in\so \cup \wro$
					
					\hspace{8mm} \quad $\implies \tup{t_2,t_1}\in\co$
				}
				
				\caption{$\mathsf{Read\ Atomic}$}
				\label{fig:ra_def}
			\end{subfigure}

			&
			
			\begin{subfigure}[b]{.33\textwidth}
				\centering
				\resizebox{.65\textwidth}{!}{
				\begin{tikzpicture}[->,>=stealth,shorten >=1pt,auto,node distance=4cm,
					semithick, transform shape]
					\node[transaction state, text=black] at (0,0)       (t_1)           {$t_1$};
					\node[transaction state] at (2,0)       (t_3)           {$r$};
					\node[transaction state, text=black, label={above:\textcolor{black}{$\writeVar{ }{x}$}}] at (-.5,1.5) (t_2) {$t_2$};
					\path (t_1) edge[wrColor] node [black] {$\wro_x$} (t_3);
					\path (t_2) edge[bend left, double equal sign distance, coColor] node {$\co$} (t_3);
					\path (t_2) edge[left,double equal sign distance, coColor] node {$\co$} (t_1);
				\end{tikzpicture}
				}
				\parbox{\textwidth}{

					$\forall x \forall t_1  \forall t_2.$ $ t_1\neq t_2\ \land$ 

					\hspace{8mm}$\tup{t_1,r}\in \wro_x \ \land$ 

					\hspace{8mm}$ \writeVar{t_2}{x}\ \land$ 

					\hspace{8mm}$\tup{t_2,r}\in\co$ 
										
					\hspace{8mm} \quad$\implies \tup{t_2,t_1}\in\co$
				}
				\vspace{-1.5mm}
				
				\caption{$\mathsf{Serializability}$}
				\label{ser_def}
			\end{subfigure} 
			\\\hline
			\end{tabularx}
	\end{subfigure}	
	}
\end{center}
\vspace{-7.85mm} 
\begin{center}
	\resizebox{\textwidth}{!}{
		\footnotesize
\begin{subfigure}{\textwidth}
\begin{tabularx}{\textwidth}{|*{2}{X|}}
	\hline \begin{subfigure}[b]{.495\textwidth}
		\centering
		\resizebox{.6\textwidth}{!}{

		\begin{tikzpicture}[->,>=stealth,shorten >=1pt,auto,node distance=4cm,
			semithick, transform shape]
		\node[transaction state, text=black] at (0,0)       (t_1)           {$t_1$};
		\node[transaction state] at (2,0)       (r)           {$r$};
		\node[transaction state, text=black,label={above:\textcolor{black}{$\writeVar{ }{x}$}}] at (-0.5,1.5) (t_2) {$t_2$};
		\node[transaction state] at (1.5,1.5) (t_4) {$t_4$};
		\path (t_1) edge[wrColor] node {$\wro_x$} (r);
		\path (t_2) edge[double equal sign distance, coColor] node {$\co^*$} (t_4);
		\path (t_4) edge[right] node {$(\so \cup \wro)$} (r);
		\path (t_2) edge[left,double equal sign distance, coColor] node {$\co$} (t_1);
		\end{tikzpicture}
		}
		\parbox{\textwidth}{

		$\forall x \forall t_1 \forall t_2.\ t_1\neq t_2\ \land \ \tup{t_1,r}\in \wro_x \ \land$
		
		\hspace{8mm}$\writeVar{t_2}{x}\ \land$ 
		
		\hspace{8mm}$\tup{t_2,r}\in\co^* ;\,(\so \cup \wro)$
		
		\hspace{8mm} \quad $\implies \tup{t_2,t_1}\in\co$
		}
		
		\caption{$\mathsf{Prefix}$}
		\label{pre_def}
	\end{subfigure}
	&

	\begin{subfigure}[b]{.495\textwidth}
		\centering
		\resizebox{.6\textwidth}{!}{
			\begin{tikzpicture}[->,>=stealth,shorten >=1pt,auto,node distance=4cm,
				semithick, transform shape]
				\node[transaction state, text=black] at (0,0)       (t_1)           {$t_1$};
				\node[transaction state, label={[xshift=-5]right:$\writeVar{ }{\yvar}$}] at (2,0) (r) {$r$};
				\node[transaction state, text=black,label={above:\textcolor{black}{$\writeVar{ }{x}$}}] at (-0.5,1.5) (t_2) {$t_2$};
				\node[transaction state, label={[yshift=-2]above:{$\writeVar{}{\yvar}$}}] at (1.5,1.5) (t_4) {$t_4$};
				\path (t_1) edge[wrColor] node {$\wro_x$} (r);
				\path (t_2) edge[double equal sign distance, coColor] node {$\co^*$} (t_4);
				\path (t_4) edge[right, double equal sign distance, coColor] node {$\co$} (r);
				\path (t_2) edge[left,double equal sign distance, coColor] node {$\co$} (t_1);
				\end{tikzpicture}
		}
		\parbox{\textwidth}{

		$\forall x \forall t_1 \forall t_2.\ t_1\neq t_2\ \land \ \tup{t_1,r}\in \wro_x \ \land$
			
		\hspace{8mm}$ \writeVar{t_2}{\xvar}\ \land \ t_4, \trans{r} \ \mathsf{write} \ \yvar\ \land$ 
			
		\hspace{8mm}$  \tup{t_2,t_4}\in\co^* \ \land \ \tup{t_4,r}\in\co$
			
		\hspace{8mm} \quad $\implies \tup{t_2,t_1}\in\co$
		}
		
		\caption{$\mathsf{Conflict}$}
		\label{confl_def}
	\end{subfigure} \\ \hline
	\end{tabularx}
\end{subfigure}
	}
\end{center}

	\vspace{-4mm}
	\caption{Axioms defining $\RC$, $\RA$, $\SER$, $\PRE$ and $\SI$ isolations levels respectively. Visibility relations are ``inlined'' to match the definitions in \cite{DBLP:journals/pacmpl/BiswasE19}.
	}
\label{fig:consistency_defs}
\vspace{-6mm}
\end{figure*}

%% file: sections/full-saturable-polynomial.tex
\section{Complexity of Checking Consistency}
\label{sec:complexity}

\subsection{Saturation and Boundedness}
\label{ssec:assumptions-isolations}

We investigate the complexity of checking if a history is consistent.
Our axiomatic framework characterize isolation levels as a conjunction of axioms as in \Cref{eq:axiom}. However, some isolation levels impose stronger constraints than others. For studying the complexity of checking consistency, we classify them in two categories, saturable or not. %
An isolation level is \emph{saturable} if its visibility relations are defined without using the $\co$ relation (i.e. the grammar in \Cref{eq:axioms-constraints-rel} omits the $\co$ relation). Otherwise, we say that the isolation level is \emph{non-saturable}. For example, $\RC$ and $\RA$ are saturable while $\PRE, \SI$ and $\SER$ are not. 

\begin{definition}
An isolation configuration $\isolation{\hist}$ is \emph{saturable} if for every transaction $t$, $\isolation{\hist}(t)$ is a saturable isolation level. Otherwise, $\isolation{h}$ is \emph{non-saturable}.
\end{definition}

\input{algorithms-tex/compute-co}

\input{algorithms-tex/check-saturable}

We say an isolation configuration $\isolation{\hist}$ is \emph{bounded} if there exists a fixed $k \in \mathbb{N}$ s.t. for every transaction $t$, $\isolation{\hist}(t)$ is defined as a conjunction of at most $k$ axioms that contain at most $k$ quantifiers. For example, $\SER$ employs one axiom and four quantifiers while $\SI$ employs two axioms, $\axpre$ and $\axconf$, with four and five quantifiers respectively. Any isolation configuration composed with $\SER, \SI, \PRE, \RA$ and $\RC$ isolation levels is bounded. We assume in the following that isolation configurations are bounded.

Checking consistency requires computing the $\valuewrName$ function and thus, evaluating $\whereName$ predicates. In the following, we assume that evaluating $\whereName$ predicates on a single row requires constant time.

\subsection{Checking Consistency of Full Histories}
\label{ssec:consistency-full-histories}

\Cref{algorithm:checking-saturable} computes necessary and sufficient conditions for the existence of a consistent execution $\exec = \tup{h, \co}$ for a history $h$ with a saturable isolation configuration.
It calls \textsc{saturate}, defined in Algorithm~\ref{algorithm:necessary-co}, to compute a ``\emph{partial}'' commit order relation $\pco$ that includes $(\so \cup \wro)^+$ and any other dependency between transactions that can be deduced from the isolation configuration. A consistent execution exists iff this partial commit order is acyclic. \Cref{algorithm:checking-saturable} generalizes the results in~\cite{DBLP:journals/pacmpl/BiswasE19} %
for full histories with heterogeneous saturable isolation configurations. $ $%
\begin{noappendixVer}
\end{noappendixVer}
\begin{appendixVer}
The correctness and complexity analysis of \Cref{algorithm:necessary-co,algorithm:checking-saturable} can be found in \Cref{app:proofs-poly-consistency-saturable}.
\end{appendixVer}

\begin{restatable}{theorem}{polyConsistencySaturable}
\label{th:poly-consistency-saturable}
Checking consistency of full histories with bounded saturable isolation configurations can be done in polynomial time.
\end{restatable}

For bounded non-saturable isolation configurations, checking if a history is consistent is NP-complete as an immediate consequence of the results in~\cite{DBLP:journals/pacmpl/BiswasE19}. These previous results apply to the particular case of transactions having the same isolation level and being formed of classic read and write instructions on a fixed set of variables. The latter can be simulated by SQL queries using $\whereName$ predicates for selecting rows based on their key being equal to some particular value. For instance, $\iselect{\lambda r: \keyof{r} = x}$ simulates a read of a ``variable'' $x$. 

%% file: algorithms-tex/compute-co.tex
\begin{algorithm}[t]
\caption{Extending an initial $\pco$ relation with necessary ordering constraints}
\begin{algorithmic}[1]
\small
\Function {\textsc{saturate}}{$\hist = \tup{T, \so, \wro}$, $\pco$}
\Comment{$\pco$ must be transitive.}

\State $\pco_\result \gets \pco$

\ForAll{$\key \in \Vars$}
\ForAll{$r\in \readOp{h}, t_2 \neq \trans{r} \in T$ s.t. $\writeVar{t_2}{\key}$ and $t_2 \neq \trans{\wro_x^{-1}(r)}$}
\label{algorithm:necessary-co:init-loop}

\State $t_1 \gets \trans{\wro_\key^{-1}(r)}$ \Comment $t_1$ is well defined as $h$ is a full history.

\ForAll{$\mathsf{v} \in \visibilitySet{\isolation{h}(\trans{r})}$}
\label{algorithm:necessary-co:all-constraints}

\If{$\mathsf{v}(t_2, r, x)$}
\label{algorithm:necessary-co:constraint-evaluation}
\State $\pco_\result \gets \pco_\result\cup \{(t_2, t_1)\}$
\label{algorithm:necessary-co:add-tuple-co}

\EndIf
\EndFor
\EndFor
\EndFor

\State \Return $\pco_\result$

\EndFunction
\label{algorithm:necessary-co:end}

\end{algorithmic}
\label{algorithm:necessary-co}
\end{algorithm}

%% file: algorithms-tex/check-saturable.tex
\begin{algorithm}[t]
\caption{Checking saturable consistency}
\begin{algorithmic}[1]
\small
\Function{\textsc{checkSaturable}}{$\hist = \tup{T, \so, \wro}$}
\If{$\so \cup \wro$ is cyclic} \Return $\bfalse$
\label{algorithm:checking-saturable:so-wr-acyclic}
    
\EndIf

\State $\pco \gets \textsc{saturate}(h, (\so \cup \wro)^+)$
\label{algorithm:checking-saturable:compute-co}

\State \Return $\btrue$ if $\pco$ is acyclic, and  $\bfalse$, otherwise
\label{algorithm:checking-saturable:co-acyclic}

\EndFunction
\end{algorithmic}
\label{algorithm:checking-saturable}
\end{algorithm}

%% file: sections/observable-np.tex
\subsection{Checking Consistency of Client Histories}
\label{ssec:consistency-client-histories}

We show that going from full histories to client histories, the consistency checking problem becomes NP-complete, independently of the isolation configurations. Intuitively, NP-hardness comes from keys that are not included in outputs of SQL queries. The justification for the consistency of omitting such rows can be ambiguous, e.g., multiple values written to a row may not satisfy the predicate of the $\whereName$ clause, or multiple deletes can justify the absence of a row.

The \emph{width} of a history $\widthHistory{h}$ is the maximum number of transactions which are pairwise incomparable w.r.t. $\so$. 
In a different context, previous work~\cite{DBLP:journals/pacmpl/BiswasE19} showed that bounding the width of a history (consider it to be a constant) is a sufficient condition for obtaining polynomial-time consistency checking algorithms. This is not true for client histories.

\vspace{-1mm}
\begin{restatable}{theorem}{kExpressiveNP}
\label{th:k-expressive-np}
Checking consistency of bounded-width client histories with bounded isolation configuration stronger than $\RC$ and $\widthHistory{h} \geq 3$ is NP-complete.
\end{restatable}

The proof of NP-hardness uses a reduction from 1-in-3 SAT which is inspired by the work of Gibbons and Korach~\cite{DBLP:journals/siamcomp/GibbonsK97} (Theorem 2.7) concerning sequential consistency for shared memory implementations. Our reduction is a non-trivial extension because it has to deal with any weak isolation configuration stronger than $\RC$. $ $%
\begin{noappendixVer}
\end{noappendixVer}%
\begin{appendixVer}
A detailed proof of the result can be found in \Cref{app:proof-bounded-complexity}.%
\end{appendixVer}

The proof of \Cref{th:k-expressive-np} relies on using non-trivial predicates in $\whereName$ clauses. We also prove that checking consistency of client histories is NP-complete irrespectively of the complexity of these predicates. This result uses another class of histories, called \emph{partial-observation} histories. These histories are a particular class of client histories where events read all inserted keys, irrespectively of their $\whereName$ clauses (as if these clauses where $\mathit{true}$).

\begin{definition}
\label{def:partial-observation-history}
A \emph{partial observation history} $h = \tup{T, \so, \wro}$ is a client history for which there is a witness $\overline{h} = \tup{T, \so, \overline{\wro}}$ of $h$, s.t. for every $\key$, if $(w,r) \in \overline{\wro}_\key \setminus \wro_\key$, then $w$ deletes $\key$.
\end{definition}

\begin{restatable}{theorem}{extendNpComplete}
\label{th:extend-np-complete}
Checking consistency of partial observation histories with bounded isolation configurations stronger than $\RC$ is NP-complete.
\end{restatable} 

\begin{appendixVer}
The proof of NP-hardness (see \Cref{app:proofs-observable-np-complete}) uses a novel reduction from $3$ SAT. $ $%
\end{appendixVer}
\begin{noappendixVer}
The proof of NP-hardness%
uses a novel reduction from $3$ SAT. $ $%
\end{noappendixVer}
The main difficulty for obtaining consistent witnesses of partial observation histories is the ambiguity of which delete event is responsible for each absent row. %

%% file: sections/csob-algorithm.tex
\section{Effectively Checking Consistency of Client Histories}
\label{sec:csob-algorithm}

The result of \Cref{th:k-expressive-np} implicitly asks whether there exist conditions on the histories for which checking consistency remains polynomial as in \cite{DBLP:journals/pacmpl/BiswasE19}. We describe an algorithm for checking consistency of client histories and identify cases in which it runs in polynomial time.

Consider a client history $h  = \tup{T, \so, \wro}$ which is consistent. For every consistent witness $\overline{h} = \tup{T, \so, \overline{\wro}}$ of $h$ there exists a consistent execution of $\overline{h}$, $\exec = \tup{\overline{h}, \co}$. The commit order $\co$ contains $(\so \cup \wro)^+$ and any other ordering constraint derived from axioms by observing that $(\so \cup \wro)^+ \subseteq \co$. More generally, $\co$ includes all constraints generated by the least fixpoint of the function $\textsc{saturate}$ defined in Algorithm~\ref{algorithm:necessary-co} when starting from $(\so \cup \wro)^+$ as partial commit order.
This least fixpoint exists because $\textsc{saturate}$ is monotonic. It is computed as usual by iterating $\textsc{saturate}$ until the output does not change.  We use 
$\mathsf{FIX}(\lambda \mathit{R}: \textsc{saturate}(h, \mathit{R}))(\so \cup \wro)^+$ to denote this least fixpoint.
In general, such a fixpoint computation is just an under-approximation of $\co$, and it is not enough for determining $h$'s consistency. 

\input{algorithms-tex/checksobound}

\color{newVersionColor}
The algorithm we propose, described in Algorithm~\ref{algorithm:checksobound}, exploits the partial commit order $\pco$ obtained by such a fixpoint computation (line~\ref{algorithm:checksobound:co}) for determining $h$'s consistency.
\color{black}
For a read $r$, key $x$, we define $\mathtt{1}_x^r(\pco)$, resp., $\mathtt{0}_x^r(\pco)$, to be the set of transactions that are \emph{not} committed after $\trans{r}$ and which write a value that satisfies, resp., does not satisfy, the predicate $\where{r}$. The formal description of both sets can be seen in Equation~\ref{eq:one-zero-sets-def}.%

\vspace{-6mm}
\begin{align}
\label{eq:one-zero-sets-def}
\mathtt{1}_x^r(\pco) = \{ t \in T \ | \ (\trans{r}, t)\not\in \pco \ \land \ \where{r}(\valuewr{t}{x}) = 1\} \nonumber\\
\mathtt{0}_x^r(\pco) = \{ t \in T \ | \ (\trans{r}, t)\not\in \pco \ \land \ \where{r}(\valuewr{t}{x}) = 0\} \\[-7mm] \nonumber
\end{align}

\input{figures-tex/example-conflict-free-extensions}

\color{newVersionColor}
The set $\mathtt{0}_x^r(\pco)$ can be used to identify extensions that are not witness of a history. Let us consider the client history $h$ depicted in \Cref{fig:example-conflict-free:history}. Observe that $t_3$ is not reading $x_1$ and $t_5$ is not reading $x_2$. Table~\ref{fig:example-conflict-free:table-extensions} describes all possible full extensions $\overline{h}$ of $h$. 
An execution $\exec = \tup{\overline{h}, \co}$ is consistent if $(t, r) \in \overline{\wro}_x \setminus \wro_x$ implies $\where{r}(\valuewr{t}{x}) = 0$. This implies that extensions $h_1$, $h_4$, and $h_7$, where $(\init, t_5) \in \overline{\wro}_{x_2}$, are not witnesses of $h$ as $\where{t_5}(\valuewr{\init}{x_2}) = 1$. We note that $\init \not\in \mathtt{0}_{x_2}^{t_5}(\pco) = \{t_1\}$. Also, observe that $(t_5, t_3) \in \wro$; so extensions $h_3, h_6$ and $h_9$, where $(t_3, t_5) \in \overline{\wro}_{x_2}$, are not a witness of $h$. Once again, $t_3 \not\in \mathtt{0}_{x_2}^{t_5}(\pco)$. In general, for every read event $r$ and key $x$ s.t. $\wro_x^{-1}(r) \uparrow$, the extension of $h$ where $(t,r) \in \overline{\wro}_x$, $t \not\in \mathtt{0}_{x}^{r}(\pco)$, is not a witness of $h$. In particular, if $\wro_x^{-1}(r) \uparrow$ but $\mathtt{0}_x^r(\pco)=\emptyset$, then no witness of $h$ can exist.

The sets $\mathtt{0}_{x}^{r}(\pco)$ are not sufficient to determine if a witness is a consistent witness as our previous example shows: $ \mathtt{0}_{x_1}^{t_3}(\pco) = \{\init,t_2, t_5\}$, but $h_2$ is not consistent.
\color{black}
\Cref{algorithm:checksobound}, combines an enumeration of history extensions with a search for a consistent execution of each extension. The extensions are \emph{not} necessarily full.
\color{newVersionColor}
In case $\wro_x^{-1}(r)$ is undefined, we use sets $\mathtt{1}_x^r(\pco)$ to decide whether the extension of $h$ requires specifying $\wro_x^{-1}(r)$ for determining $h$'s consistency. \Cref{algorithm:checksobound} specifies  $\wro_x^{-1}(r)$ only if $(r,x)$ is a so-called \emph{conflict}, i.e., $\wro_x^{-1}(r)$ is undefined and $\mathtt{1}_x^r(\mathtt{\pco}) \neq \emptyset$.

\color{black}

\color{newVersionColor}
Following the example of \Cref{fig:example-conflict-free}, we observe that $\mathtt{1}_{x_1}^{t_3}(\pco) = \emptyset$, all transactions that write on $x_1$ write non-negative values; but instead $\mathtt{1}_{x_2}^{t_5}(\pco) = \{\init\}$. Intuitively, this means that if some extension $h'$ that does not specify $\wro_{x_1}^{-1}(t_3)$ does not violate any axiom when using some commit order $\co$, then we can extend $h'$, defining $\wro_{x_1}^{-1}(t_3)$ as some adequate transaction, and obtain a full history $\overline{h}$ s.t. the execution $\exec = \tup{\overline{h}, \co}$ is consistent. On the other hand, specifying the write-read dependency of $t_5$ on $x_2$ matters. For not contradicting any axiom using $\co$, we may require $(\init, t_5) \in \overline{\wro}_{x_2}$. However, such extension is not even a witness of $h$ as $\where{\init}(\valuewr{\init}{x_2}) = 1$.
This intuition holds for the particular definitions of the isolation levels that \Cref{algorithm:checksobound} considers.

A history is \emph{conflict-free} if it does not have conflicts. Our previous discussion reduces the problem of checking consistency of a history to checking consistency of its conflict-free extensions. For example, the history $h$ in \Cref{fig:example-conflict-free:history} is not conflict-free but the extension $h_{258}$ defined in Table~\ref{fig:example-conflict-free:table-conflict-free} is. Instead of checking consistency of the nine possible extensions, we only check consistency of $h_{258}$.

\Cref{algorithm:checksobound} starts by checking if there is at least a conflict-free extension of $h$ (line~\ref{algorithm:checksobound:inconsistent-case}). If $h$ is conflict-free, it directly calls \Cref{algorithm:csob} (line~\ref{algorithm:checksobound:call-csob-h}); while otherwise, it iterates over conflict-free extensions of $h$, calling \Cref{algorithm:csob} on each of them (line~\ref{algorithm:checksobound:call-csob-h-prime}). %
\color{black}

\Cref{algorithm:csob} describes the search for the commit order of a conflict-free history $h$. This is a recursive enumeration of consistent prefixes of histories that backtracks when detecting inconsistency (it generalizes Algorithm $2$ in \cite{DBLP:journals/pacmpl/BiswasE19}). A \emph{prefix} of a history $\hist = \tup{T, \so, \wro}$ is a tuple $P = \tup{T_P, M_P}$ where $T_P \subseteq T$ is a set of transactions and $M_P : \Vars \to T_P$ is a mapping s.t. (1) $\so$ predecessors of transactions in $T_P$ are also in $T_P$, i.e., $\forall \tr\in T_P.\ \so^{-1}(\tr)\in T_P$ and (2) for every $x$, $M_P(x)$ is a $\so$-maximal transaction in $T_P$ that writes $x$ ($M_P$ records a last write for every key).

For every prefix $P = \tup{T_P, M_P}$ of a history $h$ and a transaction $t \in T \setminus T_P$, we say a prefix $P' = \tup{T_{P'}, M_{P'}}$ of $h$ is an \emph{extension} of $P$ \emph{using} $t$ if $T_{P'} = T_P \cup \{t\}$ and for every key $x$, $M_{P'}(x)$ is $t$ or $M_P(x)$. \Cref{algorithm:csob} extensions, denoted as $P \cup \{t\}$, guarantee that for every key $x$, if $\writeVar{t}{x}$, then $M_{P'}(x) = t$.

\color{newVersionColor}
Extending the prefix $P$ using $t$ means that any transaction $t'\in T_P$ is committed before $t$. \Cref{algorithm:csob} focuses on special extensions that lead to commit orders of consistent executions. 

\input{tables/table-vis-prefix}

\vspace{-1mm}
\begin{definition}
\label{def:consistent-extension}
Let $h$ be a history, $P = \tup{T_P, M_P}$ be a prefix of $h$, $t$ a transaction that is not in $T_P$ and $P' = \tup{T_{P'}, M_{P'}}$ be an exetension of $P$ using $t$. The prefix $P'$ is a \emph{consistent extension} of $P$ with $t$, denoted by $P \triangleright_t P'$, if
\vspace{-2mm}
\begin{enumerate}

    \item $P$ is $\pco$-closed: for every transaction $t' \in T$ s.t. $(t', t) \in \pco$ then $t' \in T_P$, \label{def:consistent-extension:2}

    \item $t$ does not overwrite other transactions in $P$: for every $\iread$ event $r$ outside of the prefix, i.e., $\trans{r} \in T \setminus T_{P'}$ and every visibility relation $v \in \visibilitySet{\isolation{h}}(\trans{r})$, the predicate $\visConsPrefix{v}{P}{t}{r}$ defined in \Cref{table:visibility-prefix} holds in $h$.%
     \label{def:consistent-extension:3}
\end{enumerate}
\end{definition}
\color{black}
We say that a prefix is consistent if it is either the empty prefix or it is a consistent extension of a consistent prefix.

\input{figures-tex/example-multicsob}

\color{newVersionColor}

\Cref{fig:example-csob:executions} depicts the execution of \Cref{algorithm:csob} on the conflict-free history \Cref{fig:example-csob:history} (history $h_{258}$ from Table~\ref{fig:example-conflict-free:table-conflict-free}). Blocked and effectuated calls are represented by read and blue arrows respectively. The read arrow $a$ is due to condition~\ref{def:consistent-extension:2} in \Cref{def:consistent-extension}: as $t_3$ enforces $\PRE$, reads $x_4$ from $t_2$, and $t_4$ is visible to it ($\mathsf{vis}_{\axpre}(t_4, t_3, x_4)$), $(t_4, t_2) \in \pco$; so consistent prefixes can not contain $t_2$ if they do not contain $t_4$. The read arrow $b$ is due to condition~\ref{def:consistent-extension:3}: as $t_5$ enforces $\SER$ and it reads $x_4$ from $t_4$, consistent prefixes can not contain $t_2$ unless $t_5$ is included. When reaching prefix $\langle t_3, t_5 \rangle$, the search terminates and deduces that $h$ is consistent. From the commit order induced by the search tree we can construct the extension of $h$ where missing write-read dependencies are obtained by applying the axioms on such a commit order. In our case, from $\init <_{\co} t_1 <_{\co} t_4 <_{\co} t_5 <_{\co} t_2 <_{\co} t_3$, we deduce that the execution $\exec = \tup{h_5, \co}$ is a consistent execution of $h_{258}$, and hence of $h$; where $h_5$ is the history described in Table~\ref{fig:example-conflict-free:table-extensions}.
\color{black}

\begin{appendixVer}
For complexity optimizations (see \Cref{ssec:algorithm-complexity}), \Cref{algorithm:csob} requires an isolation level-dependent equivalence relation between consistent prefixes. If there is transaction $t \in T$ s.t. $\isolation{h}(t) = \SI$, prefixes $P = \tup{T_P, M_P}$ and $P' = \tup{T_{P'}, M_{P'}}$ are \emph{equivalent} iff they are equal (i.e. $T_P = T_{P'}, M_P = M_{P'}$). Otherwise, they are \emph{equivalent} iff $T_P = T_{P'}$.

The proof of \Cref{algorithm:checksobound}'s correctness can be found in \Cref{ssec:proof-algorithm}.
\end{appendixVer}

\begin{noappendixVer}
For complexity optimizations, \Cref{algorithm:csob} requires an isolation level-dependent equivalence relation between consistent prefixes. If there is transaction $t \in T$ s.t. $\isolation{h}(t) = \SI$, prefixes $P = \tup{T_P, M_P}$ and $P' = \tup{T_{P'}, M_{P'}}$ are \emph{equivalent} iff they are equal (i.e. $T_P = T_{P'}, M_P = M_{P'}$). Otherwise, they are \emph{equivalent} iff $T_P = T_{P'}$.

\end{noappendixVer}

\input{algorithms-tex/csob}

\vspace{-2mm}
\begin{restatable}{theorem}{csobTheorem}
\label{th:csob}
Let $h$ be a client history whose isolation configuration is defined using $\{\SER, \SI, \PRE, \RA, \RC\}$. Algorithm~\ref{algorithm:checksobound} returns $\btrue$ if and only if $h$ is consistent.
\vspace{-2mm}
\end{restatable}

In general, Algorithm~\ref{algorithm:checksobound} is exponential the number of conflicts in $h$. The number of \emph{conflicts} is denoted by $\conflictsHistory{h}$. The number of conflicts exponent is implied by the number of mappings in $X_h$ explored by \Cref{algorithm:checksobound} ($E_h$ is the set of conflicts in $h$). 
The history width and size exponents comes from the number of prefixes explored by Algorithm~\ref{algorithm:csob} which is $|h|^{\widthHistory{h}} \cdot \widthHistory{h}^{|\Keys|}$ in the worst case (prefixes can be equivalently described by a set of $\so$-maximal transactions and a mapping associating keys to sessions). $ $%
\begin{noappendixVer}
\end{noappendixVer}
\begin{appendixVer}
The full detail of \Cref{algorithm:checksobound} complexity's proof can be found in \Cref{ssec:algorithm-complexity}.%
\end{appendixVer}%

\vspace{-1mm}
\begin{restatable}{theorem}{csobPolyTheorem}
\label{th:csob-poly}
For every client history $h$ whose isolation configuration is composed of $\{\SER, \SI, \PRE, \RA, \RC\}$ isolation levels, Algorithm~\ref{algorithm:checksobound} runs in $\mathcal{O}(|h|^{\conflictsHistory{h} + \widthHistory{h} + 9} \cdot \widthHistory{h}^{|\Vars|})$. Moreover, if no transaction employs $\SI$ isolation level, Algorithm~\ref{algorithm:checksobound} runs in $\mathcal{O}(|h|^{\conflictsHistory{h} + \widthHistory{h} + 8})$.
\vspace{-1mm}
\end{restatable}

On bounded, conflict-free histories only using $\SER, \PRE, \RA, \RC$ as isolation levels, \Cref{algorithm:checksobound} runs in polynomial time.
For instance, standard reads and writes can be simulated using $\einsert$ and $\eselect$ with $\whereName$ clauses that select rows based on their key being equal to some particular value. In this case, histories are conflict-less ($\wro$ would be defined for the particular key asked by the clause, and writes on other keys would not satisfy the clause). A more general setting where $\whereName$ clauses restrict only values that are immutable during the execution (e.g., primary keys) and deletes only affect non-read rows also falls in this category.

%% file: algorithms-tex/checksobound.tex
\begin{figure}[t]
\vspace{-6mm}
\begin{algorithm}[H]
\caption{Checking consistency of client histories}
\begin{algorithmic}[1]
\small
\Function{\checksobound}{$\hist = \tup{T, \so, \wro}$}

\Let $\pco = \mathsf{FIX}(\lambda \mathit{R}: \textsc{saturate}(h, \mathit{R}))(\so \cup \wro)^+$
\label{algorithm:checksobound:co}

\Let $E_h = \{(r,x) \ | \ r \in \readOp{h}, x \in \Vars. \wro_x^{-1}(r) \uparrow $ and $\mathtt{1}_r^x(\pco) \neq \emptyset\}$
\label{algorithm:checksobound:eh}

\Let $X_h$ = the set of mappings that map each $
{\arraycolsep=0.3pt
\begin{array}{lll}
    (r,x)&\in& E_h
\end{array}}$ to a member of $\mathtt{0}_x^r(\pco)$
\label{algorithm:checksobound:xh}

\If{$\pco$ is cyclic} \ReturnNameAlgorithmic $\bfalse$
\label{algorithm:checksobound:co-cyclic}

\ElsIf{there exists $(r,x)\in E_h$ such that $\mathtt{0}_x^r(\pco) = \emptyset$} \ReturnNameAlgorithmic $\bfalse$
\label{algorithm:checksobound:inconsistent-case}

\ElsIf{$E_h = \emptyset$} \Return \csobAlgorithm$(h, \emptyset)$

\label{algorithm:checksobound:call-csob-h}

\Else
\label{algorithm:checksobound:else-case}

\ForAll{$f \in X_h$}

\State $\mathtt{seen} \gets \emptyset; \; h' \gets h \bigoplus_{(r,x) \in E_h} \wro_x(f(r,x), r)$
\label{algorithm:checksobound:h-prime-definition}

\If{\csobAlgorithm$(h',\emptyset)$} \Return $\btrue$
\label{algorithm:checksobound:call-csob-h-prime}

\EndIf

\EndFor

\ReturnAlgorithmic $\bfalse$

\EndIf

\EndFunction
\label{algorithm:checksobound:end}

\end{algorithmic}
\label{algorithm:checksobound}
\end{algorithm}
\vspace{-10mm}
\end{figure}

%% file: figures-tex/example-conflict-free-extensions.tex
\begin{figure}[t]
\centering
\begin{subfigure}{0.56\textwidth}
\centering
\resizebox{.945\textwidth}{!}{
    \begin{tikzpicture}[>=stealth',shorten >=1pt,auto,node distance=3cm,
        semithick, transform shape,
        B/.style = {%
        decoration={brace, amplitude=1mm,#1},%
        decorate},
        B/.default = ,  %
        ]

        \node[minimum width=12em, draw, rounded corners=2mm,outer sep=0, label={[font=\small]10:$\init$}] (init) at (3, 0) {$\iinsert{\{x_1 : 0, x_2: 0, x_3 : 0, x_4: 0\}}$};

        \node[minimum width=12em, draw, rounded corners=2mm,outer sep=0, label={[font=\small]167:$t_1$}] (t1) at (0, -1.75) {
            \begin{tabular}{l}
                $\iinsert{\{x_2: -1, x_3 : 1 \}}$
            \end{tabular}
        };

        \node[minimum width=12em, draw, rounded corners=2mm,outer sep=0, label={[font=\small]167:$t_2$}] (t2) at (0, -3.5) {
            \begin{tabular}{l}
                $\iinsert{\{x_1 : 2, x_4: -2\}}$
            \end{tabular}
        };
        \node[minimum width=12em, draw, rounded corners=2mm,outer sep=0, label={[font=\small]160:$t_3$}] (t3) at (0, -5.25) {
            \begin{tabular}{l}
                $\iselect{\lambda r: r < 0}$ \\
                $\iinsert{\{x_2 : -3\}}$
            \end{tabular}
        };

        \node[minimum width=12em, draw, rounded corners=2mm,outer sep=0, label={[font=\small]13:$t_4$}] (t4) at (6, -1.75) {
            \begin{tabular}{l}
                $\iinsert{\{x_4: 4\}}$
            \end{tabular}
        };

        \node[minimum width=12em, draw, rounded corners=2mm,outer sep=0, label={[font=\small]20:$t_5$}] (t5) at (6, -3.5) {
            \begin{tabular}{l}
                $\iselect{\lambda r: r \geq 0}$ \\
                $\iinsert{\{x_1: 5, x_3 : -5\}}$
            \end{tabular}
        };

        \path (init) edge[->, soColor, right] node [above] {$\so$} (t4);
        \path (init) edge[->, soColor, right] node [above] {$\so$} (t1);
        \path (t4) edge[->, soColor, right, transform canvas={xshift=1mm}] node [right] {$\so$} (t5);
        \path (t1) edge[->, soColor, right, transform canvas={xshift=0mm}] node [right] {$\so$} (t2);
        \path (t2) edge[->, soColor, right, transform canvas={xshift=1mm}] node [right] {$\so$} (t3);

        \path (t1.south east) edge[->, wrColor] node [black, above] {$\wro_{x_3}$} (t5.north west);

        \path (t1.south east) edge[->, wrColor, bend left] node [black, right] {$\wro_{x_2}$} (t3.north east);
        \path (init.south) edge[->, wrColor] node [black, left, shift={(0,0.2)}] {$\wro_{x_1}$} (t5.north west);
        \path (t5.south west) edge[->, wrColor] node [black, below] {$\wro_{x_3}$} (t3.north east);
        \path (t4) edge[->, wrColor, transform canvas={xshift=-1mm}] node [black, left] {$\wro_{x_4}$} (t5);
        \path (t2) edge[->, wrColor, transform canvas={xshift=-1mm}] node [black, left] {$\wro_{x_4}$} (t3);

    \end{tikzpicture}  
}

\caption{A history where $t_3, t_5$ have $\PRE$ and $\SER$ as isolation levels respectively. The isolation levels of the other transactions are unspecified.
}
\label{fig:example-conflict-free:history}
\end{subfigure}
\hfill
\begin{subfigure}{0.43\textwidth}
\centering
\begin{minipage}{\textwidth}
\centering
\begin{tabular}{|l|l|l|}
    \hline
    \textbf{History} & \textbf{$\wro_{x_1}^{-1}(t_3)$} & \textbf{$\wro_{x_2}^{-1}(t_5)$}\\ \hline
    $h_1$ & $\init$ & $\init$\\
    $h_2$ & $\init$ & $t_1$\\
    $h_3$ & $\init$ & $t_3$\\
    $h_4$ & $t_2$ & $\init$\\
    $h_5$ & $t_2$ & $t_1$\\
    $h_6$ & $t_2$ & $t_3$\\
    $h_7$ & $t_5$ & $\init$\\
    $h_8$ & $t_5$ & $t_1$\\
    $h_9$ & $t_5$ & $t_3$\\\hline
    \end{tabular}
\caption{Table describing all possible full extensions of the history in \Cref{fig:example-conflict-free:history}.}
\label{fig:example-conflict-free:table-extensions}
\end{minipage}
\begin{minipage}{\textwidth}
    \centering
\begin{tabular}{|l|l|l|}
    \hline
    \textbf{History} & \textbf{$\wro_{x_1}^{-1}(t_3)$} & \textbf{$\wro_{x_2}^{-1}(t_5)$}\\ \hline
    $h_{258}$ & \textbf{undef} & $t_1$\\\hline
\end{tabular}
\caption{Table describing the only conflict-free extension of \Cref{fig:example-conflict-free:history}.} 
\label{fig:example-conflict-free:table-conflict-free}
\end{minipage}

\end{subfigure}
\vspace{-1mm}
\caption{
\color{newVersionColor}
Comparison between conflict-free extensions and full extensions of the history $h$ in \Cref{fig:example-conflict-free:history}. In $h$, $\wro^{-1}$ is not defined for two pairs: $(t_3, x_1)$ and $(t_5, x_2)$; where we identify the single $\eselect$ event in a transaction with its transaction. Table~\ref{fig:example-conflict-free:table-extensions} describes all possible full extensions of $h$. For example, the first extension, $h_1$, states that $(\init, t_3) \in \wro_{x_1}$ and $(\init, t_5) \in \wro_{x_2}$. \Cref{algorithm:checksobound} only explore the only extension $h_{258}$ described in Table~\ref{fig:example-conflict-free:table-conflict-free}; where $\wro_{x_1}^{-1}(t_3) \uparrow$ and $(t_1, t_5) \in \wro_{x_2}$. The history $h_{258}$ can be extended to histories $h_2$, $h_5$ and $h_8$.
\color{black}
}
\label{fig:example-conflict-free}
\vspace{-6mm}

\end{figure}

%% file: tables/table-vis-prefix.tex
\begin{table}[!ht]
\vspace{-2mm}
\centering
\begin{adjustbox}{width=\textwidth}
    \begin{tabular}{|l|l|}
    \hline
        \textbf{Axiom} & \textbf{Predicate} \\ \hline
        \axser, \axpre, & $\nexists x \in \Vars$ s.t. $ \writeVar{t}{x}$, $\wro_x^{-1}(r) \downarrow$ \\
        \axra, \axrc &  \; $v(\pco^P_t)(t, r, x)$ holds in $h$ and $\wro_x^{-1}(r) \in T_P$\\\hline
        \axconf &  $\nexists x \in \Vars, t' \in T_P \cup \{t\}$ s.t. $ \writeVar{t'}{x}$, $\wro_x^{-1}(r) \downarrow$\\
        & \;  $v(\pco^P_t)(t', r, x)$ holds in $h$ and $\wro_x^{-1}(r) \neq M_P(x)$ \\\hline
    \end{tabular}
\end{adjustbox}
\caption{Predicates relating prefixes and visibility relations where $\pco_t^P$ is defined as $ \pco \cup  \{(t', t) \ | \ t' \in T_P\} \cup \{(t, t'') \ | \  t'' \in T \setminus (T_P \cup \{t\})\}$.}
\label{table:visibility-prefix}
\vspace{-10mm}
\end{table}

%% file: figures-tex/example-multicsob.tex
\begin{figure*}[t]
\centering
\begin{subfigure}{0.545\textwidth}
\centering
\resizebox{.945\textwidth}{!}{
    \begin{tikzpicture}[>=stealth',shorten >=1pt,auto,node distance=3cm,
        semithick, transform shape,
        B/.style = {%
        decoration={brace, amplitude=1mm,#1},%
        decorate},
        B/.default = ,  %
        ]

        \node[minimum width=12em, draw, rounded corners=2mm,outer sep=0, label={[font=\small]10:$\init$}] (init) at (3, 0) {$\iinsert{\{x_1 : 0, x_2: 0, x_3 : 0, x_4: 0\}}$};

        \node[minimum width=12em, draw, rounded corners=2mm,outer sep=0, label={[font=\small]167:$t_1$}] (t1) at (0, -1.75) {
            \begin{tabular}{l}
                $\iinsert{\{x_2: -1, x_3 : 1 \}}$
            \end{tabular}
        };

        \node[minimum width=12em, draw, rounded corners=2mm,outer sep=0, label={[font=\small]167:$t_2$}] (t2) at (0, -3.5) {
            \begin{tabular}{l}
                $\iinsert{\{x_1 : 2, x_4: 2\}}$
            \end{tabular}
        };
        \node[minimum width=12em, draw, rounded corners=2mm,outer sep=0, label={[font=\small]160:$t_3$}] (t3) at (0, -5.25) {
            \begin{tabular}{l}
                $\iselect{\lambda r: r < 0}$ \\
                $\iinsert{\{x_2 : -3\}}$
            \end{tabular}
        };

        \node[minimum width=12em, draw, rounded corners=2mm,outer sep=0, label={[font=\small]13:$t_4$}] (t4) at (6, -1.75) {
            \begin{tabular}{l}
                $\iinsert{\{x_4: 4\}}$
            \end{tabular}
        };

        \node[minimum width=12em, draw, rounded corners=2mm,outer sep=0, label={[font=\small]20:$t_5$}] (t5) at (6, -3.5) {
            \begin{tabular}{l}
                $\iselect{\lambda r: r \geq 0}$ \\
                $\iinsert{\{x_1: 5, x_3 : -5\}}$
            \end{tabular}
        };

        \path (init) edge[->, soColor, right] node [above] {$\so$} (t4);
        \path (init) edge[->, soColor, right] node [above] {$\so$} (t1);
        \path (t4) edge[->, soColor, right, transform canvas={xshift=1mm}] node [right] {$\so$} (t5);
        \path (t1) edge[->, soColor, right, transform canvas={xshift=0mm}] node [right] {$\so$} (t2);
        \path (t2) edge[->, soColor, right, transform canvas={xshift=1mm}] node [right] {$\so$} (t3);

        \path (t1.south east) edge[->, wrColor] node [black, above, shift={(0.1,0)}] {$\wro_{x_2, x_3}$} (t5.north west);

        \path (init.south) edge[->, wrColor] node [black, left, shift={(0,0.2)}] {$\wro_{x_1}$} (t5.north west);
        \path (t1.south east) edge[->, wrColor, bend left] node [black, right] {$\wro_{x_2}$} (t3.north east);
        \path (t5.south west) edge[->, wrColor] node [black, below] {$\wro_{x_3}$} (t3.north east);
        \path (t4) edge[->, wrColor, transform canvas={xshift=-1mm}] node [black, left] {$\wro_{x_4}$} (t5);
        \path (t2) edge[->, wrColor, transform canvas={xshift=-1mm}] node [black, left] {$\wro_{x_4}$} (t3);

    \end{tikzpicture}     
}

\caption{Conflict-free history corresponding to the extension $h_{258}$ (Table~\ref{fig:example-conflict-free:table-conflict-free}) of the history in \Cref{fig:example-conflict-free:history}
}
\label{fig:example-csob:history}
\end{subfigure}
\hfill
\begin{subfigure}{0.445\linewidth}
\centering
\resizebox{.65\textwidth}{!}{
    \begin{tikzpicture}[>=stealth',shorten >=1pt,auto,node distance=3cm,
        semithick, transform shape,
        B/.style = {%
        decoration={brace, amplitude=1mm,#1},%
        decorate},
        B/.default = ,  %
        ]
        \node[minimum width=2em,  minimum height=1.5em, draw, rounded corners=2mm,outer sep=0] (start) at (4, 0) {$\emptyset$};

        \node[minimum width=2em, draw, rounded corners=2mm,outer sep=0] (p1) at (3, -0.8) {
            $\langle t_1\rangle$
        };
        \node[minimum width=2em, draw, rounded corners=2mm,outer sep=0] (p2) at (2, -1.6) {
            $\langle t_2\rangle$
        };
    
        \node[minimum width=2em, draw, rounded corners=2mm,outer sep=0] (p14) at (4, -1.6) {
            $\langle t_1, t_4\rangle$
        };
       
        \node[minimum width=2em, draw, rounded corners=2mm,outer sep=0] (p24) at (3, -2.4) {
            $\langle t_2, t_4\rangle$
        };
        \node[minimum width=2em, draw, rounded corners=2mm,outer sep=0] (p15) at (5, -2.4) {
            $\langle t_1, t_5\rangle$
        };
        \node[minimum width=2em, draw, rounded corners=2mm,outer sep=0] (p25) at (4, -3.2) {
            $\langle t_2, t_5\rangle$
        };
        \node[minimum width=2em, draw, rounded corners=2mm,outer sep=0] (p35) at (3, -4) {
            $\langle t_3, t_5\rangle$
        };

        \node[minimum width=4.5em, minimum height=2.5em, draw, wrColor, rounded corners=2mm,outer sep=0] (22circ) at (3, -4) {
            
        };

        \path (start) edge[->, blue, right] node [above] {} (p1);
        \path (p1) edge[->, red, right] node [above] {a} (p2);
        \path (p1) edge[->, blue, right] node [above] {} (p14);
        \path (p14) edge[->, red, right] node [above, shift={(-0.2,-0.05)}] {b} (p24);
        \path (p14) edge[->, blue, right] node [above] {} (p15);
        \path (p15) edge[->, blue, right] node [above] {} (p25);
        \path (p25) edge[->, blue, right] node [above] {} (p35);

    \end{tikzpicture}  
}

\caption{Execution of \Cref{algorithm:checksobound} on the history in \Cref{fig:example-csob:history}. }
\label{fig:example-csob:executions}

\end{subfigure}

\vspace{-1mm}
\caption{
\color{newVersionColor}
Applying \Cref{algorithm:csob} on the conflict-free consistent history $h_{258}$ on the left.
The right part pictures a search for valid extensions of consistent prefixes on $h_{258}$. Prefixes are represented by their
$\so$-maximal transactions, e.g., $\langle t_2 \rangle$ contains all transactions which are before $t_2$ in $\so$, i.e., $\{\init, t_1,t_2\}$. A red arrow means that the search is blocked (the prefix at the target is not a consistent extension), while a blue arrow mean that the search continues.
\color{black}
}
\label{fig:example-csob}
\vspace{-6mm}

\end{figure*}

%% file: algorithms-tex/csob.tex
\begin{figure}[t]
\vspace{-6mm}
\begin{algorithm}[H]
\caption{check consistency of conflict-free histories}
\begin{algorithmic}[1]
\small
\Function{\csobAlgorithm}{$\hist = \tup{T, \so, \wro}, P$}
\label{algorithm:csob:call-csob}

\If{$|P| = |T|$} \Return $\btrue$
\label{algorithm:csob:base-case}

\EndIf

\ForAll{$t \in T \setminus P$ s.t. $P \triangleright_t (P \cup \{t\})$}
\label{algorithm:csob:for-condition}

\If{$\exists P' \in \mathtt{seen}$ s.t. $P' \equiv_{\isolation{h}} (P \cup \{t\})$} \textbf{continue}
\label{algorithm:csob:no-conflict-invalid}

\ElsIf{\csobAlgorithm$(h, P \cup \{t\})$} \Return $\btrue$
\label{algorithm:csob:no-conflict-valid}

\Else $\ \mathtt{seen} \gets \mathtt{seen} \cup (P \cup \{t\})$
\label{algorithm:csob:add-prefix-seen}

\EndIf

\EndFor

\ReturnAlgorithmic $\bfalse$
\label{algorithm:csob:return-false}

\EndFunction

\end{algorithmic}
\label{algorithm:csob}
\end{algorithm}
\vspace{-7mm}
\end{figure}

%% file: sections/experiments.tex
\vspace{-2mm}
\section{Experimental evaluation}
\label{sec:exp}

\vspace{-2mm}
We evaluate an implementation of $\checksobound$ in the context of the Benchbase~\cite{difallah2013oltp} database benchmarking framework. We apply this algorithm on histories extracted from randomly generated client programs of a number of database-backed applications. We use PostgreSQL 14.10 as a database. The experiments were performed on an Apple M1 with $8$ cores and $16$ GB of RAM.

\input{figures-tex/experiments-sessions}
\noindent
\textbf{Implementation.}
We extend the Benchbase framework with an additional package for generating histories and checking consistency.
Applications from Benchbase are instrumented in order to be able to extract histories, the $\wro$ relation in particular. %
Our implementation is publicly available~\cite{csob-implementation}.

Our tool takes as input a configuration file specifying the name of the application and the isolation level of each transaction in that application.
For computing the $\wro$ relation and generating client histories, we extend the database tables with an extra column \texttt{WRITEID} which is updated by every $\iwrite$ instruction with a unique value. %
SQL queries are also modified to return whole rows instead of selected columns. To extract the $\wro$ relation for $\eupdate$ and $\edelete$ we add $\mathtt{RETURNING}$ clauses. Complex operators such as \texttt{INNER JOIN} are substituted by simple juxtaposed SQL queries (similarly to~\cite{DBLP:journals/pacmpl/BiswasKVEL21}). We map the result of each query to local structures for generating the corresponding history. Transactions aborted by the database (and not explicitly by the application) are discarded.

\noindent
\textbf{Benchmark.}
We analyze a set of benchmarks inspired by real-world applications and evaluate them under different types of clients and isolation configurations. \nver{We focus on isolation configurations implemented in PostgreSQL, i.e. compositions of $\SER, \SI$ and $\RC$ isolation levels.}

In average, the ratio of SER/SI transactions is $11\%$ for Twitter and $88\%$ for TPC-C and TPC-C PC. These distributions are obtained via the random generation of client programs implemented in BenchBase. In general, we observe that the bottleneck is the number of possible history extensions enumerated at line 9 in Alg. 3 and not the isolation configuration. This number is influenced by the distribution of types of transactions, e.g., for TPC-C, a bigger number of transactions creating new orders increases the number of possible full history extensions. We will clarify.

\textit{Twitter~\cite{difallah2013oltp}} models a social network that allows users to publish tweets and get their followers, tweets and tweets published by other followers. 
We consider five isolation configurations: $\SER$, $\SI$ and $\RC$ and the heterogeneous $\SER+\RC$ and $\SI+\RC$, where publishing a tweet is $\SER$ (resp., $\SI$) and the rest are $\RC$. \nver{The ratio of $\SER$ (resp. $\SI$) transactions w.r.t. $\RC$ is $11\%$ on average.}

\textit{TPC-C~\cite{TPCC}} models an online shopping application with five types of transactions: reading the stock, creating a new order, getting its status, paying it and delivering it. We consider five isolation configurations: the homogeneous $\SER$, $\SI$ and $\RC$ and the combinations $\SER+\RC$ and $\SI+\RC$, where creating a new order and paying it have $\SER$ (respectively $\SI$) as isolation level while the rest have $\RC$. \nver{The ratio of $\SER$ (resp. $\SI$) transactions w.r.t. $\RC$ is $88\%$ on average.}

\textit{TPC-C PC} is a variant of the TPC-C benchmark whose histories are always conflict-free.
\texttt{DELETE} queries are replaced by \texttt{UPDATE} with the aid of extra columns simulating the absence of a row. Queries whose $\whereName$ clauses query mutable values are replaced by multiple simple instructions querying only immutable values such as unique ids and primary keys. %

\noindent
\textbf{Experimental Results.}
We designed two experiments to evaluate
$\checksobound$'s performance 
for different isolation configurations increasing the number of transactions per session (the number of sessions is fixed), the number of sessions (the number of transactions per session is fixed), resp. We use a timeout of $60$ seconds per history.

The first experiment investigates the scalability of \Cref{algorithm:checksobound} when increasing the number of sessions. For each benchmark and isolation configuration, we consider $5$ histories of random clients (each history is for a different client) with an increasing number of sessions and $10$ transactions per session (around 400 histories across all benchmarks). No timeouts appear with less than $4$ sessions. \Cref{fig:sessions} shows the running time of the experiment. %

\input{figures-tex/experiments-running-time}

The second experiment investigates the scalability of \Cref{algorithm:checksobound} when increasing the number of transactions. For each benchmark and isolation configuration, we consider $5$ histories of random clients, each having $3$ sessions and an increasing number of transactions per session (around 1900 histories across all benchmarks). \Cref{fig:running-time} shows its running time. 

The runtime similarities between isolation configurations containing $\SI$ versus those without it show that in practice, the bottleneck of \Cref{algorithm:checksobound} is \nver{the number of possible history extensions enumerated at line~\ref{algorithm:checksobound:call-csob-h-prime} in \Cref{algorithm:checksobound}; i.e. the number of conflicts in a history. This number is influenced by the distribution of types of transactions, e.g., for TPC-C, a bigger number of transactions creating new orders increases the number of possible full history extensions.} Other isolation levels not implemented by PostgreSQL, e.g., prefix consistency $\PRE$, are expected to produce similar results.

Both experiments show that \Cref{algorithm:checksobound} scales well for histories with a small number of writes (like Twitter) or conflicts (like TPC-C PC). In particular, \Cref{algorithm:checksobound} is quite efficient for typical workloads needed to expose bugs in 
production databases which contain less than 10 transactions~\cite{DBLP:journals/pacmpl/BiswasE19,DBLP:journals/pacmpl/0003GWB24,DBLP:journals/pvldb/HuangLCWBLP23}.

A third experiment compares \Cref{algorithm:checksobound} with a baseline consisting in a naive approach where we enumerate witnesses and executions of such witnesses until consistency is determined.
We consider Twitter and TPC-C as benchmarks and execute 5 histories of random clients, each having 3 sessions and an increasing number of transactions per session (around 100 histories across all benchmarks). We execute each client under $\RC$ and check the obtained histories for consistency with respect to $\SER$. %

The naive approach either times out for $35.5\%$, resp., $95.5\%$ of the histories of Twitter, resp., TPC-C, or finishes in $5$s on average (max $25$s). In comparison, \Cref{algorithm:checksobound} has no timeouts for Twitter and times out for $5.5\%$ of the TPC-C histories; finishing in $1.5$s on average (max $12$s). Averages are computed w.r.t. non-timeout instances. The total number of executed clients is around $100$. Only one TPC-C history was detected as inconsistent, which shows that the naive approach does not  timeout only in the worst-case (inconsistency is a worst-case because all extensions and commit orders must be proved to be invalid).

A similar analysis on the TPC-C PC benchmark is omitted: TPC-C PC is a conflict-free variation of TPC-C with more operations per transaction. Thus, the rate of timeouts in the naive approach increases w.r.t. TPC-C, while the rate of timeouts using \Cref{algorithm:checksobound} decreases.

Comparisons with prior work \cite{DBLP:journals/pacmpl/BiswasE19,DBLP:journals/pvldb/AlvaroK20,DBLP:journals/pvldb/HuangLCWBLP23,DBLP:journals/pacmpl/0003GWB24} are not possible as they do not apply to SQL (see \Cref{sec:related-work} for more details).

This evaluation demonstrates that our algorithm scales well to practical testing workloads and that it outperforms brute-force search.

%% file: figures-tex/experiments-sessions.tex
\begin{figure*}[t]
	\centering
	\begin{subfigure}[t]{0.32\linewidth}
		\centering
		\includegraphics[width=\linewidth]{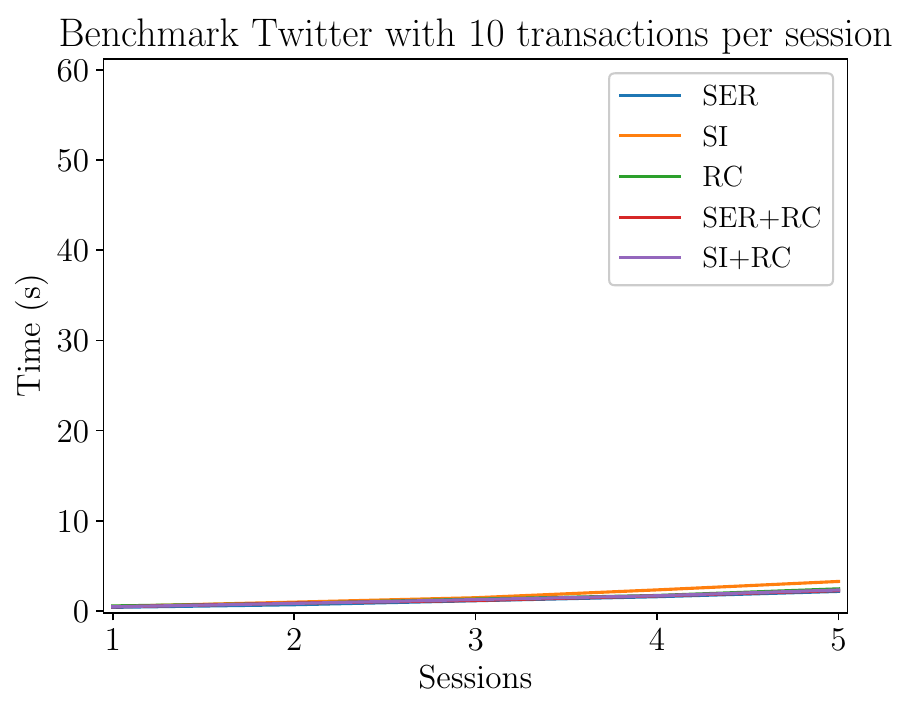}
		\vspace{-5mm}
		\label{fig:sessions-twitter}
	\end{subfigure}
	\hfill
	\begin{subfigure}[t]{0.32\linewidth}
		\centering
		\includegraphics[width=\linewidth]{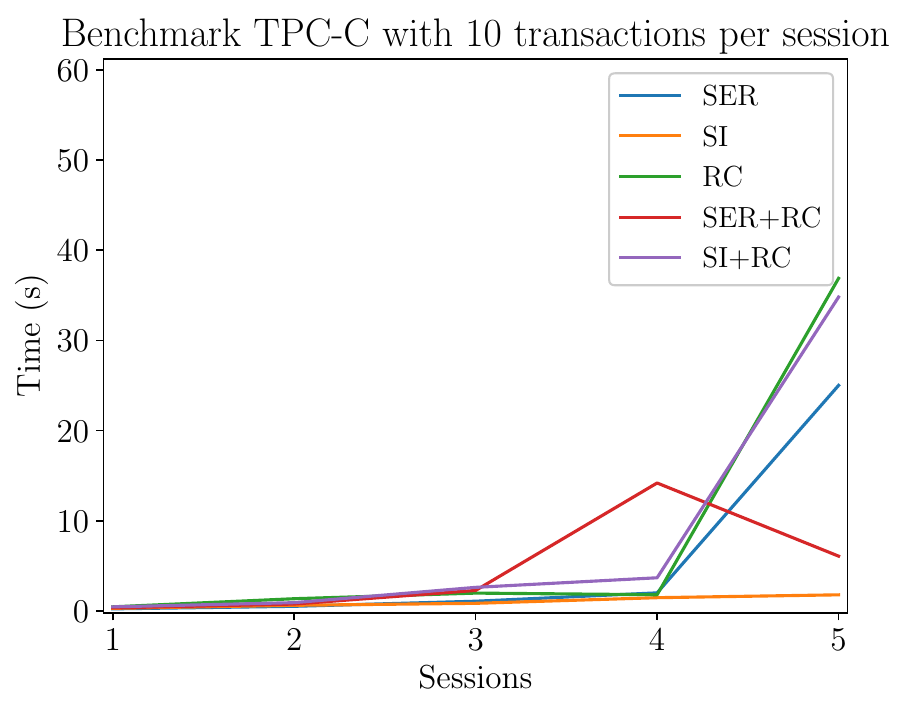}
		\vspace{-5mm}
		\label{fig:sessions-tpcc}
	\end{subfigure}
	\hfill
	\begin{subfigure}[t]{0.32\linewidth}
		\centering
		\includegraphics[width=\linewidth]{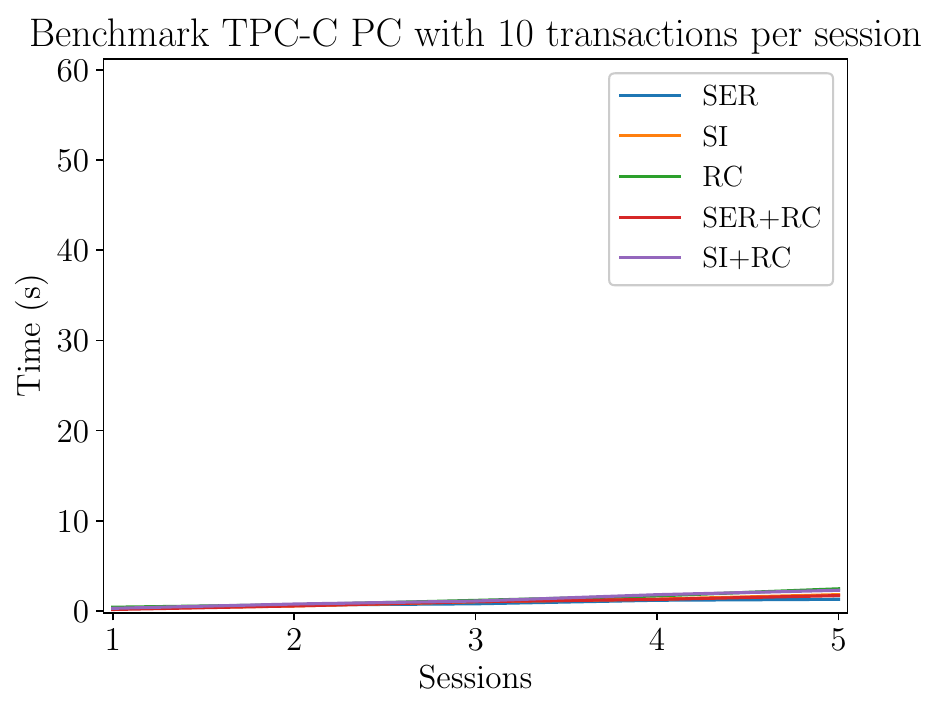}
		\vspace{-5mm}
		\label{fig:sessions-tpccpc}
	\end{subfigure}
	
	\vspace{-3mm}
	\caption[]{ 
	Running time of \Cref{algorithm:checksobound} while increasing the number of sessions. Each point represents the average running time of $5$ random clients of such size.} %
	\label{fig:sessions}
\vspace{-3mm}
\end{figure*}

%% file: figures-tex/experiments-running-time.tex
\begin{figure*}[t]
	\centering
	\begin{subfigure}[t]{0.32\linewidth}
		\centering
		\includegraphics[width=\linewidth]{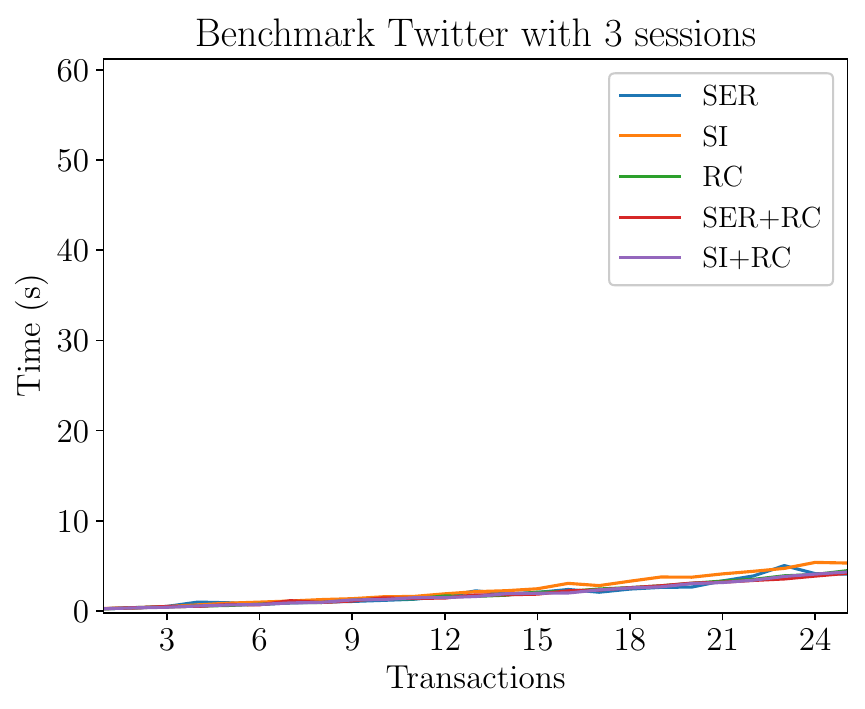}
		\vspace{-5mm}
		\label{fig:running-time-twitter}
	\end{subfigure}
	\hfill
	\begin{subfigure}[t]{0.32\linewidth}
		\centering
		\includegraphics[width=\linewidth]{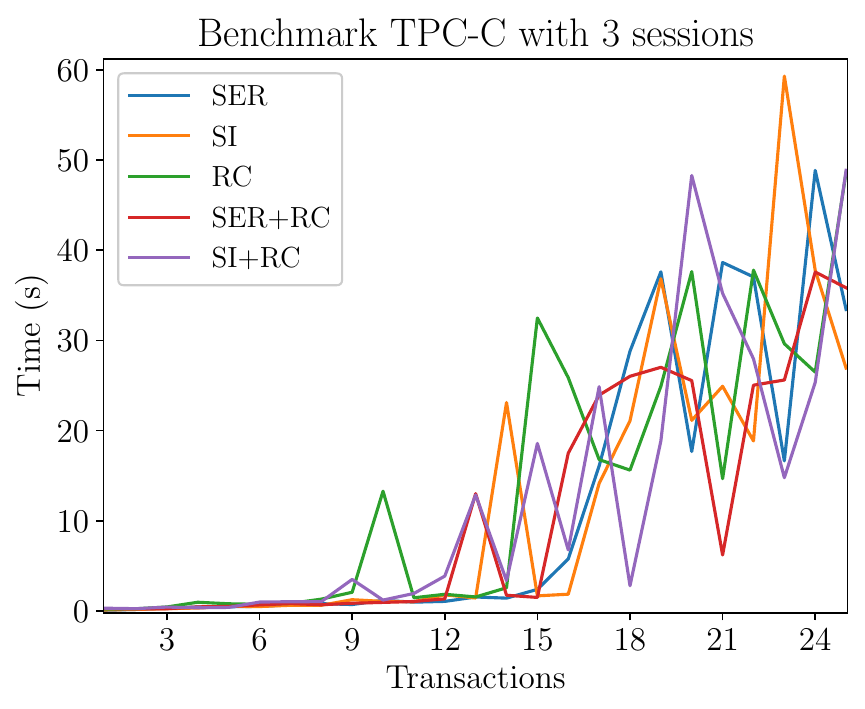}
		\vspace{-5mm}
		\label{fig:running-time-tpcc}
	\end{subfigure}
	\hfill
	\begin{subfigure}[t]{0.32\linewidth}
		\centering
		\includegraphics[width=\linewidth]{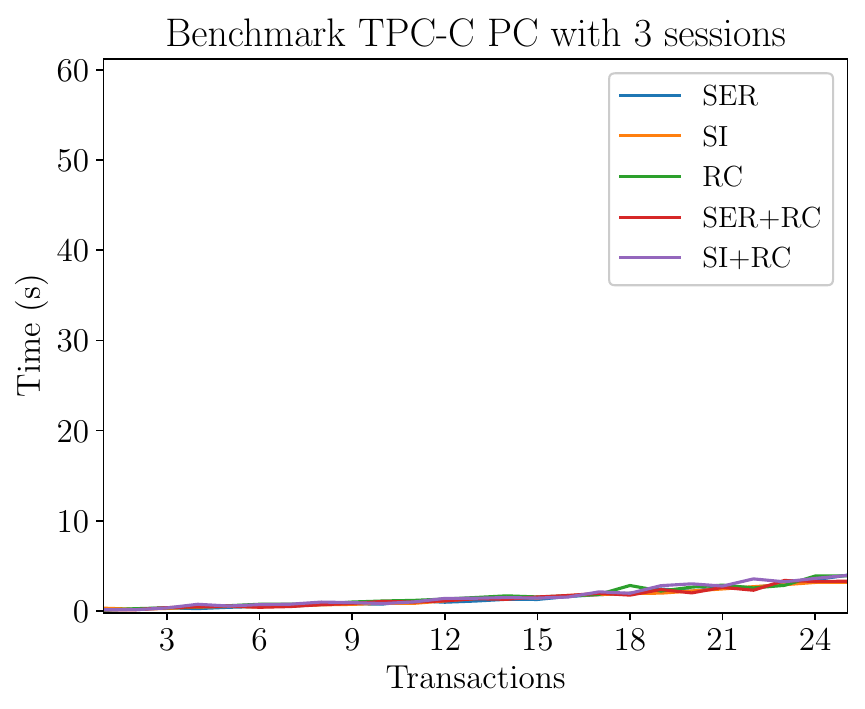}
		\vspace{-5mm}
		\label{fig:running-time-tpccpc}
	\end{subfigure}
	\vspace{-3mm}

	\caption{%
	Running time of \Cref{algorithm:checksobound} increasing the number of transactions per session. 
	We plot the average running time of $5$ random clients of such size.
	}
	\label{fig:running-time}
\vspace{-6mm}
\end{figure*}

%% file: sections/related-work.tex
\vspace{-2mm}
\section{Related work}
\label{sec:related-work}

\vspace{-1mm}
The formalization of database isolation levels has been considered in previous work. Adya~\cite{adya-thesis} has proposed axiomatic specifications for isolation levels, which however do not concern more modern isolation levels like $\PRE$ or $\SI$ and which are based on low-level modeling of database snapshots. We follow the more modern approach in \cite{DBLP:conf/concur/Cerone0G15,DBLP:journals/pacmpl/BiswasE19} which however addresses the restricted case when transactions are formed of reads and writes on a \emph{static} set of keys (variables) and not generic SQL queries, and all the transactions in a given execution have the same isolation level. Our axiomatic model builds on axioms defined by Biswas et al.~\cite{DBLP:journals/pacmpl/BiswasE19} which are however applied on a new model of executions that is specific to SQL queries. 

The complexity of checking consistency w.r.t isolation levels has been studied in~\cite{DBLP:journals/jacm/Papadimitriou79b,DBLP:journals/pacmpl/BiswasE19}. The work of Papadimitriou~\cite{DBLP:journals/jacm/Papadimitriou79b} shows that checking serializability is NP-complete while the work of Biswas et al.~\cite{DBLP:journals/pacmpl/BiswasE19} provides results for the same isolation levels as in our work, but in the restricted case mentioned above. 

Checking consistency in a non-transactional case, shared-memory or distributed systems, has been investigated in a number of works, e.g.,~\cite{DBLP:conf/popl/BouajjaniEGH17,DBLP:journals/siamcomp/GibbonsK97,DBLP:journals/pacmpl/EmmiE18,DBLP:journals/tpds/CantinLS05,DBLP:journals/ppl/GontmakherPS03,DBLP:journals/tecs/FurbachM0S15,DBLP:journals/pacmpl/AbdullaAJN18,DBLP:conf/spaa/GibbonsK94,DBLP:conf/cav/AgarwalCPPT21}. Transactions introduce additional challenges that make these results not applicable. 

Existing tools for checking consistency in the transactional case of distributed databases, e.g., \cite{DBLP:journals/pacmpl/BiswasE19,DBLP:journals/pvldb/AlvaroK20,DBLP:journals/pvldb/HuangLCWBLP23,DBLP:journals/pacmpl/0003GWB24} cannot handle SQL-like semantics, offering guarantees modulo their transformations to reads and writes on static sets of keys. Our results show that handling the SQL-like semantics is strictly more complex (NP-hard in most cases).

%% file: sections/acknowledgements.tex
\section*{Acknowledgements}

We thank the anonymous reviewers for their feedback. This work was partially supported by the Agence National de Recherche (ANR) grants ``AdeCoDS'' and ``CENTEANES''. %

%% file: appendices/operational-semantics.tex
\section{An operational semantics for SQL-like distributed databases (\Cref{ssec:operational-semantics}).}
\label{app:operational-semantics}

\input{figures-tex/operational-semantics}

\input{figures-tex/operational-semantics-2}

Formally, the operational semantics is defined as a transition relation $\Rightarrow$ between \emph{configurations}. A configuration is a tuple containing the following:
\begin{itemize}
	\item history $\hist$ recording the instructions executed in the past, 
	\item a valuation map $\vec{\gamma}$ that records local variable values in the current transaction of each session ($\vec{\gamma}$ associates identifiers of sessions that have live transactions with valuations of local variables),
	\item a map $\vec{\mathsf{B}}$ that stores the code of each live transaction (associating session identifiers with code),
	\item a map $\vec{\mathsf{I}}$ that tracks the isolation level of each executed transaction,
	\item a map $\vec{\mathsf{T}}$ that associates events in the history with unique timestamps,
	\item a map $\vec{\mathsf{S}}$ that associates events in the history with snapshots of the database,
	\item sessions/transactions $\mathsf{P}$ that remain to be executed from the original program.
\end{itemize}

For readability, we define a program as a partial function $\mathsf{P}:\mathsf{SessId}\rightharpoonup \mathsf{Sess}$ that associates session identifiers in $\mathsf{SessId}$ with sequences of transactions as defined in Section~\ref{ssec:syntax}. Similarly, the session order $\so$ in a history is defined as a partial function $\so:\mathsf{SessId}\rightharpoonup \mathsf{Tlogs}^*$ that associates session identifiers with sequences of transaction logs. Two transaction logs are ordered by $\so$ if one occurs before the other in some sequence $\so(j)$ with 
$j\in \mathsf{SessId}$.

Before presenting the definition of $\Rightarrow_I$, we introduce some notation. Let $\hist$ be a history that contains a representation of $\so$ as above. We use $\hist\oplus_j \tup{\tr,\iota_t,E,\po_t}$ to denote a history where $\tup{\tr,\iota_t,E,\po_t}$ is appended to $\so(j)$. 
Also, for an event $e$, $\hist\oplus_j e$ is the history obtained from $\hist$ by adding $e$ to the last transaction log in $\so(j)$ and as a last event in the program order of this log (i.e., if $\so(j)=\sigma; \tup{t,\iota_t,E,\po_t}$, then the session order $\so'$ of $\hist\oplus_j e$ is defined by $\so'(k)=\so(k)$ for all $k\neq j$ and $\so(j) =\sigma; \tup{t,\iota_t,E\cup{e},\po\cup \{(e',e): e'\in E\}}$). Finally, for a history $\hist = \tup{T, \so, \wro}$, $\hist\oplus\wro(\tr,e)$ is the history obtained from $\hist$ by adding $(\tr,e)$ to the write-read relation.

Figures \ref{fig:operational_semantics1}%
and \ref{fig:operational_semantics2} list the rules defining $\Rightarrow$. We distinguish between local computation rules (\textsc{if-true}, \textsc{if-false} and \textsc{local}) and database-accesses rules (\textsc{begin}, \textsc{insert}, \textsc{select}, \textsc{update}, \textsc{delete}, \textsc{commit} and \textsc{abort}); each associated to its homonymous instruction. Database-accesses get an increasing timestamp $\tau$ as well as an isolation-depending snapshot of the database using predicate $\snapshot$; updating adequately the timestamp and snapshot maps ($\vec{\mathsf{T}}$ and $\vec{\mathsf{S}}$ respectively).
Timestamps are used for validating the writes of a transaction and blocking inconsistent runs as well as for defining the set of possible snapshots any event can get. We use predicate $\selectOperational$ for determining the values read by an event. Those reads depend on both the event's snapshot as well as the timestamp of every previously executed event. Their formal definitions are described in Figure~\ref{fig:operational_aux}.

The \textsc{begin} rule starts a new transaction, provided that there is no other live transaction ($\mathsf{B}=\epsilon$) in the same session. It adds an empty transaction log to the history and schedules the body of the transaction. 
\textsc{if-true} and \textsc{if-false} check the truth value of a Boolean condition of an $\mathtt{if}$ conditional. \textsc{local} handles the case where some local computation is required. \textsc{insert}, \textsc{select}, \textsc{update} and \textsc{delete} handle the database accesses. \textsc{insert} add some rows $\setRows$ in the history. \textsc{select}, \textsc{update} and \textsc{delete} read every key from a combination of its snapshot and the local writes defined by $\selectOperational$ function. The predicate $\writeVar{\_}{\_}$ implicitly uses the previous information stored in the history via the function $\mathsf{value}_{\wro}$. Finally \textsc{commit} and \textsc{abort} validate that the run of the transaction correspond to the isolation level specification. These rules may block in case the validation is not satisfied as the predicate valuation does not change with the application of posterior rules.

An \emph{initial} configuration for program $\prog$ contains the program $\prog$ along with a history 
$\hist=\tup{\{\tr_0\},\emptyset,\emptyset}$, where $\tr_0$ is a transaction log containing only writes that write the initial values of all keys and whose timestamp and snapshot is $0$ ($\vec{\mathsf{S}}, \vec{\mathsf{T}} = [t_0 \mapsto 0])$, and it does not contain transaction code nor local keys ($\vec{\gamma}, \vec{\mathsf{B}} = \emptyset$).
A \emph{run} $\rho$ of a program $\prog$
is a sequence of configurations $c_0 c_1\ldots c_n$ where $c_0$ is an initial configuration for $\prog$, and $c_m\Rightarrow c_{m+1}$, for every $0\leq m < n$. We say that $c_n$ is \emph{reachable} from $c_0$.
The history of such a run, $\historyExecution{\rho}$, is the history $\hist_n$ in the last configuration $c_n$. 
A configuration is called \emph{final} if it contains the empty program ($\prog=\emptyset$).
Let $\histOf{\prog}$ denote the set of all histories of a run of $\prog$ that ends in a final configuration.

\input{figures-tex/select-commit-conditions}

The proof of Theorem~\ref{th:operational-semantics-full-and-consistent} is split in two parts: Lemma~\ref{lemma:operational-semantics-full} and Lemma~\ref{lemma:operational-semantics-consistent}. In Lemma~\ref{lemma:operational-semantics-full}, we prove by induction that for any run $\rho$, $\historyExecution{\rho}$ is a full history; using the auxiliary Lemma~\ref{lemma:operational-semantics-pending} about pending transactions. We then define in Equation~\ref{eq:operational-semantics-commit-order} a relation on transactions that plays the role of consistency witness for $\historyExecution{\rho}$. Then, we prove in Lemma~\ref{lemma:operational-semantics-total-order} that such relation is a commit order for $\historyExecution{\rho}$ to conclude in Lemma~\ref{lemma:operational-semantics-consistent} that $\historyExecution{\rho}$ is indeed consistent. In all cases, we do a case-by-case analysis depending on which rule is employed during the inductive step.

For the sake of simplifying our notation, we denote by $\labelRule{\rho, j, \rho'}$ to the rule s.t. applied to run $\rho$ on session $j$ leads to configuration $\rho'$. %

\begin{lemma}
\label{lemma:operational-semantics-pending}
Let $\rho$ be a run and $\historyExecution{\rho} = \tup{T, \so, \wro}$ be its history. Any pending transaction in $T$ is $(\so \cup \wro)$-maximal.
\end{lemma}

\begin{proof}
We prove by induction on the length of a run $\rho$ that any pending transaction is $(\so \cup \wro)$-maximal; where $\historyExecution{\rho} = \tup{T, \so, \wro}$. The base case, where $\rho = \{c_0\}$ and $c_0$ is an initial configuration, is immediate by definition. Let us suppose that for every run of length at most $n$ the property holds and let $\rho'$ a run of length $n+1$. As $\rho'$ is a sequence of configurations, there exist a reachable run $\rho$ of length $n$, a session $j$ and a rule $r$ s.t. $r = \labelRule{\rho, j, \rho'}$. Let us call $h = \tup{T, \so, \wro}$, $h' = \tup{T', \so', \wro'}$ and $e$ to $\historyExecution{\rho}$, $\historyExecution{\rho'}$ and the last event in $\po$-order belonging to $\last{h,j}$ respectively. By induction hypothesis, any pending transaction in $h$ is $(\so \cup \wro)$-maximal. To conclude the inductive step, we show that for every possible rule $r$ s.t. $r = \labelRule{\rho,j, \rho'}$, the property also holds in $h'$.

\begin{itemize}
    \item \textsc{local}, \textsc{if-false}, \textsc{if-true}, \textsc{insert}, \textsc{commit}, \textsc{abort}: The result trivially holds as $\wro' = \wro$, $\so' = \so$ and $\transNotPending{T'} \subseteq \transNotPending{T}$.

    \item \textsc{begin}: We observe that in this case, $T = T \cup \{\last{h, j}\}$, $\readOp{T'} = \readOp{T}$, $\wro' = \wro$ and $\so' = \so \cup \{\tup{t', \last{h, j}} \ | \ \sessionTransaction{t'} = j \}$. Thus, $\last{h, j}$ is pending and $\so' \cup \wro'$-maximal. Moreover, as described in Figure~\ref{fig:operational_semantics1}, $\vec{\mathsf{B}}(j) = \epsilon$; so there is no other transaction in session $j$ that is pending. Hence, as $T' \setminus \transNotPending{T'} = T \transNotPending{T} \cup \{\last{h, j}\}$, by induction hypothesis, every pending transaction is $\so' \cup \wro'$-maximal.
    
    \sloppy \item \textsc{select}, \textsc{update}, \textsc{delete}: \Cref{fig:operational_semantics2} describes $h'$ by the equation $h' = (h \oplus_j \tup{e, \labelRule{\rho, j, s}}) \bigoplus_{\key \in \Vars, \vec{\mathsf{w}}[x] \neq \bot} \wro(\vec{\mathsf{w}}[\key], e)$; where $e$ is the new event executed and $\vec{\mathsf{w}}$ is defined following the descriptions in Figures~\ref{fig:operational_semantics2} and \ref{fig:operational_aux}. 
	In this case, $T' = T, \readOp{T'} = \readOp{T} \cup \{e\}$, $\so' = \so$, $\forall \key \in \Vars$ s.t. $\vec{\mathsf{w}}[\key] = \bot$, $\wro'_\key = \wro_\key$ and $\forall \key \in \Vars$ s.t. $\vec{\mathsf{w}}[\key] \neq \bot$, $\wro'_\key = \wro_\key \cup  \{(\vec{\mathsf{w}}[\key], e)\}$. Note that as described by Figure~\ref{fig:operational_aux}, in the latter case, when $\vec{\mathsf{w}}[\key] \neq \bot$, $\trans{\vec{\mathsf{w}}[\key]} \in \transNotPending{T} = \transNotPending{T'}$. In conclusion, using the induction hypothesis, we also conclude that every pending transaction is $\so' \cup \wro'$-maximal.

\end{itemize}
\end{proof}
 
\begin{lemma}
\label{lemma:operational-semantics-full}
For every run $\rho$, $\historyExecution{\rho}$ is a full history.
\end{lemma}

\begin{proof}
We prove by induction on the length of a run $\rho$ that $\historyExecution{\rho}$ is a full history; where the base case, $\rho = \{c_0\}$ and $c_0$ is an initial configuration, is trivial by definition. Let us suppose that for every run of length at most $n$ the property holds and let $\rho'$ a run of length $n+1$. As $\rho'$ is a sequence of configurations, there exist a reachable run $\rho$ of length $n$, a session $j$ and a rule $r$ s.t. $r = \labelRule{\rho, j, \rho'}$. Let us call $h = \tup{T, \so, \wro}$, $h' = \tup{T', \so', \wro'}$ and $e$ to $\historyExecution{\rho}$, $\historyExecution{\rho'}$ and the last event in $\po$-order belonging to $\last{h,j}$ respectively. By induction hypothesis, $h$ is a full history. To conclude the inductive step, we show that for every possible rule $r$ s.t. $r = \labelRule{\rho,j, \rho'}$, the history $h'$ is also a full history. In particular, by definitions \ref{def:history} and \ref{def:full-history}, it suffices to prove that $\so' \cup \wro'$ is an acyclic relation and that for every variable $\key$ and read event $r$, ${\wro'_\key}^{-1}(r) \downarrow $ if and only if $r$ does not read $x$ from a local write and in such case, $\valuewr[\wro']{{\wro'_\key}^{-1}(r)}{\key} \neq \bot$. %

\begin{itemize}
	\item \textsc{local}, \textsc{if-false}, \textsc{if-true}, \textsc{insert}, \textsc{commit}, \textsc{abort}: The result trivially holds as $\readOp{T'} = \readOp{T}, \wro' = \wro$ and $\so' = \so$; using that $h$ is consistent.
	
	\item \textsc{begin}: We observe that $h' = h \oplus_j \tup{e,\ebegin}$, so $T = T \cup \{\last{h, j}\}$, $\readOp{T'} = \readOp{T}$, $\wro' = \wro$ and $\so' = \so \cup \{\tup{t', \last{h, j}} \ | \ \sessionTransaction{t'} = j \}$. In such case, by Lemma~\ref{lemma:operational-semantics-pending}, $t$ is $\so' \cup \wro'$-maximal. Thus, $\so' \cup \wro'$ is acyclic as $\so \cup \wro$ is also acyclic. Finally, as $\wro' = \wro$, we conclude 
	that $h'$ is a full history.
	
	\item \textsc{select}, \textsc{update}, \textsc{delete}: Here $h' = (h \oplus_j \tup{e, \labelRule{\rho, j, s}}) \bigoplus_{\key \in \Vars, \vec{\mathsf{w}}[x] \neq \bot} \wro(\vec{\mathsf{w}}[\key], e)$ where $e$ is the new event executed and $\vec{\mathsf{w}}$ is defined following the descriptions in Figures~\ref{fig:operational_semantics2} and \ref{fig:operational_aux}. In this case, $T' = T, \readOp{T'} = \readOp{T} \cup \{e\}$, $\so' = \so$, $\forall \key \in \Vars$ s.t. $\vec{\mathsf{w}}[\key] = \bot$, $\wro'_\key = \wro_\key$ and $\forall \key \in \Vars$ s.t. $\vec{\mathsf{w}}[\key] \neq \bot$, $\wro'_\key = \wro_\key \cup  \{(\vec{\mathsf{w}}[\key], e)\}$. Note that as the timestamp of any event is always positive and $\vec{\mathsf{T}}(\init) = 0$; for any key $\key$, $\mathsf{w}[x] \neq \bot$ if and only if $\localWrite[\key] = \bot$. Thus, $\vec{\mathsf{w}}$ is well defined, and ${\wro'_\key}^{-1}(r) \downarrow $ if and only $\localWrite[x] = \bot$. In such case, as any event $w$ writes on a key $x$ if and only if $\valuewr[\wro]{w}{\key} \neq \bot$, we conclude that $\valuewr[\wro']{{\wro'_\key}^{-1}(r)}{\key} \neq \bot$. To conclude the result, we need to show that $\so' \cup \wro'$ is acyclic.
	As $\rho$ is reachable, by Figure~\ref{fig:operational_aux}'s definition we know that for any event $r$ and key $\key$, if $\wro_\key^{-1}(r) \downarrow$, $\trans{\wro_\key^{-1}(r)} \in \transC{h}$. Thus, by Lemma~\ref{lemma:operational-semantics-pending}, $\last{h, j}$ is $\so' \cup \wro'$-maximal as it is not committed. Therefore, by the definition of $\so'$ and $\wro'$, as $\so \cup \wro$ is acyclic and $\last{h, j}$ is $\so' \cup \wro'$-maximal, $\so' \cup \wro'$ is also acyclic. In conclusion, $h'$ is a full history.

\end{itemize}
\end{proof}

Once proven that for any run $\rho$, $\historyExecution{\rho}$ is a full history, we need to prove that there exists a commit order $\co_\rho$ that witnesses $\historyExecution{\rho}$ consistency. Equation~\ref{eq:operational-semantics-commit-order} defines a relation that we prove in Lemma~\ref{lemma:operational-semantics-total-order} that it is a total order for $\historyExecution{\rho}$.

\begin{equation}
(t, t') \in \co_{\rho} \iff \left\{\begin{array}{l}
    t \in \transNotPending{T} \land t' \in \transNotPending{T} \ \land \\
	\qquad  \vec{\mathsf{T}}(\eend(t)) < \vec{\mathsf{T}}(\eend(t'))\\
    t \in \transNotPending{T} \land t' \not\in \transNotPending{T} \\
    t \not\in \transNotPending{T} \land t' \not\in \transNotPending{T} \ \land \\
	\qquad \vec{\mathsf{T}}(\ebegin(t)) < \vec{\mathsf{T}}(\ebegin(t'))\\
\end{array} \right.
\label{eq:operational-semantics-commit-order}
\end{equation}

\begin{lemma}
\label{lemma:operational-semantics-total-order}
For every run $\rho$, the relation $\co_\rho$ defined above is a commit order for $\historyExecution{\rho}$.
\end{lemma}

\begin{proof}
We prove by induction on the length of a run $\rho$ that the relation $\co_\rho$ defined by the equation below is a commit order for $\historyExecution{\rho}$, i.e., if $\historyExecution{\rho} = \tup{T, \so, \wro}$, then $\so \cup \wro \subseteq \co_\rho$.

The base case, where $\rho$ is composed only by an initial configuration is immediate as in such case $\wro = \emptyset$. Let us suppose that for every run of length at most $n$ the property holds and let $\rho'$ a run of length $n+1$. As $\rho'$ is a sequence of configurations, there exist a reachable run $\rho$ of length $n$, a session $j$ and a rule $\mathsf{r}$ s.t. $\mathsf{r} = \labelRule{\rho, j, \rho'}$. Let us call $h = \tup{T, \so, \wro}$, $h' = \tup{T', \so', \wro'}$ and $e$ to $\historyExecution{\rho}$, $\historyExecution{\rho'}$ and the last event in $\po$-order belonging to $\last{h,j}$ respectively. By induction hypothesis, $\co_\rho$ is a commit order for $h$. To conclude the inductive step, we show that $\co_{\rho'}$ is also a commit order for $h'$.

\begin{itemize}
    \item \textsc{local}, \textsc{if-false}, \textsc{if-true}: As $h = h'$ and $\vec{\mathsf{T}}_{\rho'} =\vec{\mathsf{T}}_\rho$, $\co_{\rho'} = \co_{\rho}$. Thus, the result trivially holds.
    
    \item \textsc{begin}: In this case, $e = \ebegin(\last{h, j})$ and $\last{h,j} \not\in \transNotPending{T_{\rho'}}$. Note that for any event $e' \neq e$, $\vec{\mathsf{T}}(e) > \vec{\mathsf{T}}(e')$ and $\transNotPending{T_\rho'} = \transNotPending{\rho}$. Thus, $\co_{\rho'} = \co_{\rho} \cup \{(t', \last{h', j}) \ | \ t' \in T\}$. As $\so \cup \wro \subseteq \co_\rho$, $\wro' = \wro$ and $\so' = \so \cup \{(t', \last{h, j}) \ | \ \sessionTransaction{t'} = j \}$, $\so' \cup \wro' \subseteq \co_{\rho'}$; so $\co_{\rho'}$ is a commit order for $h'$. 
    
    \item \textsc{insert}: In this case, as $\transNotPending{T_\rho'} = \transNotPending{T_\rho}$, $\co_{\rho'} = \co_{\rho}$. Hence, as $\so' = \so$ and $\wro' = \wro$, $\so' \cup \wro' \subseteq \co_{\rho'}$.
    
    \item \textsc{select}, \textsc{update}, \textsc{delete}: Once again, as $\transNotPending{T_\rho'} = \transNotPending{T_\rho}$, $\co_{\rho'} = \co_{\rho}$. Note that $\so' = \so$, $\forall \key \in \Vars$ s.t. $\vec{\mathsf{w}}[\key] = \bot$, $\wro'_\key = \wro_\key$ and $\forall \key \in \Vars$ s.t. $\vec{\mathsf{w}}[\key] \neq \bot$, $\wro'_\key = \wro_\key \cup \{(\vec{\mathsf{w}}[\key], e)\}$. In the latter case, where $\vec{\mathsf{w}}[\key] \neq \bot$, we know that $\trans{\vec{\mathsf{w}}[\key]} \in \transNotPending{T}$ thanks to the definitions on Figure~\ref{fig:operational_aux}. By Equation~\ref{eq:operational-semantics-commit-order}, as $\last{h, j}$ is pending, we deduce that $(\trans{\vec{\mathsf{w}}[\key]}, \trans{e}) \in \co_{\rho'}$. Therefore, as $\so \cup \wro \subseteq \co_{\rho} = \co_{\rho'}$, we conclude that $\so' \cup \wro' \subseteq \co_{\rho'}$.
     
    \item \textsc{commit}, \textsc{abort}: In this case, $e = \eend{\last{h, j}}$, $\co_{\rho'}\restriction_{T \setminus \{\last{h, j}\} \times T \setminus \{\last{h, j}\}} = \co_{\rho}\restriction_{T \setminus \{\last{h, j}\} \times T \setminus \{\last{h, j}\}}$, $\so' = \so$ and $\wro' = \wro$. Thus, to prove that $\so' \cup \wro' \subseteq \co_{\rho'}$ we only need to discuss about $\last{h,j}$. By Lemma~\ref{lemma:operational-semantics-pending}, $\last{h,j}$ is $\so' \cup \wro'$-maximal. Hence, we focus on proving that for any transaction $t'$ s.t. $(t', \last{h, j}) \in \so' \cup \wro'$, $(t', \last{h, j}) \in \co_{\rho'}$. Any such transaction $t'$ must be completed by Lemma~\ref{lemma:operational-semantics-pending}. However, by the definition on Figure~\ref{fig:operational_semantics1}, we know that $\vec{\mathsf{T}}(e) > \vec{\mathsf{T}}(\eend(t'))$, so $(t', \last{h, j}) \in \co_{\rho'}$ by Equation~\ref{eq:operational-semantics-commit-order}. Thus, $\so' \cup \wro' \subseteq \co_{\rho'}$.
\end{itemize}
\end{proof}

\begin{lemma}
\label{lemma:operational-semantics-consistent}
For every total run $\rho$, the $\historyExecution{\rho}$ is consistent.
\end{lemma}

\begin{proof}
Let $\rho^T$ be a total run. By Lemma~\ref{lemma:operational-semantics-full}, $\historyExecution{\rho^T}$ is a full history. Thus, to prove that $\historyExecution{\rho}$ is consistent, by Definition~\ref{def:consistency-full}, we need to show that there exists a commit order $\co$ that witnesses its consistency. We prove by induction on the length of a prefix $\rho$ of a total run $\rho^T$ that the relation $\co_\rho$ defined in Equation~\ref{eq:operational-semantics-commit-order} is a commit order that witnesses $\historyExecution{\rho}$'s consistency. Note that by Lemma~\ref{lemma:operational-semantics-total-order}, the relation $\co_\rho$ is indeed a commit order.

The base case, where $\rho$ is composed only by an initial configuration is immediate as in such case $\wro = \emptyset$. Let us suppose that for every run of length at most $n$ the property holds and let $\rho'$ a run of length $n+1$. As $\rho'$ is a sequence of configurations, there exist a reachable run $\rho$ of length $n$, a session $j$ and a rule $\mathsf{r}$ s.t. $\mathsf{r} = \labelRule{\rho, j, \rho'}$. Let us call $h = \tup{T, \so, \wro}$, $h' = \tup{T', \so', \wro'}$ and $e$ to $\historyExecution{\rho}$, $\historyExecution{\rho'}$ and the last event in $\po$-order belonging to $\last{h,j}$ respectively. By induction hypothesis, $\co_\rho$ is a commit order that witnesses $h$'s consistency. To conclude the inductive step, we show that for every possible rule $\mathsf{r}$ s.t. $\mathsf{r} = \labelRule{\rho,j, \rho'}$, $\co_{\rho'}$ is a commit order witnessing $h'$'s consistency.

By contradiction, let suppose that $\co_{\rho'}$ does not witness $h'$'s consistency. Then, there exists a variable $\key$, a read event $r$, an axiom $a \in \iota$ and two committed transactions $t_1, t_2$ s.t. $(t_1, e) \in \wro_x$, $\writeVar{t_2}{\key}$, $\visibilityRelation{a}{\co_{\rho'}}{t_2}{r}{x}$ holds in $h'$ but $(t_1, t_2) \in \co_{\rho'}$; where $\iota = \vec{\mathsf{I}}(\ebegin(e))$. Thus, if we prove that such dependencies can be seen in $h$ using $\co_\rho$, we obtain a contradiction as $\co_\rho$ witnesses $h$'s consistency. Note that as shown during the proof of Lemma~\ref{lemma:operational-semantics-total-order}, $\co_{\rho'}\restriction_{T \setminus \{\last{h, j}\} \times T \setminus \{\last{h, j}\}} = \co_{\rho}\restriction_{T \setminus \{\last{h, j}\} \times T \setminus \{\last{h, j}\}}$; so we simply prove that $\last{h, j}$ cannot be $t_1$, $t_2$, $\trans{r}$ or any intermediate transaction causing $\visibilityRelation{a}{\co_{\rho'}}{t_2}{r}{x}$ to hold in $h'$.

\begin{itemize}
    \item \textsc{local}, \textsc{if-false}, \textsc{if-true}: As $h = h'$ and $\co_{\rho'} = \co_{\rho}$, this case is impossible.

    \item \textsc{begin}: In this case, $\co_{\rho'} = \co_{\rho} \cup \{(t', \last{h', j}) \ | \ t' \in T\}$. By Lemma~\ref{lemma:operational-semantics-total-order}, $\last{h', j}$ is $(\so' \cup \wro')$-maximal, so $\last{h', j} \neq t_1$. Moreover, $\readOp{\last{h', j}} = \emptyset$, so $r \neq \readOp{\last{h', j}} $. In addition, $\last{h, j} \neq t_2$ as $\writeOp{\last{h, j}} = \emptyset$. 
    \begin{itemize}
        \item \underline{$a = \axser, \axpre$ or $\axrc$}: In all cases, the axioms do not relate any other transactions besides $t_1, t_2$ and $\trans{r}$, so this case is impossible.
        \item \underline{$a = \axconf$}: In this case, $\last{h, j} \neq t_4$ as it is $\co_{\rho'}$-maximal; so this case is also impossible.
    \end{itemize}
    
    \item \textsc{insert}: In this case, $\co_{\rho'} = \co_{\rho}$. Moreover, $\readOp{T'} = \readOp{T}$, $\writeOp{T'} = \writeOp{T}$, $\so' = \so$ and $\wro' = \wro$. Thus, this case is also impossible.
    
    \item \textsc{select}, \textsc{update}, \textsc{delete}: In this case, $\co_{\rho'} = \co_{\rho}$, $\so' = \so$, $\forall \key \in \Vars$ s.t. $\vec{\mathsf{w}}[\key] = \bot$, $\wro'_\key = \wro_\key$ and $\forall \key \in \Vars$ s.t. $\vec{\mathsf{w}}[\key] \neq \bot$, $\wro'_\key = \wro_\key \cup \{(\vec{\mathsf{w}}[\key], e)\}$. 
    As $\last{h, j}$ is pending, by Lemma~\ref{lemma:operational-semantics-total-order}, $\last{h, j} \neq t_1$ as it is $(\so' \cup \wro')$-maximal. Moreover, as $\writeOp{\last{h, j}} = \emptyset$, $\last{h, j} \neq t_2$. Then, we analyze if $\last{h, j}$ can be $\trans{r}$ (and thus, $r = e$) or any intermediate transaction. Note that for all three isolation levels we study, $\selectOperational$ returns the value written by the transaction with the last commit timestamp for a given snapshot time. Hence, as $(t_1, r) \in \wro_x$ and $(t_2, \trans{r}) \in \co_{\rho}$, we deduce that $\vec{\mathsf{T}}_\rho(\ecommit(t_2)) > \vec{\mathsf{T}}_\rho(\ebegin(\last{h, j}))$. We continue the analysis distinguishing between one case per axiom:

    \begin{itemize}
        
        \item \sloppy \underline{$a = \axser$}: As $\rho'$ is a prefix of a total run $\rho^T$, there exists runs $\hat{\rho}, \hat{\rho}'$ s.t. $\labelRule{\hat{\rho}, j', \hat{\rho}'}$ is either \textsc{commit} or \textsc{abort} and both a prefix of $\rho^T$; where $j'$ is the session of $\trans{r}$. Without loss of generality, we can assume that $\hat{\rho}$ and $\hat{\rho}'$ have minimal size; so $\last{\historyExecution{\hat{\rho}}, j'} = \trans{r}$. As $\rho^T$ is total and $\hat{\rho}'$ is a prefix of $\rho^T$, $\commitOperational[\iota](\historyExecution{\hat{\rho}, \vec{\mathsf{T}}_{\hat{\rho}'}, \trans{r}})$ holds.
        
        By the monotonicity of $\vec{\mathsf{T}}$, $\vec{\mathsf{T}}_{\rho'} \subseteq \vec{\mathsf{T}}_{\hat{\rho}'}$. Hence, as $(t_1, r) \in \wro_x$ and 
        $\vec{\mathsf{T}}_{\hat{\rho}'}(\ecommit(t_1)) < \vec{\mathsf{T}}_{\hat{\rho}'}(\ecommit(t_2))$,
        by the definitions of Figure~\ref{fig:operational_semantics2} and Figure~\ref{fig:operational_aux} we deduce that $\vec{\mathsf{T}}_{\hat{\rho}'}(\ebegin(\trans{r})) < \vec{\mathsf{T}}_{\hat{\rho}'}(\ecommit(t_2))$.
        However, as $\vec{\mathsf{T}}_{\hat{\rho}'}(\ebegin(\trans{r})) < \vec{\mathsf{T}}_{\hat{\rho}'}(\ecommit(t_2)) < \vec{\mathsf{T}}_{\hat{\rho}'}(\eend(\trans{r}))$, $\readVar{\trans{r}}{\key}$, $\writeVar{t_2}{\key}$; we conclude that $\commitOperational[\SER](\historyExecution{\hat{\rho}'}, \vec{\mathsf{T}}_{\hat{\rho}'}, \trans{r})$ does not hold; so this case is impossible.
        
        \item \sloppy \underline{$a = \axconf$}: In this case, $\last{h, j}$ cannot be an intermediate transaction nor $\trans{r}$ as $\writeOp{\last{h, j}} = \emptyset$; so this case is also impossible.
        
        \item \sloppy \underline{$a = \axpre$}: In this case, $\last{h, j}$ cannot be an intermediate transaction as by Lemma~\ref{lemma:operational-semantics-pending}, $\last{h, j}$ is $\so' \cup \wro'$-maximal. Thus, $\last{h, j}$ must be $\trans{r}$ and $e = r$. Therefore, there exists a transaction $t_4$ s.t. $(t_2, t_4) \in {\co_{\rho'}}^*$ and $(t_4, \last{h, j}) \in (\so' \cup \wro')$. Note that $t_4$ must be committed and that $\vec{\mathsf{T}}_{\rho'}(\ecommit(t_4)) < \vec{\mathsf{T}}_{\rho'}(\ebegin(\last{h, j}))$. Hence, as $(t_2, t_4) \in  {\co_{\rho'}}^*$ and $(t_1, t_2) \in \co_{\rho'}$ and they are both committed, we deduce that $\vec{\mathsf{T}}_{\rho'}(\ecommit(t_2)) < \vec{\mathsf{T}}_{\rho'}(\ecommit(t_4)) < \vec{\mathsf{T}}_{\rho'}(\ebegin(\last{h, j}))$. However, this contradicts that $\vec{\mathsf{T}}_{\rho'}(\ecommit(t_2)) > \vec{\mathsf{T}}_{\rho'}(\ebegin(\last{h, j}))$ Thus, this case is impossible.

        \item \sloppy \underline{$a = \axrc$}: In this case, $\last{h, j}$ must be $\trans{r}$ and in particular, $e = r$. As depicted on Figure~\ref{fig:operational_semantics2} and Figure~\ref{fig:operational_aux}, as $(t_1, r) \in \wro_x$, $\vec{\mathsf{S}}_{\rho'}(e) \leq \vec{\mathsf{T}}_{\rho'}(t_1)$. However, as $(t_1, t_2) \in \co_{\rho'}$, $\vec{\mathsf{T}}_{\rho'}(\ecommit(t_1)) < \vec{\mathsf{T}}_{\rho'}(\ecommit(t_2))$. Hence, as $(t_2, e) \in (\so \cup \wro);\po^*$, there exists an event $e' \in \last{h,j}$ s.t. $(e, e') \in \po^*$ and $\vec{\mathsf{T}}_{\rho'}(\ecommit(t_2)) < \vec{\mathsf{T}}_{\rho'}(e')$. However, by $\snapshot[\RC]$'s definition, $\vec{\mathsf{S}}(e') \leq \vec{\mathsf{S}}_{\rho'}(e)$; so we deduce that $\vec{\mathsf{T}}_{\rho'}(\ecommit(t_1)) <\vec{\mathsf{T}}_{\rho'}(\ecommit(t_2)) < \vec{\mathsf{S}}_{\rho'}(e)$. This contradicts the definition of $\selectOperational$; so this case is impossible.

    \end{itemize}

    \item \textsc{commit}, \textsc{abort}: In this case, $\co_{\rho'} \restriction (T \setminus \{\last{h, j}\} \times T \setminus \{\last{h, j}\}) = \co_{\rho} \restriction (T \setminus \{\last{h, j}\} \times T \setminus \{\last{h, j}\})$, $\so' = \so$, $\wro' = \wro$. First, using that by induction hypothesis any prefix $\tilde{\rho}$ of $\rho$ is consistent using $\co_{\tilde{\rho}}$; we define $\tilde{\rho}$ the prefix of $\rho$ that introduces the read event $r$. As $\historyExecution{\tilde{\rho}} = \tup{\tilde{T}, \tilde{\so}, \tilde{\wro}}$ is consistent and $(t_1, r) \in \tilde{\wro}_\key$; $t_1$ is committed. Hence, by the definitions of $\selectOperational$ and $\snapshot[\iota]$ on Figure~\ref{fig:operational_aux} and the rules semantics on Figure~\ref{fig:operational_semantics2}, we deduce that $\vec{\mathsf{T}}_\rho(\ecommit(t_1)) > \vec{\mathsf{T}}_\rho(\ebegin(t_2))$. Next, as $\last{h, j}$ is pending in $h$, it is $\so \cup \wro$-maximal. Therefore, it is also $\so' \cup \wro'$-maximal; so it cannot play the role of $t_1$. However, it can play the role of $t_2$, $\last{h, j}$ or the role of an intermediate transaction. Let us analyze case by case depending on the axiom:
    
    \begin{itemize}
        \item \underline{$a = \axser$}: Two sub-cases arise:
        \begin{itemize}
            \item \sloppy\underline{$\last{h, j} = t_2$}: I this case, $\writeVar{t_2}{\key}$ must hold. As $\rho'$ is a prefix of a total run $\rho^T$, there exists runs $\hat{\rho}, \hat{\rho}'$ s.t. $\labelRule{\hat{\rho}, j', \hat{\rho}'}$ is either \textsc{commit} or \textsc{abort} and both a prefix of $\rho^T$; where $j'$ is the session of $\trans{r}$. Without loss of generality, we can assume that $\hat{\rho}$ and $\hat{\rho}'$ have minimal size; so $\last{\historyExecution{\hat{\rho}}, j'} = \trans{r}$. As $\rho^T$ is total and $\hat{\rho}'$ is a prefix of $\rho^T$, $\commitOperational[\iota](\historyExecution{\hat{\rho}, \vec{\mathsf{T}}_{\hat{\rho}'}, \trans{r}})$ holds. Note that as $(t_1, t_2) \in \co_{\rho'}$ and they are both committed, $\vec{\mathsf{T}}_{\hat{\rho}'}(\ecommit(t_1)) < \vec{\mathsf{T}}_{\hat{\rho}'}(\ecommit(t_2))$. However, $\readVar{\trans{r}}{\key}$, $\writeVar{t_2}{\key}$ and $\vec{\mathsf{T}}_{\hat{\rho}'}(\ebegin(\trans{r})) < \vec{\mathsf{T}}_{\hat{\rho}'}(\ecommit(\last{h, j})) < \vec{\mathsf{T}}_{\hat{\rho}'}(\ecommit(\trans{r}))$; which contradicts that $\commitOperational[\iota](\historyExecution{\hat{\rho}, \vec{\mathsf{T}}_{\hat{\rho}'}, \trans{r}})$ holds. In conclusion, this case is impossible.

            \item \underline{$\last{h, j} =  \trans{r}$}: In such case, as $t_1$ and $t_2$ are committed, $(t_2, \last{h, j}) \in \co_{\rho}$ and $(t_1, t_2) \in \co_{\rho}$. Hence, this case is also impossible as $\co_{\rho}$ witnesses that $h$ is consistent.
            
        \end{itemize}

        \item \underline{$a = \axpre$}: In this case, there exists a transaction $t_4$ s.t. $(t_2, t_4) \in \co_{\rho'}^*$ and $(t_4, \trans{r}) \in \so' \cup \wro'$. As $\last{h, j}$ is pending in $h$, by Lemma~\ref{lemma:operational-semantics-pending}, $(\so \cup \wro)$-maximal. Thus, as $\so' = \so$ and $\wro' = \wro$, $t_4 \neq \last{h, j}$. Moreover, as $(t_2, t_4) \in {\co_{\rho'}}^*$, $t_4$ is committed and $\last{h, j} \neq t_4$ is the $\co_{\rho'}$-maximal transaction that is committed, $t_2 \neq \last{h, j}$. Hence, $\last{h, j} = \trans{r}$. However, as $\so' = \so$, $\wro' = \wro'$ and $\co_{\rho'}\restriction_{T \setminus \{\last{h, j}\} \times T \setminus \{\last{h, j}\}} = \co_{\rho}\restriction_{T \setminus \{\last{h, j}\} \times T \setminus \{\last{h, j}\}}$; we conclude that $(t_1, t_2) \in \co_{\rho}$, $(t_2, t_4) \in {\co_{\rho}}^*$ and $(t_4, \last{h,j}) \in \so \cup \wro$; which contradicts that $\co_{\rho}$ witnesses $h$'s consistency, so this case is impossible.
        
        \item \underline{$a = \axconf$}: In this case, there exists a variable $y$ and a transaction $t_4$ s.t. $\writeVar{t_4}{y}$, $\writeVar{\trans{r}}{y}$ $(t_2, t_4) \in \co_{\rho'}^*$, $(t_4, \trans{r}) \in \co_{\rho'}$. As $\last{h, j}$ is the $\co_{\rho'}$-maximal transaction that is committed,  $(t_2, \trans{r}), (t_4, \trans{r}) \in \co_{\rho'}$ and $\writeOp{\trans{r}} \neq \emptyset$, we deduce that $\last{h, j} \neq t_2, t_4$. Hence, $\last{h, j}$ must be $\trans{r}$ and $e = \ecommit(\last{h, j})$. On one hand, we observe that as $(t_4, \last{h, j}) \in \co_{\rho'}$ and they are both committed, $\vec{\mathsf{T}}_{\rho'}(\ecommit(t_4)) < \vec{\mathsf{T}}_{\rho'}(e)$. On the other hand, as $(t_2, t_4) \in \co_{\rho'}^*$ and $\vec{\mathsf{T}}_{\rho'}(\ebegin(\trans{r})) < \vec{\mathsf{T}}_{\rho'}(\ecommit(t_2))$; we conclude that $\vec{\mathsf{T}}_{\rho'}(\ebegin(\trans{r})) < \vec{\mathsf{T}}_{\rho'}(\ecommit(t_4))$. In conclusion, we obtain that $\commitOperational[\SI](h',\vec{\mathsf{T}}_{\rho'}, \last{h, j})$ does not hold due to the existence of $t_4$; which contradict the hypothesis, so this case is impossible.        
        
        \item \underline{$a = \axrc$}: In this case, $r \neq e$ as $r$ is a read event and $e$ is not, and $(t_2, r) \in (\so' \cup \wro'); {\po'}^*$. Hence, as $\so' = \so, \wro' = \wro$ and $\po' = \po \cup \{(e', e) \ | \ e' \in \last{h, j}\}$; $(t_2, r) \in (\so \cup \wro); \po^*$. Finally, as $\last{h, j}$ is pending in $h$, $\last{h, j} \neq t_2$. Thus, as $\co_{\rho'}\restriction_{T \setminus \{\last{h, j}\} \times T \setminus \{\last{h, j}\}} = \co_{\rho}\restriction_{T \setminus \{\last{h, j}\} \times T \setminus \{\last{h, j}\}}$; we deduce that $(t_1, t_2) \in \co_\rho$. However, this contradicts that $\co_\rho$ witnesses $h$'s consistency; so this case is also impossible.

    \end{itemize}

    As every possible case is impossible, we deduce that the hypothesis, $\co_{\rho'}$ does not witnesses $h'$'s consistency is false; so we conclude the proof of the inductive step.
\end{itemize}
\end{proof}

%% file: figures-tex/operational-semantics.tex
\begin{figure}[!ht]
\footnotesize
\centering
    \begin{mathpar}
    \inferrule[begin]{\tr \mbox{ fresh}\quad e \mbox{ fresh}\quad 
    \mathsf{P}(j) = \ibegin{\iota}; \mathsf{Body}; \icommit; \mathsf{P} \quad \vec{\mathsf{B}}(j) = \epsilon  \\ %
    \tau = 1 + \max\{\vec{\mathsf{T}}(e') \ | \ e' \in \events{\hist}\} \quad  \vec{\mathsf{T}}' = \vec{\mathsf{T}}[e \to \tau] \quad \delta = \snapshot(\hist, \vec{\mathsf{S}}, \vec{\mathsf{T}}', e, \ebegin) \\
    \hist' = \hist \oplus_j \tup{\tr,\iota,\{\tup{e,\ebegin}\},\emptyset}}{
        \hist,\vec{\gamma},\vec{\mathsf{B}},\vec{\mathsf{I}},\vec{\mathsf{T}}, \vec{\mathsf{S}}, \mathsf{P}
        \Rightarrow
        \hist',\vec{\gamma}[j\mapsto \emptyset],\vec{\mathsf{B}}[j\mapsto \mathsf{Body}; \icommit],\vec{\mathsf{I}}[t\mapsto \iota],\vec{\mathsf{T}}',\vec{\mathsf{S}}[e \mapsto \delta],\mathsf{P}[j\mapsto \mathsf{S}]
    } 

    \inferrule[if-true]{\psi(\vec{a})[\vec{\gamma}(j)(a) / a: a\in\vec{a}] \quad
    \vec{\mathsf{B}}(j) = \iif{\psi(\vec{a})}{\mathsf{Instr}};\mathsf{B}
    }{
        \hist,\vec{\gamma},\vec{\mathsf{B}}, \vec{\mathsf{I}},\vec{\mathsf{T}},\vec{\mathsf{S}},\mathsf{P}
        \Rightarrow
        \hist,\vec{\gamma},\vec{\mathsf{B}}[j\mapsto \mathsf{Instr};\mathsf{B}],\vec{\mathsf{I}},\vec{\mathsf{T}},\vec{\mathsf{S}},\mathsf{P}
    } 

    \inferrule[if-false]{\lnot \psi(\vec{a})[\vec{\gamma}(j)(a) / a: a\in\vec{a}]\quad
    \vec{\mathsf{B}}(j) = \iif{\psi(\vec{x})}{\mathsf{Instr}};\mathsf{B}
    }{
        \hist,\vec{\gamma},\vec{\mathsf{B}}, \vec{\mathsf{I}},\vec{\mathsf{T}},\vec{\mathsf{S}},\mathsf{P}
        \Rightarrow
        \hist,\vec{\gamma},\vec{\mathsf{B}}[j\mapsto \mathsf{B}],\vec{\mathsf{I}},\vec{\mathsf{T}},\vec{\mathsf{S}},\mathsf{P}
    } 

    \inferrule[local]{v = e[\vec{\gamma}(j)(a') / a': a' \in \vec{a'}] \quad \vec{\mathsf{B}}(j) = a := e(\vec{a'});\mathsf{B}
    }{
        \hist,\vec{\gamma},\vec{\mathsf{B}}, \vec{\mathsf{I}},\vec{\mathsf{T}},\vec{\mathsf{S}},\mathsf{P}
        \Rightarrow
        \hist,\vec{\gamma}[j,a\mapsto v],\vec{\mathsf{B}}[j\mapsto \mathsf{B}],\vec{\mathsf{I}},\vec{\mathsf{T}},\vec{\mathsf{S}},\mathsf{P}
    }

    \inferrule[commit]{e\mbox{ fresh} \quad t = \last{h, j} \quad \iota = \isolation{h}(t) \quad
    \vec{\mathsf{B}}(j) = \icommit \\
        \tau = 1 + \max\{\vec{\mathsf{T}}(e') \ | \ e' \in \events{\hist}\} \quad \vec{\mathsf{T}}' = \vec{\mathsf{T}}[e \to \tau] 
    \\
    \delta = \snapshot(\hist, \vec{\mathsf{S}}, \vec{\mathsf{T}}', e, \ecommit)  \quad \commitOperational(h,\vec{\mathsf{T}}', t)
    }{
        \hist,\vec{\gamma},\vec{\mathsf{B}}, \vec{\mathsf{I}},\vec{\mathsf{T}},\mathsf{P}
        \Rightarrow
        \hist \oplus_j \tup{e,\icommit},\vec{\gamma},\vec{\mathsf{B}}[j\mapsto \epsilon], \vec{\mathsf{I}},\vec{\mathsf{T}}',\vec{\mathsf{S}}[e \mapsto \delta],\mathsf{P}
    } 

    \inferrule[abort]{e\mbox{ fresh} \quad t = \last{h, j} \quad \iota = \isolation{h}(t) \quad \vec{\mathsf{B}}(j) = \iabort; B \\
        \tau = 1 + \max\{\vec{\mathsf{T}}(e') \ | \ e' \in \events{\hist}\} \quad \vec{\mathsf{T}}' = \vec{\mathsf{T}}[e \to \tau]   \\
    \delta = \snapshot(\hist, \vec{\mathsf{S}}, \vec{\mathsf{T}}', e, \eabort) \quad \commitOperational(h,\vec{\mathsf{T}}', t)%
    }{
        \hist,\vec{\gamma},\vec{\mathsf{B}}, \vec{\mathsf{I}},\vec{\mathsf{T}},\vec{\mathsf{S}},\mathsf{P} 
        \Rightarrow
        \hist \oplus_j \tup{e,\iabort},\vec{\gamma},\vec{\mathsf{B}}[j\mapsto \epsilon], \vec{\mathsf{I}},\vec{\mathsf{T}}',\vec{\mathsf{S}}[e \mapsto \delta],\mathsf{P}
    }
    \end{mathpar}
\caption{
    An operational semantics for transactional programs. Above, $\last{h,j}$ denotes the last transaction log in the session order $\so(j)$ of $h$ while $\snapshot$ and $\selectOperational$ denote the snapshot visible to an instruction and the writes it reads from, respectively. The $\commitOperational{}$ checks if a transaction can be committed. They are defined in Figure~\ref{fig:operational_aux}.
    }.
    \label{fig:operational_semantics1}
\end{figure}

%% file: figures-tex/operational-semantics-2.tex
\begin{figure}[!ht]
\footnotesize
\centering
    \begin{mathpar}
    \inferrule[insert]{
        e\mbox{ fresh} \quad t = \last{h, j}\quad \iota = \isolation{h}(t) \quad \vec{\mathsf{B}}(j) = \iinsert{\setRows};\mathsf{B} \\
            \tau = 1 + \max\{\vec{\mathsf{T}}(e') \ | \ e' \in \events{\hist}\}\quad \vec{\mathsf{T}}' = \vec{\mathsf{T}}[e \to \tau] \\
            \delta = \snapshot(\hist, \vec{\mathsf{S}}, \vec{\mathsf{T}}', e, \einsert) \quad \hist' = \hist \oplus_j \tup{e,\iinsert{\setRows}}
    }{
        \hist,\vec{\gamma},\vec{\mathsf{B}}, \vec{\mathsf{I}},\vec{\mathsf{T}},\vec{\mathsf{S}},\mathsf{P}
        \Rightarrow
        \hist',\vec{\gamma},\vec{\mathsf{B}}[j\mapsto \mathsf{B}], \vec{\mathsf{I}},\vec{\mathsf{T}}',\vec{\mathsf{S}}[e \mapsto \delta],\mathsf{P}
    } 

    \inferrule[select]{
        e\mbox{ fresh } \quad t = \last{h, j}\quad \iota = \isolation{h}(t) \quad \vec{\mathsf{B}}(j) = a := \iselect{\predicate};\mathsf{B} \\ 
        \tau = 1 + \max\{\vec{\mathsf{T}}(e') \ | \ e' \in \events{\hist}\} \quad \vec{\mathsf{T}}' = \vec{\mathsf{T}}[e \to \tau] \\
        \delta = \snapshot(\hist, \vec{\mathsf{S}}, \vec{\mathsf{T}}', e, \eselect)\quad \vec{\mathsf{w}} =\selectOperational(h, \vec{\mathsf{T}}, t, \delta) \\ \hist' = (\hist \oplus_j \tup{e, \iselect{\predicate}}) \bigoplus_{\key \in \Vars, \vec{\mathsf{w}}[x] \neq \bot} \wro(\vec{\mathsf{w}}[\key], e)
    }{
        \hist,\vec{\gamma},\vec{\mathsf{B}}, \vec{\mathsf{I}},\vec{\mathsf{T}},\vec{\mathsf{S}},\mathsf{P}
        \Rightarrow
        \hist',\vec{\gamma}[(j,a)\mapsto \{\row\in\delta:\predicate(\row)\}],\vec{\mathsf{B}}[j\mapsto \mathsf{B}],\vec{\mathsf{I}},\vec{\mathsf{T}}',\vec{\mathsf{S}}[e \mapsto \delta],\mathsf{P}
    } 

    \inferrule[update]{
        e\mbox{ fresh } \quad t = \last{h, j}\quad \iota = \isolation{h}(t)  \quad \vec{\mathsf{B}}(j) = \iupdate{\predicate}{\mapRows};\mathsf{B} \\
            \tau = 1 + \max\{\vec{\mathsf{T}}(e') \ | \ e' \in \events{\hist}\} \quad
        \vec{\mathsf{T}}' = \vec{\mathsf{T}}[e \to \tau] \\
        \delta = \snapshot(\hist, \vec{\mathsf{S}}, \vec{\mathsf{T}}', e, \eupdate) \quad
        \vec{\mathsf{w}} =\selectOperational(h, \vec{\mathsf{T}}, t, \delta)  \\ 
        \hist' = (\hist \oplus_j \tup{e, \iupdate{\predicate}{\mapRows}}) \bigoplus_{\key \in \Vars, \vec{\mathsf{w}}[x] \neq \bot} \wro(\vec{\mathsf{w}}[\key], e)
    }{
        \hist,\vec{\gamma},\vec{\mathsf{B}}, \vec{\mathsf{I}},\vec{\mathsf{T}},\vec{\mathsf{S}},\mathsf{P}
        \Rightarrow
        \hist', \vec{\gamma},\vec{\mathsf{B}}[j\mapsto \mathsf{B}],\vec{\mathsf{I}},\vec{\mathsf{T}}',\vec{\mathsf{S}}[e \mapsto \delta],\mathsf{P} 
    } 

    \inferrule[delete]{
        e\mbox{ fresh } \quad t = \last{h, j}\quad \iota = \isolation{h}(t) \quad
        \vec{\mathsf{B}}(j) = \idelete{\predicate};\mathsf{B} \\
            \tau = 1 + \max\{\vec{\mathsf{T}}(e') \ | \ e' \in \events{\hist}\} \quad
        \vec{\mathsf{T}}' = \vec{\mathsf{T}}[e \to \tau] \\
        \delta = \snapshot(\hist, \vec{\mathsf{S}}, \vec{\mathsf{T}}', e, \edelete) \quad
        \vec{\mathsf{w}} =\selectOperational(h, \vec{\mathsf{T}}, t, \delta)  \\
            \hist' = (\hist \oplus_j \tup{e, \idelete{\predicate}}) \bigoplus_{\key \in \Vars, \vec{\mathsf{w}}[x] \neq \bot} \wro(\vec{\mathsf{w}}[\key], e)
    }{
        \hist,\vec{\gamma},\vec{\mathsf{B}}, \vec{\mathsf{I}},\vec{\mathsf{T}},\vec{\mathsf{S}},\mathsf{P}
        \Rightarrow
        \hist',\vec{\gamma},\vec{\mathsf{B}}[j\mapsto \mathsf{B}],\vec{\mathsf{I}},\vec{\mathsf{T}}',\vec{\mathsf{S}}[e \mapsto \delta],\mathsf{P}
    }
    \end{mathpar}
        \caption{
    An operational semantics for transactional programs. Above, $\last{h,j}$ denotes the last transaction log in the session order $\so(j)$ of $h$ while $\snapshot$ and $\selectOperational$ denote the snapshot visible to an instruction and the writes it reads from, respectively. The $\commitOperational{}$ checks if a transaction can be committed. They are defined in Figure~\ref{fig:operational_aux}.
    }.
    \label{fig:operational_semantics2}
\end{figure}

%% file: figures-tex/select-commit-conditions.tex
\begin{figure} [ht]
    \small
        \centering
        \begin{equation*}
        \begin{array}{ll}
            \snapshot[\SER](\hist, \vec{\mathsf{S}}, \vec{\mathsf{T}}', e, \xi) = & \left\{ 
                \begin{array}{ll}
                    \max\left\{\vec{\mathsf{T}}'(c_{t'})\ \left| \
                        \begin{array}{l}
                            t' \in h \ \land \\
                            c_{t'} = \ecommit(t') \ \land\\[1mm]
                            \vec{\mathsf{T}}'(c_{t'}) < \vec{\mathsf{T}}'(e)
                        \end{array}
                        \right.\right\} & \text{if } \xi = \ebegin \\[5mm]
                    \vec{\mathsf{S}}(\ebegin(\trans{e})) & \text{otherwise}
                \end{array} \right. \\[8mm]
                
            \snapshot[\SI](\hist, \vec{\mathsf{S}}, \vec{\mathsf{T}}', e, \xi) = & \left\{ 
                \begin{array}{ll}
                    \choice{ \left\{\vec{\mathsf{T}}'(c_{t'})\ \left| \
                        \begin{array}{l}
                            t' \in h \ \land \\[1mm]
                            c_{t'} = \ecommit(t') \ \land \\[1mm]
                            vec{\mathsf{T}}'(c_{t'}) < \vec{\mathsf{T}}'(e)
                        \end{array}
                        \right.\right\}} & \text{if } \xi = \ebegin \\[5mm]
                        \vec{\mathsf{S}}(\ebegin(\trans{e})) & \text{otherwise}
                \end{array} \right. \\[8mm]
    
            \snapshot[\RC](\hist, \vec{\mathsf{S}}, \vec{\mathsf{T}}', e, \xi) = & \choice{\left\{\vec{\mathsf{T}}'(c_{t'}) \left| 
                \begin{array}{l}
                    t' \in h \land \\[1mm]
                    c_{t'} = \ecommit(t') 
                    \land \vec{\mathsf{T}}'(c_{t'}) < \vec{\mathsf{T}}'(e) \land \\[1mm]
                    \forall e'. \left(\begin{array}{l}
                        (e', e) \in \po \ \lor \\
                        (\trans{e'}, \trans{e}) \in \so 
                    \end{array} \right) \\[1mm] 
                    \hspace{1.5cm}\implies \vec{\mathsf{S}}(e') \leq \vec{\mathsf{T}}(c_{t'})
                \end{array}
                \right.\right\}} 
        \end{array}
        \end{equation*}

        \begin{equation*}
        \begin{array}{rl}
            \selectOperational(\hist, \vec{\mathsf{T}}, t, \delta) = & [x \mapsto (\localWrite[x] \neq \bot)\,?\, \bot\ :\ w_x \text{ for each } x \in \Vars] \\[2mm]
            \text{ where } \localWrite[x] = & \max_{\po}\{e \ | \ \trans{e} = t \ \land \ \writeVar{e}{x}\} \cup \{\bot\} \\[2mm]
            \text{ and }  \writeVar{w_x}{x} \ \land &
                \vec{\mathsf{T}}(w_x) = \max\left\{\vec{\mathsf{T}}(w') \ \left|  \ {\begin{array}{c} 
                w' \in \events{h} \ \land \ \writeVar{w'}{x} \ \land \\
                    \vec{\mathsf{T}}(\ecommit(\trans{w'})) \leq \delta
            \end{array}} \right. \right\}\\
        \end{array}
        \end{equation*}
        
        \begin{equation*}
        \begin{array}{ll}
            \commitOperational[\SER](\hist, \vec{\mathsf{T}}', t) = & \left(
                \begin{array}{ll}
                    \not\exists t' \in h, \key \in \Vars \text{ s.t. } (\readVar{t}{\key} \lor \writeVar{t}{\key}) \ \land \ \writeVar{t'}{\key}  \\[1mm]
                    \qquad \land\ \vec{\mathsf{T}}'(\ebegin(t)) < \vec{\mathsf{T}}'(\ecommit(t')) <  \vec{\mathsf{T}}'(\eend(t))
                \end{array} \right) \\[6mm]
                
            \commitOperational[\SI](\hist, \vec{\mathsf{T}}', t) = & \left(
                \begin{array}{ll}
                    \not\exists t' \in h, \key \in \Vars \text{ s.t. } \writeVar{t}{\key} \ \land \ \writeVar{t'}{\key} \ \land \\[1mm]
                    \qquad \land\ \vec{\mathsf{T}}'(\ebegin(t)) < \vec{\mathsf{T}}'(\ecommit(t'))  <  \vec{\mathsf{T}}'(\eend(t))
                \end{array} \right) \\[6mm]
    
            \commitOperational[\RC](\hist, \vec{\mathsf{T}}', t) = & \btrue \\
        \end{array}
        \end{equation*}
        
        \caption{Definition of auxiliary functions for the operational semantics. The function $\mathsf{choice}$ receives a set as input and returns one of its elements.
        }.
        \label{fig:operational_aux}
    \end{figure}

%% file: appendices/proof-theorems.tex
\section{Proofs of \Cref{th:poly-consistency-saturable,th:extend-np-complete,th:k-expressive-np,th:csob}}

\input{appendices/newproofs/proofs-full-histories}

\newpage

\input{appendices/newproofs/proofs-bounded-complexity}

\newpage

\input{appendices/newproofs/proofs-partial-observation}

\newpage

\input{appendices/newproofs/proofs-algorithm}

\newpage

\input{appendices/newproofs/proofs-algorithm-complexity}

\newpage

%% file: appendices/newproofs/proofs-full-histories.tex
\subsection{Complexity analysis of \Cref{algorithm:necessary-co,algorithm:checking-saturable} (Proof of~\Cref{th:poly-consistency-saturable})}
\label{app:proofs-poly-consistency-saturable}

For a given history $h$, \Cref{algorithm:necessary-co} computes necessary and sufficient conditions for an execution of $h$ $\exec = \tup{h, \co}$ to be consistent. It computes a bigger relation $\pco_\result$ that includes $\co$ and any other dependency between transactions that can be deduced from the isolation configuration. \Cref{algorithm:necessary-co} decides if $\co$ is a commit order witnessing consistency of the history (\Cref{lemma:necessary-co:correctness}) and it runs in polynomial time (\Cref{lemma:necessary-co:polynomial-time}).

\begin{lemma}
\label{lemma:necessary-co:correctness}
For any full history $h = \tup{T, \so, \wro}$, the execution $\exec = \tup{h, \co}$ is consistent if and only if $\pco_\result = \computeCO{h}{\co}$ is acyclic.
\end{lemma}

\begin{proof}
Let $h = \tup{T, \so, \wro}$ be a history, $\exec = \tup{h, \co}$ be an execution of $h$ and $\pco_\result = \computeCO{h}{\co}$ be the relation obtained thanks to \Cref{algorithm:necessary-co}.

\underline{$\implies$} Let us suppose that $\exec$ is consistent. As $\co$ is acyclic, it suffice to prove that $\pco_\result = \co$. By contradiction, let us suppose that $\pco_\result \neq \co$. As $\co \subseteq \pco_\result$ (line~\ref{algorithm:consistent-extension:co}), there exists $t_1, t_2$ s.t. $(t_2, t_1) \in \pco_\result \setminus \co$. In such case, such tuple must be added in line~\ref{algorithm:necessary-co:add-tuple-co}. Hence, there exists $x \in \Vars, e \in \readOp{h}$ and $v \in \visibilitySet{\isolation{h}(\trans{r})}$ s.t. $t_1 = \wro_x^{-1}(r)$ and $\visibilityRelationAxiomCo{t_2}{r}{x}$ holds in $h$. As $\exec$ is consistent, $(t_2, t_1) \in \co$; which is impossible. Hence, $\pco_\result = \co$.

\underline{$\impliedby$}
Let us suppose that $\pco_\result$ is acyclic. By contradiction, let us suppose that $\exec$ is not consistent. Then, there exists an read event $r$ s.t. $\constraintAxiom[\isolation{h}(\trans{r})]{\co}{r}$ does not hold. Hence, by \Cref{eq:axiom}, there exists $v \in \visibilitySet{\isolation{h}(\trans{r})}$, $x \in \Vars$, $t_2 \in T$ s.t. $\visibilityRelationInstanceApplied{t_2}{r}{x}$ hold in $h$ but $(t_2, t_1) \not\in \co$; where $t_1 = \wro_x^{-1}(r)$. In such case, \Cref{algorithm:necessary-co} ensures in line~\ref{algorithm:necessary-co:add-tuple-co} that $(t_2, t_1) \in \pco_\result$. However, as $\co \subseteq \pco_\result$ (line~\ref{algorithm:consistent-extension:co}), $\co$ is a total order and $\pco_\result$ is acyclic, $\co = \pco_\result$. Thus, $(t_2, t_1) \in \co$; which is impossible. Thus, $\exec$ is consistent.
\end{proof}

\begin{lemma}
\label{lemma:evaluate-v:polynomial-time}
Let $h = \tup{T, \so, \wro}$ be a history s.t. $\isolation{h}$ is bounded by $k \in \mathbb{N}$, $x \in \Vars$ be a key, $t \in T$ be a transaction, $r$ be a read event, $\pco \subseteq T \times T$ be a partial order and $v$ be a visibility relation in $\visibilitySet{\isolation{h}(\trans{r})}$. Evaluating $\visibilityRelationInstanceApplied[\pco]{t}{r}{x}$ is in $\mathcal{O}(|h|^{k-2})$.
\end{lemma}

\begin{proof}
As $\isolation{h}$ is bounded, there exists $k \in \mathbb{N}$ s.t. $|\visibilitySet{\isolation{h}(t)}| \leq k$. Hence, the number of quantifiers employed by a visibility relation is at most $k$ (and at least $3$ according to Equation~\ref{eq:axiom}). In addition, for each $\mathsf{v} \in \visibilitySet{\isolation{h}(t)}$ evaluating each condition $\visibilityRelationInstanceApplied[\pco]{t}{r}{x}$ can be modelled with an algorithm that employ $k-3$ nested loops, one per existential quantifier employed by $\mathsf{v}$, and that for each quantifier assignment evaluates the quantifier-free part of the formula. 

First, we observe that as $\confWr$ predicate only query information about the $k-1$ quantified events, the size of such sub-formula is in $\mathcal{O}(k)$. Next, we notice that as $\whereName$ predicate can be evaluated in constant time, for every key $x$ and event $w$, computing $\valuewr{x}{w}$ is in $\mathcal{O}(k \cdot T)$. Hence, as $k$ is constant, evaluating the quantifier-free formula of $v$ is in $\mathcal{O}(|h|)$ and thus, evaluating $\visibilityRelationInstanceApplied[\pco]{t}{r}{x}$ is in $\mathcal{O}(|h|^{k-3} \cdot |h|) = \mathcal{O}(|h|^{k-2})$.
\end{proof}

\begin{lemma}
\label{lemma:necessary-co:polynomial-time}
Let $h = \tup{T, \so, \wro}$ be a full history, $k$ be a bound on $\isolation{h}$ and $\pco \subseteq T \times T$ be a partial order. Algorithm~\ref{algorithm:necessary-co} runs in $\mathcal{O}(|h|^{k+1})$.
\end{lemma}

\begin{proof}
Let $h = \tup{T, \so, \wro}$ be a full history. \Cref{algorithm:necessary-co} can be decomposed in two blocks: lines \ref{algorithm:necessary-co:init-loop}-\ref{algorithm:necessary-co:add-tuple-co} and lines \ref{algorithm:necessary-co:all-constraints}-\ref{algorithm:necessary-co:add-tuple-co}. Hence, the cost of Algorithm~\ref{algorithm:necessary-co} is in $\mathcal{O}(|\Vars| \cdot |\events{h}| \cdot |T| \cdot U)$; where $U$ is an upper-bound of the cost of evaluating lines \ref{algorithm:necessary-co:all-constraints}-\ref{algorithm:necessary-co:add-tuple-co}. On one hand, both $|\Vars|$, $|\events{h}|$ and $|T|$ are in $\mathcal{O}(|h|)$. On the other hand, as $\isolation{h}$ is bounded by $k$, by \Cref{lemma:evaluate-v:polynomial-time}, $U \in \mathcal{O}(|h|^{k-2})$. Altogether, we deduce that \Cref{algorithm:necessary-co} runs in $\mathcal{O}(|h|^{k+1})$.
\end{proof}

Algorithm~\ref{algorithm:checking-saturable} generalizes the results for $\RA$ and $\RC$ in~\cite{DBLP:journals/pacmpl/BiswasE19} %
for full histories with heterogeneous saturable isolation configurations; proving that such histories can be checked in polynomial time.
\polyConsistencySaturable*

We split the proof of Theorem~\ref{th:poly-consistency-saturable} in two Lemmas: Lemma~\ref{lemma:poly-consistency-saturable:correctness} that proves the correctness of Algorithm~\ref{algorithm:checking-saturable} and Lemma~\ref{lemma:poly-consistency-saturable:polynomial-time} that ensures its polynomial-time behavior.

\begin{lemma}
\label{lemma:poly-consistency-saturable:correctness}
For every full history $h = \tup{T, \so, \wro}$ whose isolation configuration is saturable, Algorithm~\ref{algorithm:checking-saturable} returns $\btrue$ if and only if $h$ is consistent. 
\end{lemma}

\begin{proof}
Let $h = \tup{T, \so, \wro}$ a full history whose isolation configuration is saturable and let $\pco$ be the visibility relation defined in line~\ref{algorithm:checking-saturable:compute-co} in \Cref{algorithm:checking-saturable}.

On one hand, let suppose that $h$ is consistent and let $\exec = \tup{h, \co}$ be a consistent execution of $h$. If we show that $\pco \subseteq \co$, we can conclude that Algorithm~\ref{algorithm:checking-saturable} returns $\btrue$ as $\co$ is acyclic.
Let $(t_2, t_1) \in \pco$ and let us prove that $(t_2, t_1) \in \co$. As $\so \cup \wro \subseteq \co$, by the definition of commit order, we can assume that $(t_2, t_1) \in \pco \setminus (\so \cup \wro)$. In such case, there must exists $\key \in \Vars, e \in \readOp{h}$ and $\mathsf{v} \in \visibilitySet{\isolation{h}(\trans{e})}$ s.t. $\writeVar{t_2}{\key}$ and $\visibilityRelationInstanceApplied[(\so \cup \wro)^+]{t_2}{e}{x}$ holds.
As $\isolation{h}(\trans{e})$ is saturable, $\visibilityRelationInstanceApplied[(\so \cup \wro)^+]{t_2}{e}{x}$ holds. Hence, as $\co$ is a commit order and $(\so \cup \wro)^+ \subseteq \co$; $\visibilityRelationInstanceApplied[\co]{t_2}{e}{x}$ also holds. Therefore, as $\co$ witnesses $h$'s consistency, we deduce that $(t_2, t_1) \in \co$.

On the other hand, let us suppose that Algorithm~\ref{algorithm:checking-saturable} returns $\btrue$. Then, $\pco$ must be acyclic by the condition in line \ref{algorithm:checking-saturable:co-acyclic}. Therefore, as $\pco$ is acyclic it can be extended to a total order $\co$. Let us prove that the execution $\exec = \tup{h, \co}$ is consistent. Let $\key \in \Vars, t_2 \in T, e \in \readOp{h}$ and $\mathsf{v} \in \visibilitySet{\isolation{h}(\trans{e})}$ s.t. $\writeVar{t_2}{\key}$ and $\visibilityRelationInstanceApplied[\co]{t_2}{e}{x}$ holds. As Algorithm~\ref{algorithm:checking-saturable} returns $\btrue$, we deduce that Algorithm~\ref{algorithm:necessary-co} checks the condition at line \ref{algorithm:necessary-co:constraint-evaluation}. As $\isolation{h}(\trans{e})$ is saturable, $\visibilityRelationInstanceApplied[(\so \cup \wro)^+]{t_2}{e}{x}$ also holds. Thus, $(t_2, t_1) \in \pco$. As $\pco \subseteq \co$, $(t_2, t_1) \in \co$; so $\co$ witnesses $h$'s consistency.

\end{proof}

\begin{lemma}
\label{lemma:poly-consistency-saturable:polynomial-time}
For every full history $h$ whose isolation configuration is bounded, Algorithm~\ref{algorithm:checking-saturable} runs in polynomial time with respect $\mathcal{O}(|h|)$.
\end{lemma}

\begin{proof}
Let $h = \tup{T, \so, \wro}$ be a full history whose isolation configuration is saturable.
First, we observe that checking if a graph $G = (V, E)$ is acyclic can be easily done with a DFS in $\mathcal{O}(|V| + |E|)$. Thus, the cost of checking acyclicity of both $\so \cup \wro$ (line~\ref{algorithm:checking-saturable:so-wr-acyclic}) and $\pco$ (line~\ref{algorithm:checking-saturable:co-acyclic}) is in $\mathcal{O}(|T| + |T|^2) = \mathcal{O}(|T|^2) \subseteq \mathcal{O}(|h|^2)$. Furthermore, by Lemma~\ref{lemma:necessary-co:polynomial-time}, the cost of executing Algorithm~\ref{algorithm:necessary-co} is in $\mathcal{O}(|h|^{k+1})$; where $k$ is a bound in $\isolation{h}$. Thus, checking $h$'s consistency with Algorithm~\ref{algorithm:checking-saturable} can be done in polynomial time.
\end{proof}

%% file: appendices/newproofs/proofs-bounded-complexity.tex
\subsection{Proof of \Cref{th:k-expressive-np}}
\label{app:proof-bounded-complexity}

\kExpressiveNP*

The proof of \Cref{th:k-expressive-np} is structured in two parts: proving that the problem is in NP and proving that is NP-hard. The first part corresponds to \Cref{lemma:k-expressive:problem-np}; which is analogous as the proof of Lemma~\ref{lemma:np-complete:problem-in-np}. The second part, based on a reduction to $1$-in-$3$ SAT problem, corresponds to \Cref{lemma:k-expressive:poly-size,lemma:k-expressive:sat-consistent,lemma:k-expressive:consistent-sat}.

\begin{lemma}
\label{lemma:k-expressive:problem-np}
The problem of checking consistency for a bounded width client history $h$ with an isolation configuration stronger than $\RC$ and $\widthHistory{h} \geq 3$ is in NP.   
\end{lemma}

\begin{proof}
Let $h = \tup{T, \so, \wro}$ a client history whose isolation configuration is stronger than $\RC$. Guessing a witness of $h$, $\overline{h}$, and an execution of $\overline{h}$, $\exec = \tup{\overline{h}, \co}$, can be done in $\mathcal{O}(|\Vars| \cdot |\events{h}|^2 + |T|^2) \subseteq \mathcal{O}(|h|^3)$. By \Cref{lemma:necessary-co:correctness}, checking if $\exec$ is consistent is equivalent as checking if $\computeCO{h'}{\co}$ is an acyclic relation. As by \Cref{lemma:necessary-co:polynomial-time}, \Cref{algorithm:necessary-co} requires polynomial time, we conclude the result.
\end{proof}

For showing NP-hardness, we will reduce $1$-in-$3$ SAT to checking consistency. Let $\varphi$ be a boolean formula with $n$ clauses and $m$ variables of the form $\varphi = \bigwedge\limits_{i = 1}^n (v_i^0 \lor v_i^1 \lor v_i^2)$; we construct a history $h_\varphi$ s.t. $h_\varphi$ is consistent if and only if $\varphi$ is satisfiable with exactly only one variable assigned the value $\btrue$. The key idea is designing a history with width $3$ that is stratified in \emph{rounds}, one per clause. In each round, three transactions, one per variable in the clause, ``compete'' to be first in the commit order. The one that precedes the other two correspond to the variable in $\varphi$ that is satisfied. %

First, we define the round $0$ corresponding to the variables of $\varphi$. For every variable $x_i \in \var{\varphi}, 1 \leq i \leq m$ we define an homonymous key $x_i$ that represents such variable. Doing an abuse of notation, we say that $x_i \in \var{\varphi}$. Then, we create two transactions $1_i$ and $0_i$ associated to the two states of $x_i$, $1$ and $0$. The former contains the event $\iinsert{\{x_i : 1, 1_i : 1\}}$ while the latter $\iinsert{\{x_i : 0, 0_i : 1\}}$. Both $1_i$ and $0_i$ write also on a special key named $1_i$ and $0_i$ respectively to indicate on the database that they have committed. 

Next, we define rounds $1-n$ representing each clause in $\varphi$. For each clause $C_i \coloneqq (v_i^0 \lor v_i^1 \lor v_i^2), 1 \leq i \leq n$, we define the round $i$. Round $i$ is composed by three transactions: $t_i^0$, $t_i^1$ and $t_i^2$, representing the choice of the variable among $v_i^0, v_i^1$ and $v_i^2$ that is selected in the clause $C_i$. Transactions $t_i^j$ write on keys $v_i^j$ and $v_i^{j+1\bmod 3}$ to preserve the structure of the clause $C_i$, as well on the special homonymous key $t_i^j$ to indicate that such transaction has been executed; in a similar way as we did in the round $0$. For that, we impose that transactions $t_i^j$ are composed of an event $\iselect{\lambda x: \mathtt{eq}(x, v_i^j, v_i^{j+1 \bmod 3}, v_i^{j+2 \bmod 3})}$ followed by an event $\iinsert{\{v_i^j : 0, v_i^{j+1 \bmod 3} : 1, t_i^j : -1\}}$.

The function $\mathtt{eq} : \Vals \times \Vars^3 \to \{\btrue, \bfalse\}$ is described in Equation~\ref{eq:k-expressive:eq-def} and assumes that $\Vals$ contains two distinct values $0$ and $1$ and that there is a predicate $\mathtt{val}: \Vals \to \{0,1\}$ that returns the value of a variable in the database. Intuitively, for any key $r$, if $a, b, c$ correspond to the three variables in a clause $C_i$ (possibly permuted), whenever $\lnot \mathtt{eq}(r, a, b, c)$ holds, we deduce that the value assigned at key $a$ is $1$ while on the other two keys the assigned value is $0$. Moreover, whenever $r$ refers to any of the special keys such as $0_i, 1_i$ or $t_i^j$, the predicate $\mathtt{eq}(r, a, b, c)$ always holds.

\begin{equation}
\label{eq:k-expressive:eq-def}
\mathtt{eq}(r, a, b, c) = \left\{
    \begin{array}{ll}
    \val{r} \neq 1 & \text{if } \keyof{r} = a \\
    \val{r} \neq 0 & \text{if } \keyof{r} = b \lor \keyof{r} = c \\
    \btrue & \text{if } \keyof{r} \in \{t_i^j \ | \ 1 \leq i \leq n, 0 \leq j \leq 2\} \\
    \btrue & \text{if } \keyof{r} \in \{1_i, 0_i \ | \ 1 \leq i \leq m\} \\
    \bfalse & \text{otherwise}
    \end{array}\right. 
\end{equation}

Finally, we add an initial transaction that writes on every key the value $1$. For that, we assume that $\Vars$ contains only one key per variable used in $\varphi$ as well as one key per aforementioned transaction. %
We denote by $T$ the set of all described transactions as well as by $\roundTransaction{t}$ to the round a transaction $t \in T$ belongs to.

We describe the session order in the history $h_\varphi$ using an auxiliary relation $\overline{\so}$. We establish that $(1_i, 1_j), (0_i, 0_j) \in \overline{\so}$ for any pair of indices $i, j, 1 \leq i < j \leq m$. We also enforce that $(t_i^j, t_{i+1}^j) \in \overline{\so}$, for every $1 \leq i \leq n, 0 \leq j, j' \leq 2$. Finally, we connect round $0$ with round $1$ by enforcing that $(1_m, t_1^0) \in \overline{\so}$ and $(0_m, t_1^1) \in \overline{\so}$. Then, we denote by $\so$ to the transitive closure of $\overline{\so}$. Note that $\so$ is a union of disjoint total orders, so it is acyclic. 

For describing the write-read relation, we distinguish between two cases: keys associated to variables in $\varphi$ or to a transaction in $T$. On one hand, for every key $x_i, 1 \leq i \leq m$, we define $\wro_{x_i} = \emptyset$. On the other hand, for every key $x$ associated to a transaction $t_x$ and every read event $r$ in a transaction $t$, we impose that $(t_x, r) \in \wro_x$ if $\roundTransaction{t_x} < \roundTransaction{t}$ while otherwise we declare that $(\init, r) \in \wro_x$. Then, we denote by $\wro = \bigcup_{x \in \Vars} \wro_x$ as well as by $h_\varphi$ to the tuple $h_\varphi = \tup{T, \so, \wro}$. A full depiction of $h_\varphi$ can be found in Figure~\ref{fig:k-expressive-np-code}.

We observe that imposing $\wro_x = \emptyset $ on every key $x \in \var{\varphi}$ ensures that, for any witness of $h_\varphi$, $\overline{h} = \tup{T, \so, \overline{\wro}}$, if $(w, r) \in \overline{\wro}$, then $\where{r}(\valuewr[\overline{\wro}]{w}{x}) = 0$. In particular, this implies that each transaction $t_i^j$ must read key $v_i^j$ from a transaction that writes $1$ as value while it also must read keys $v_i^{j+1 \bmod 3}$ and $v_i^{j+2\bmod 3}$ from a transaction that writes $0$ as value. Intuitively, this property shows that $\varphi$ is well-encoded in $h_\varphi$.

\input{figures-tex/k-expressive-np-code}

The proof is divided in four steps: Lemma~\ref{lemma:k-expressive:poly-size} proves that the $h_\varphi$ is a polynomial-size transformation of $\varphi$, Lemma~\ref{lemma:k-expressive:history} proves that the $h_\varphi$ is indeed a history and Lemmas~\ref{lemma:k-expressive:sat-consistent} and~\ref{lemma:k-expressive:consistent-sat} prove that $h_\varphi$ is consistent if and only if $\varphi$ is $1$-in-$3$ satisfiable.

\begin{lemma}
\label{lemma:k-expressive:poly-size}
$h_\varphi$ is a polynomial size transformation on the length of $\varphi$.
\end{lemma}

\begin{proof}
If $\varphi$ has $n$ clauses and $m$ variables, $h_\varphi$ employs $6n+2m+1$ transactions. As $m \leq 3n$, $|T| \in \mathcal{O}(n)$. The number of variables, $|\Vars| = m + |T|$, so $|\Vars| \in \mathcal{O}(n)$. As every transaction has at most two events, $|\events{h_\varphi}| \in \mathcal{O}(n)$. Moreover, $\wro \subseteq \Vars \times T \times T$ and $\so \subseteq T \times T$, so $|\wro| \in \mathcal{O}(n^3)$ and $|\so| \in \mathcal{O}(n^2)$. Thus, $h_\varphi$ is a polynomial transformation of $\varphi$.
\end{proof}

For proving that $h_\varphi$ is a history, by Definition~\ref{def:history} it suffices to prove that $\so \cup \wro$ is an acyclic relation. Indeed, by our choice of $\wro$, for every key $x$, $\wro_x^{-1}$ is a partial function that, whenever it is defined, associates reads to writes on $x$. Hence, from Lemma~\ref{lemma:k-expressive:history} we conclude that $h_\varphi$ is a history.

\begin{lemma}
\label{lemma:k-expressive:history}
The relation $\so \cup \wro$ is acyclic.
\end{lemma}

\begin{proof}
For proving that $\so \cup \wro$ is acyclic, we reason by induction on the number of clauses. In particular, we show that for every pair of transactions $t, t'$ if $\roundTransaction{t'} \leq i$ and $(t, t') \in \so \cup \wro$, then $\roundTransaction{t} \leq i$ and $(t', t) \not\in \so \cup \wro$. 

\begin{itemize}

    \item \underline{Base case:} The base case refers to round $0$; which contains $\init$ and transactions $0_j, 1_j, 1 \leq j \leq m$. We observe that transactions in round $0$ do not contain any read event. Hence, $(t, t') \in \so$. In such case, the result immediately holds by construction of $\so$.

    \item \underline{Inductive case:} Let us suppose that the induction hypothesis holds for every $1 \leq i \leq k \leq n$ and let us prove it also for $k+1 \leq n$. If $\roundTransaction{t'} < k +1$, $\roundTransaction{t'} \leq k$ and the result holds by induction hypothesis; so we can assume without loss of generality that $\roundTransaction{t'} = k+1$. By construction of both $\so$ and $\wro$, if $(t, t') \in \so \cup \wro$, $\roundTransaction{t} < \roundTransaction{t'}$. Hence, $\roundTransaction{t} \leq k$. By induction hypothesis on $t$, if $(t', t) \in \so \cup \wro$, $\roundTransaction{t'} \leq k  < k + 1 = \roundTransaction{t'}$; which is impossible. Thus, we conclude that $(t', t) \not\in \so \cup \wro$.

\end{itemize}
\end{proof}

\begin{lemma}
\label{lemma:k-expressive:sat-consistent}
If $\varphi$ is $1$-in-$3$ satisfiable then $h_\varphi$ is consistent.
\end{lemma}

\begin{proof}
Let $\alpha: \Vars \to \{0,1\}$ an assignment that makes $1$-in-$3$ satisfiable. To construct a witness of $h_\varphi$ we define a write-read relation $\overline{\wro}$ that extends $\wro$ and a total order on its transactions. For that, we first define a total order $\co$ between the transactions in $T$. In Equation~\ref{eq:k-expressive:eq-def} we define two auxiliary relations $\hat{\mathtt{r}}$ and $\hat{\mathtt{b}}$ based on $\alpha$ that totally orders the transactions that belongs to the same round. 

For every clause $C_i, 1 \leq i \leq n$ let $j_i$ be the unique index s.t. $\alpha(v_i^{j_i}) = 1$. Such index allow us to order the transactions in the round $i$: $t_i^{j_i}$ preceding $t_i^{j_i +1 \bmod 3}$ while $t_i^{j_i +1 \bmod 3}$ preceding $t_i^{j_i +2 \bmod 3}$. Intuitively, $t_i^{j_i}$ must precede the other two transactions in the total order as $v_i^j$ is the variable that is satisfied. Then, we connect every pair of consecutive rounds thanks to relation $\hat{\mathtt{c}}_1$.

For transactions in round $0$, we enforce that transactions associated to the same variable are totally ordered using $\alpha$. In particular, for every $i, 1 \leq i \leq m$, $0_i$ precedes $1_i$ in $\hat{\mathtt{b}}$ if and only if $\alpha(v_i) = 1$. Then, we connect every pair tuple in $\hat{\mathtt{b}}$ with relation $\hat{\mathtt{c}}_2$. Finally, we connect $\init$ with transactions in round $0$ as well as round $0$ with round $1$ thanks to relation $\hat{\mathtt{c}}_3$.

\begin{align}
\label{eq:k-expressive:auxiliary-relations}
\hat{\mathtt{r}} = \left\{\left.\begin{array}{l}
    (t_i^{j_i}, t_i^{j_i+1 \bmod 3})  \\
    (t_i^{j_i+1 \bmod 3},t_i^{j_i+2 \bmod 3})
\end{array}\right| 
\begin{array}{l}
    1 \leq i \leq n, 0 \leq j_i \leq 2  \\
    \alpha(v_i^{j_i}) = 1
\end{array}
\right\} \nonumber\\
\hat{\mathtt{b}} = \{(0_i, 1_i) \ | \ x_i \in \Vars \land \alpha(x_i) = 1\} \cup \{(1_i, 0_i) \ | \ x_i \in \Vars \land \alpha(x_i) = 0\} \nonumber\\
\hat{\mathtt{c}}_1 =  \{(t_i^{j_i + 2 \bmod 3}, t_{i+1}^{j_{i+1}}) \ | \ 1 \leq i < n, 0 \leq j_i, j_{i+1} \leq 2, \alpha(v_i^{j_i}) = 1 = \alpha(v_i^{j_{i+1}})\} \nonumber \\
\hat{\mathtt{c}}_2 =
\{(1_i, 0_j), (1_i, 1_j), (0_i, 0_j), (0_i, 1_j) \ | \ 1 \leq i < j \leq m\} \nonumber\\
\hat{\mathtt{c}}_3 = \{(\init, 0_1), (\init, 1_1)\} \cup \{(1_m, t_1^{j_1}), (0_m, t_1^{j_1}) \ | \ 0 \leq j_1 \leq 2, \alpha(v_i^{j_1}) = 1\}
\end{align}

Let $\co = (\hat{\mathtt{r}} \cup\hat{\mathtt{c}}_1 \cup \hat{\mathtt{b}} \cup \hat{\mathtt{c}}_2 \cup \hat{\mathtt{c}}_3)^+$. The proof of \Cref{lemma:k-expressive:sat-consistent} concludes thanks to \Cref{lemma:k-expressive:sat-cons:co-total-order,lemma:k-expressive:sat-cons:full-history}, \Cref{prop:constraint-co} and \Cref{corollary:ser-strongest}. First, \Cref{lemma:k-expressive:sat-cons:co-total-order} proves that the relation $\co$ is a total order between transactions. Then, \Cref{lemma:k-expressive:sat-cons:full-history} shows that $\co$ allow us define $\overline{h}$, a witness of $h_\varphi$. And finally, with the aid of \Cref{prop:constraint-co} and \Cref{corollary:ser-strongest} we conclude that $\overline{h}$ is consistent; so it is a consistent witness of $h_\varphi$.

\end{proof}

\begin{lemma}
\label{lemma:k-expressive:sat-cons:co-total-order}
The relation $\co$ is a total order.
\end{lemma}

\begin{proof}
For proving that $\co$ is a total order, we show by induction that if $(t, t') \in \co$ and $\roundTransaction{t'} \leq i$, then $\roundTransaction{t} \leq i$ and $(t', t) \not\in \co$.

\begin{itemize}
    \item \underline{Base case:} We observe that by construction of $\co$, $t' \neq \init$. We prove the base case by a second induction that if there exists $i', 1 \leq i' \leq m$ s.t. $t'\in \{0_{i'}, 1_{i'}\}$ and $(t, t') \in \co$ then either $t = \init$ or there exists $i \leq i'$ s.t. $t \in \{0_i, 1_i\}$ and $(t', t) \not\in \co$. 
    \begin{itemize}
        \item \underline{Base case:} Let us suppose that $\alpha(x_0) = 1$ as the other case is symmetric. If $t = \init$, $(t', t) \not\in \co$ as $\init$ is minimal in $\co$. If not, then $t' = 1_1$ and $t = 0_1$. We conclude once more that $(t, t') \not\in \co$ as $0_1$ only have $\init$ as a $\co$-predecessor; which is $\co$-minimal.
        
        \item \underline{Induction hypothesis:} Let us suppose that the induction hypothesis holds for every $1 \leq i \leq k \leq m$ and let us prove it also for $k+1 \leq m$. If $i' < k$, we conclude the result by induction hypothesis; so we can assume that $i' = k$. Moreover, as $\init$ is $\co$-minimal, we can assume without loss of generality that $t \neq \init$. Thus, by construction of $\co$, there must exists $i, 1 \leq i \leq m$ s.t. $t \in \{0_i, 1_i\}$. In particular, $i \leq i'$. Thus, if $i < i'$ and $(t', t)$ would be in $\co$, by induction hypothesis on $t$ we would deduce that $i' \leq i < i'$; which is impossible. Hence, we can assume that $i' = i$. Let us assume that $\alpha(x_i) = 1$ as the other case is symmetric. Thus, we deduce that $t =0_i$ and $t' = 1_i$. We observe that $(t', t) \not\in \hat{\mathtt{r}} \cup\hat{\mathtt{c}}_1 \cup \hat{\mathtt{b}} \cup \hat{\mathtt{c}}_2 \cup \hat{\mathtt{c}}_3$. As $T$ is finite, if $(t', t) \in \co$, there would exist a transaction $t'' \neq t'$ s.t $(t'', t)\in \hat{\mathtt{r}} \cup\hat{\mathtt{c}}_1 \cup \hat{\mathtt{b}} \cup \hat{\mathtt{c}}_2 \cup \hat{\mathtt{c}}_3$ and $(t', t'') \in \co$. But in such case, either $t'' = \init$ or there would exists an integer $i'', 1 \leq i'' < i \leq m$ s.t. $t'' \in \{0_{i''}, 1_{i''}\}$; which is impossible by induction hypothesis. In conclusion, $(t', t) \not\in \co$.
        
    \end{itemize}

    \item \underline{Inductive case:} Let us suppose that the induction hypothesis holds for every $1 \leq i \leq k \leq n$ and let us prove it also for $k+1 \leq n$. Let thus $t,t'$ transactions s.t. $\roundTransaction{t}' \leq k+1$ and $(t, t') \in \co$. If $\roundTransaction{t'} < k +1$, $\roundTransaction{t'} \leq k$ and the result holds by induction hypothesis; so we can assume without loss of generality that $\roundTransaction{t'} = k+1$. By construction of $\co$, $\roundTransaction{t} \leq k+1$. If $\roundTransaction{t} \leq k$ and $(t', t) \in \co$, by induction hypothesis on $t$ we obtain that $\roundTransaction{t'} \leq k < k+1 = \roundTransaction{t'}$; which is impossible. Thus, we can also assume without loss of generality that $\roundTransaction{t} = k+1$. In such case, we observe that $\hat{\mathtt{c}}_1$, $\hat{\mathtt{b}}$ and $\hat{\mathtt{c}}_1$ do not order transactions belonging to the same round. Hence if $(t, t') \in \co$ and $(t', t) \in \co$, we deduce that $(t, t') \in \hat{\mathtt{r}}$ and $(t', t) \in \hat{\mathtt{r}}$. However, by construction of $\hat{\mathtt{r}}$, this is impossible, so we conclude once more that $(t', t) \not\in \co$.
\end{itemize}
\end{proof}

Next, we construct a full history $\overline{h}$ using $\co$ that extends $h_\varphi$. For every key $x$ and $\iread$ event $r$, we define $w_x^r$ as follows:

\begin{equation}
w_x^r = \max_{\co} \{t \in T \ | \ \writeVar{t}{x} \land (t, r) \in \co\}
\end{equation}

Observe that $w_x^r$ is well-defined as $\co$ is a total order and $\init$ write every key. For each key $x \in \var{\varphi}$, we define the relation $\overline{\wro}_x = \{(w_x^r, r) \ | \ r \in \readOp{h}\}$. Then, we define the relation $\overline{\wro} = \bigcup_{x \in \var{\varphi}} \overline{\wro}_x \cup \wro$ as well as the history $\overline{h} = \tup{T, \so, \overline{\wro}}$. Lemma~\ref{lemma:k-expressive:sat-cons:full-history} proves that $\overline{h}$ is indeed a full history while Lemma~\ref{lemma:k-expressive:sat-cons:witness} shows that $\overline{h}$ is a witness of $h_\varphi$.

\begin{lemma}
\label{lemma:k-expressive:sat-cons:full-history}
$\overline{h}$ is a full history.
\end{lemma}
\begin{proof}

    For showing that $\overline{h}$ is a full history it suffices to show that $\so \cup \overline{\wro}$ is acyclic. As $\co$ is a total order and $\overline{\wro} \setminus \wro \subseteq \co$, proving that $\so \cup \wro \subseteq \co$ concludes the result. 
    First, we prove that $\so \subseteq \co$. Let $t, t'$ be transactions s.t. $(t, t') \in \so$. In such case, $\roundTransaction{t} \leq \roundTransaction{t'}$; and they only coincide if $\roundTransaction{t}  = \roundTransaction{t'} = 0$. Three cases arise:

    \begin{itemize}
        \item \underline{$\roundTransaction{t} = \roundTransaction{t'} = 0$:} As $(t, t') \in \hat{\mathtt{c}}_2$, we conclude that $(t, t') \in \co$.
        
        \sloppy \item \underline{$\roundTransaction{t}, \roundTransaction{t'} > 0$:} As $\roundTransaction{t}, \roundTransaction{t'} > 0$ and $\roundTransaction{t} \leq \roundTransaction{t'}$, by construction of $\so$ we deduce that $\roundTransaction{t} < \roundTransaction{t'}$. As $\co$ is transitive, we can assume without loss of generality that $\roundTransaction{t'} = \roundTransaction{t} + 1$. Therefore, there exists $i, j, 1 \leq i < n, 0 \leq j \leq 2$ s.t. $t = t_i^j$ and $t' = t_{i + 1}^j$. Let $j_i, j_{i+1}, 0 \leq j_i, j_{i+1} \leq 2$ be the integers s.t. $\alpha(v_i^{j_i}) = 1 = \alpha(v_{i+1}^{j_{i+1}})$. In such case, we know that $(t_i^j, t_i^{j_i+2 \bmod 3}) \in \hat{\mathtt{r}}^*$, $(t_i^{j_i+2 \bmod 3}, t_{i+1}^{j_i}) \in \hat{\mathtt{c}}_1$ and $(t_{i+1}^{j_{i+1}}, t_{i+1}^{j+1}) \in \hat{\mathtt{r}}^*$. Hence, as $\co$ is transitive, $(t, t') \in \co$.

        \item \underline{$\roundTransaction{t} = 0$, $\roundTransaction{t'} > 0$:} In this case, as $\roundTransaction{t} = 0$, there exists $i, 1 \leq i \leq m$ s.t. $x_i \in \var{\varphi}, t \in \{0_i, 1_i\}$. We assume without loss of generality that $t = 1_i$ as the other case is symmetric. In addition, as $\roundTransaction{t'} > 0$ and $(t, t') \in \so$, there exists $i, 1 \leq i \leq n$ s.t. $t' = t_i^0$. We rely on the two previous proven cases to deduce the result: as $(0_i, 0_m) \in \so \subseteq \co$, $(0_m, t_0^{j_0}) \in \hat{\mathtt{c}}_3$, $(t_0^{j_0}, t_0^0) \in \hat{\mathtt{r}}^*$ and $(t_0^0, t_i^0) \in \so \subseteq \co$, we conclude that $(t, t') \in \co$.
    \end{itemize}

    Next, we prove that $\wro \subseteq \co$. Let $r$ be a read event and $w$ be a write event s.t. $(w, r) \in \wro$. Then, there exists $i, i', 1 \leq i < i' \leq n$ and $j, j', 0 \leq j, j' \leq 2$ s.t. $w = t_i^j$ and $\trans{r} = t_{i'}^{j'}$. Let $j_{i'-1},j_{i'}, 0 \leq j_{i'-1},j_{i'} \leq 2$ be the integers s.t. $\alpha(v_{i'-1}^{j_{i'-1}}) = 1 = \alpha(v_{i'}^{j'})$. In such case, we know that $(t_i^j, t_{i'-1}^j) \in \so^*$, $(t_{i'-1}^j, t_{i'-1}^{j_{i'}+2 \bmod 3}) \in \hat{\mathtt{r}}^*$, $(t_{i'-1}^{j_{i'}+2 \bmod 3}, t_{i'}^{j_i}) \in \hat{\mathtt{c}}_1$ and $(t_{i'}^{j_{i'}}, t_{i'}^{j'}) \in \hat{\mathtt{r}}^*$. As $\so \subseteq \co$ and $\co$ is transitive, we conclude that $(w, r) \in \co$.
\end{proof}

We show that $\overline{h}$ is indeed a full history, that is a witness of $h_\varphi$ and that also witness $h_\varphi$'s consistency.

\begin{lemma}
\label{lemma:k-expressive:sat-cons:witness}
The history $\overline{h}$ is a witness of $h_\varphi$.
\end{lemma}

\begin{proof}
By Lemma~\ref{lemma:k-expressive:sat-cons:full-history} $\overline{h}$ is a full history. Hence, for proving that $\overline{h}$ is a witness of $h_\varphi$, we need to show that for every key $x \in \Vars$ and every read $r$, if $\wro_x^{-1}(r) \uparrow$, $\where{r}(\valuewr[\overline{\wro}]{w_x^r}{x}) = 0$. Note that by construction of $h_\varphi$, such cases coincide with $x \in \var{\varphi}$. In addition, we observe that if $r$ is a read event, there exists indices $1 \leq i \leq n, 0 \leq j \leq 2$ s.t. $r \in t_i^j$. Thus, by Equation~\ref{eq:k-expressive:eq-def}, we only need to prove that $\where{r}(\valuewr[\overline{\wro}]{w_x^r}{x}) = 0$ whenever $x$ is $v_i^0, v_i^1$ or $v_i^2$. 

We prove as an intermediate step that in each round, every key has the same value in $\overline{h}$. For every round $i$ and key $x \in \var{\varphi}$, we consider the transaction $t_x^i = \max_{\co}\{t \ | \ \writeVar{t}{x} \land \roundTransaction{t} \leq i\}$. We prove by induction on the number of the round that $\valuewr[\overline{\wro}]{t_x^i}{x} = \valuewr[\overline{\wro}]{t_x^0}{x} = (x, \alpha(x))$. 
\begin{itemize}
    \item \underline{Base case}: The base case, $i = 0$, is immediately trivial. Note that in this case $\valuewr[\overline{\wro}]{t_x^0}{x} =  \alpha(x)$.
    
    \item \underline{Inductive case}: Let us assume that $\valuewr[\overline{\wro}]{t_x^{i-1}}{x} = \valuewr[\overline{\wro}]{t_x^0}{x}$ and let us prove that $\valuewr[\overline{\wro}]{t_x^i}{x} = (x, \alpha(x))$. Note that in round $i$ only keys $v_i^0, v_i^1$ and $v_i^2$ are written; so for every other key $x$, $t_x^i = t_x^{i-1}$ and by induction hypothesis, $\valuewr[\overline{\wro}]{t_x^i}{x} = (x, \alpha(x))$. Let thus $j, 0 \leq j \leq 2$ s.t. $\alpha(v_i^j) = 1$. In this case, $t_{v_i^j}^i = t_{v_i^{j+2 \bmod 3}}^i = t_i^{j+2\bmod 3}$ and $t_{v_i^{j+1 \bmod 3}}^i = t_i^{j+1\bmod 3}$. Hence, we can conclude the inductive step as:
    \begin{align*}
        \valuewr[\overline{\wro}]{t_i^{j + 2 \bmod 3}}{v_i^j} =  (v_i^j, 1) = (v_i^j,\alpha(v_i^j)) \\
        \valuewr[\overline{\wro}]{t_i^{j + 1 \bmod 3}}{v_i^{j+ 1 \bmod 3}} =  (v_i^{j+ 1 \bmod 3}, 0) = (v_i^{j+ 1 \bmod 3},\alpha(v_i^{j+1 \bmod 3})) \\
        \valuewr[\overline{\wro}]{t_i^{j + 2 \bmod 3}}{v_i^{j+ 2 \bmod 3}} = (v_i^{j+2 \bmod 3}, 0) = (v_i^{j+2 \bmod 3}, \alpha(v_i^{j+2 \bmod 3}))
    \end{align*}
\end{itemize}

We can thus conclude that $\overline{h}$ is a witness of $h_\varphi$. Let $i, j, 1 \leq i \leq n, 0 \leq j \leq 2$ be indices s.t. $\alpha(v_i^j) = 1$. For simplifying notation, we denote by $t_0, t_1, t_2$ to the transactions $t_i^j, t_i^{j+1 \bmod 3}$ and $t_i^{j+2 \bmod 3}$ respectively. We also denote by $r_0, r_1, r_2$ to the read events that belong to $t_0, t_1$ and $t_2$ respectively as well by $v_0, v_1, v_2$ to the keys associated to $t_0, t_1$ and $t_2$ respectively. For every key $x \neq v_0, v_1, v_2$ and for every transaction $t$ that writes $x$, $\where{r_j}(\valuewr{t}{x}) = 0, 0 \leq j \leq 2$; so we can focus only on keys $v_0, v_1$ and $v_2$. Three cases arise:
\begin{itemize}
    \item \underline{Transaction $t_0$:} Let thus $x$ be a key in $\{v_0, v_1, v_2\}$. By construction of $h_\varphi$ and $\co$, $t_0$ reads $x$ from $t_x^{i-1}$. As proved before, $\valuewr[\overline{\wro}]{t_x^{i-1}}{x} = (x, \alpha(x))$ and $\alpha(x) = (x, 1)$ if and only if $x = v_0$. Hence, as $\where{r_0}(\valuewr[\overline{\wro}]{t_x^{i-1}}{x}) = 0$ we conclude that $\where{r_0}(\valuewr[\overline{\wro}]{w_x^{r_0}}{x}) = 0$.
    
    \item \underline{Transaction $t_1$:} In this case, $t_1$ reads $v_2$ from $t_{v_2}^{i-1}$ and it reads $v_0$ and $v_1$ from $t_0$. On one hand, $\valuewr[\overline{\wro}]{t_{v_2}^{i-1}}{v_2} = (v_2, \alpha(v_2)) = (v_2, 0)$. Thus, as $\where{r_1}(t_{v_2}^{i-1}) = 0$, we conclude that $\where{r_1}(\valuewr[\overline{\wro}]{w_{v_2}^{r_1}}{v_2}) = 0$. On the other hand, by construction of $h_\varphi$, $\where{r_1}(\valuewr[\overline{\wro}]{t_0}{v_0}) = \where{r_1}(\valuewr[\overline{\wro}]{t_0}{v_1}) = 0$. Thus, the result hold.
    
    \item \underline{Transaction $t_2$:} In this case, $t_2$ read $v_0$ from $t_0$ and $v_1$ and $v_2$ from $t_1$. By construction of $h_\varphi$ both $\where{r_2}(\valuewr[\overline{\wro}]{t_0}{v_0})$, $\where{r_2}(\valuewr[\overline{\wro}]{t_1}{v_1})$ and $\where{r_2}(\valuewr[\overline{\wro}]{t_1}{v_2})$ are equal to $0$; so we conclude the result.
\end{itemize}

\end{proof}

We conclude the proof showing that the execution $\exec = \tup{\overline{h}, \co}$ is a consistent execution of $h_\varphi$. We observe that by construction of $\overline{\wro}$ and $\co$, $\overline{h}$ satisfies $\SER$ using $\co$. \Cref{corollary:ser-strongest} proves that $\isolation{h_\varphi}$ is weaker than $\SER$; which allow us to conclude that so $\overline{h}$ satisfies $\isolation{h_\varphi}$. In other words, that $\overline{h}$ is consistent.

\begin{proposition}
\label{prop:constraint-co}
Let $h = \tup{T, \so, \wro}$ be a full history, $\exec = \tup{h, \co}$ be an execution of $h$, $r$ be a $\iread$ event, $t_2$ be a transaction distinct from $\trans{r}$, $x$ be a key and $\mathsf{v} \in \visibilitySet{\isolation{h}(\trans{e})}$. If $\mathsf{v}(t_2, r, x)$ holds in $\exec$, then $(t_2, \trans{r}) \in \co$.
\end{proposition}

\begin{proof}
The proposition is result of an immediate induction on the definition of $\mathsf{v}$. The base case is $\po, \so, \wro \subseteq \co$, which holds by definition. The inductive case follows from the operators employed in Equation~\ref{eq:axioms-constraints-rel}: union, composition and transitive closure of relations; which are monotonic.
\end{proof}

As a consequence of \Cref{prop:constraint-co} and $\axser$ axiom definition, we obtain the following result.
\begin{corollary}
\label{corollary:ser-strongest}
Any isolation level is weaker than $\SER$.
\end{corollary}

\begin{lemma}
\label{lemma:k-expressive:consistent-sat}
If $h_\varphi$ is consistent then $\varphi$ is $1$-in-$3$ satisfiable.
\end{lemma}

\begin{proof}
If $h_\varphi$ is consistent, there exists a consistent witness of $h_\varphi$ $\overline{h} = \tup{T, \so, \overline{\wro}}$. As $\overline{h}$ is consistent and $\isolation{h}$ is stronger than $\RC$, there exists a consistent execution of $\overline{h}$, $\exec = \tup{\overline{h}, \co}$. Let $\alpha_h : \var{\varphi} \to \{0, 1\}$ s.t. for every variable $v_j, 1 \leq j \leq m$, $\alpha_h(v_j) = 1$ if and only if $(0_j, 1_j) \in \co$. We show that $\varphi$ is $1$-in-$3$ satisfiable using $\alpha$.

As an intermediate step, we prove that $\alpha_h$ describes the value read by any transaction in $\overline{h}$. For every $i, 0 \leq i \leq n$ and key $x \in \var{\varphi}$, let $t_x^i = \max_{\co} \{t \ | \ \writeVar{t}{x} \land \roundTransaction{t} \leq i\}$. We prove by induction that for every $i, 0 \leq i \leq n$ (1) $\valuewr[\overline{\wro}]{t_x^i}{x} = (x, \alpha(x))$, (2) for any $\iread$ event $r$ from a transaction $t$ s.t. $\roundTransaction{t} \leq i$, if $(w, r) \in \overline{\wro}_x$, then $w$ coincides with $\max_{\co}\{t \in T \ | \ \writeVar{t}{x} \land (t, \trans{r}) \in \co\}$ and (3) if $i > 0$, $\alpha(v_i^j) = 1$ if and only if $(t_i^j, t_i^{j+1 \bmod 3}) \in \co$ and $(t_i^{j+1 \bmod 3}, t_i^{j+2 \bmod 3}) \in \co$.
\begin{itemize}
    \item \underline{Base case}: Let $j, 1 \leq j \leq m$ be the integer s.t. $x = v_j$. In such case, (1) holds as $t_x^0 = 1_j$ if and only if $\alpha(v_j) = 1$; and in such case, $\valuewr[\overline{\wro}]{t_x^0}{v_j} = (v_j, \alpha(v_j))$. Also (2) trivially holds as there is no $\iread$ event in a transaction belonging to round $0$. Finally, (3) also trivially holds as $i = 0$.
    
    \item \underline{Inductive case}: We assume that (1), (2) and (3) hold for round $i-1$ and let us prove it for round $i$. Let $j$ the index of the $\co$-minimal transaction among $t_i^1, t_i^2, t_i^3$. We denote by $t_0, t_1, t_2$ to $t_i^j, t_i^{j+1 \bmod 3}$ and $t_i^{j+2 \bmod 3}$ respectively, by $r_0, r_1, r_2$ to the unique $\iread$ event in $t_0, t_1$ and $t_2$ respectively and by $v_0, v_1$ and $v_2$ to the keys associated to $t_0, t_1$ and $t_2$ respectively. 
    
    Let thus $x \in \var{\varphi}$ be a key, $t$ be a transaction among $t_0, t_1, t_2$ and let $w_x^t$ be a transaction s.t. $(w_x^t, t) \in \overline{\wro}_x$. As $\roundTransaction{t_x^{i-1}} < \roundTransaction{t}$, $(t_x^{i-1}, t) \in \wro_{t_x^{i-1}}$. Hence, as $\overline{h}$ satisfies $\RC$, either $w_x^t = t_x^{i-1}$ or $\roundTransaction{w_x^t} = i$. 
    
    First we prove (3) analyzing $t_0$. As $(t_0, t_1) \in \co$ and $(t_0, t_2) \in \co$ and $\overline{\wro} \subseteq \co$ we deduce that $w_x^t = t_x^{i-1}$. In such case, as (1) holds by induction hypothesis and $\where{r_0}(\valuewr[\overline{\wro}]{t_x^{i-1}}{x}) = 0$, we conclude that $\alpha(x) = 1$ if $x = v_0$ and $\alpha(x) = 0$ if $x = v_1, v_2$.
    
    For proving (2) we analyze three cases depending on $t$:
    \begin{itemize}
        \item \underline{$t = t_0$:} As proved before, if $t = t_0$, $w_x^t = t_x^{i-1}$. By definition of $t_x^{i-1}$, (2) holds.

        \item \underline{$t = t_1$:} As $t_0$ only writes $v_0$ and $v_1$ and $(t_1, t_2) \in \co$ we deduce that for every key $x \neq v_0, v_1$, $w_x^{t_1} = t_x^{i-1}$; which immediately implies (2). As (3) holds for round $i$, we know that $\alpha(v_0) = 1$ and $\alpha(v_1) = 0$. Thus, if $x = v_0, v_1$, $\where{r_2}(\valuewr[\overline{\wro}]{t_{x}^{i-1}}{x}) = 1$; so $(t_x^{i-1}, t_1) \not\in \overline{\wro}_x$. In conclusion, $w_x^{t_1} = t_0$; which implies (2) by definition of $t_0$.
        
        \item \underline{$t = t_2$:} As $t_0, t_1$ only write $v_0, v_1$ and $v_2$ we deduce that for every other key, $w_x^{t_2} = t_x^{i-1}$; which implies (2). Otherwise, we analyze the three sub-cases arising:
        \begin{itemize}
            \item \underline{$x = v_2$:} In this case, $t_0$ does not write $v_2$; so there is only two options left, $t_x^{i-1}$ and $t_1$. As (3) holds for round $i$, $\alpha(v_2) = 0$. Thus, as by induction hypothesis (2) holds for round $i-1$, $\valuewr[\overline{\wro}]{t_{v_2}^{i-1}}{v_2} = (v_2,0)$ and hence, $\where{r_2}(\valuewr[\overline{\wro}]{t_{v_2}^{i-1}}{v_2}) = 1$. Therefore, $w_{v_2}^{t_2}$ must be $t_1$; which implies (2).

            \item \underline{$x = v_0$:} Once again, there is only two possible options as $t_1$ does not write $v_0$. As (3) holds for round $i$, $\alpha(v_0) = 1$. Thus, as by induction hypothesis (2) holds for round $i-1$, $\valuewr[\overline{\wro}]{t_{v_0}^{i-1}}{v_0} = (v_0,1)$ and hence, $\where{r_2}(\valuewr[\overline{\wro}]{t_{v_0}^{i-1}}{v_0}) = 1$. Therefore, $w_{v_0}^{t_2}$ must be $t_0$; which implies (2).

            \sloppy \item \underline{$x = v_1$:} We observe in this case that $\valuewr[\overline{\wro}]{t_0}{v_1} = (x, 1)$; so $\where{r_2}(\valuewr[\overline{\wro}]{t_0}{v_1}) = 1$. Therefore, there is only two possible options, $t_1$ and $t_x^{i-1}$. As $\overline{h}$ satisfies $\RC$ and $(t_1, t_2) \in \overline{\wro}_{v_2}$, if $(t_x^{i-1}, t_2) \in \overline{\wro}_{v_1}$, we deduce that $(t_1, t_x^{i-1}) \in \co$. However, as $\roundTransaction{t_x^{i-1}} < \roundTransaction{t_1}$, $(t_x^{i-1}, t_1) \in \wro_{t_x^{i-1}}$; which is impossible as $\overline{\wro} \subseteq \co$. Thus, we conclude that $w_{v_2}^{t_2} = t_1$; which implies (2).
        \end{itemize}
    \end{itemize}

    For proving (1) we observe that for every key $x \neq v_0, v_1, v_2$, $t_x^{i} = t_x^{i-1}$ and by induction hypothesis we conclude that $\valuewr[\overline{\wro}]{t_x^i}{x} = (x, \alpha(x))$. Moreover, as $(t_0, t_1) \in \co$ and $(t_1, t_2) \in \co$, $t_{v_0}^i = t_{v_2}^i = t_2$ and $t_{v_1}^i = t_1$. In addition, as (3) holds, $\alpha(v_0) = 1$ and $\alpha(v_1) = \alpha(v_2) = 0$. This allow us to conclude (1) also for the keys $v_0, v_1$ and $v_2$; so the inductive step is proven.
\end{itemize}

After proving (1), (2) and (3) we can conclude that $\varphi$ is $1$-in-$3$ satisfiable. For every round $i$, we observe that by (1) $\valuewr[\overline{\wro}]{t_x^i}{x} = (x, \alpha(x))$. Moreover, as (2) holds, $(t_x^i, t_0^i) \in \overline{\wro}_x$; where $t_0^i$ is the first transaction in $\co$ among the transactions in round $i$. As $\overline{h}$ is a witness of $h_\varphi$, $\where{r_0^i}(\valuewr[\overline{\wro}]{t_x^i}{x}) = 0$; where $r_0^i$ is the read event of $t_0^i$. Hence, exactly one variable among $v_i^0, v_i^1$ and $v_i^2$ has $1$ as image by $\alpha$. Therefore, $\varphi$ is $1$-in-$3$ satisfiable.

\end{proof}

%% file: figures-tex/k-expressive-np-code.tex
\begin{figure}[t]
\centering
\resizebox{\textwidth}{!}{
	\begin{tikzpicture}[>=stealth,shorten >=1pt,auto,node distance=3cm,
		semithick, transform shape,
		B/.style = {%
		decoration={brace, amplitude=1mm,#1},%
		decorate},
		B/.default = ,  %
		]
		\node[minimum width=14.5em, minimum height=1.5em, draw, rounded corners=2mm,outer sep=0, label={[font=\normalsize]3:$\mathbf{1_1}$}] (x1) at (0, .5) {$\iinsert{\{x_1 : 1, 1_1 : -1\}}$};	
		\node[minimum width=11em] (dots1) at (0, -.25) {\ldots};
		\node[minimum width=14.5em, minimum height=1.5em, draw, rounded corners=2mm,outer sep=0, label={[font=\normalsize]3:$\mathbf{1_k}$}] (xk1) at (0, -1) {$\iinsert{\{x_k : 1, 1_k : -1\}}$};
		\node[minimum width=11em] (dots1bis) at (0, -1.75) {\ldots};
		\node[minimum width=14.5em, minimum height=1.5em, draw, rounded corners=2mm,outer sep=0, label={[font=\normalsize]3:$\mathbf{1_m}$}] (xm1) at (0, -2.5) {$\iinsert{\{x_m : 1, 1_m : -1\}}$};
		\node[minimum width=14.5em, minimum height=1.5em, draw, rounded corners=2mm,outer sep=0, label={[font=\normalsize]2:$\mathbf{t_1^0}$}] (r11) at (0, -4.5) {
			\begin{tabular}{l}
			$\eselect(\lambda r: \mathtt{eq}(r, v_1^0, v_1^1, v_1^2))$ \\
			$\iinsert{\{v_1^0 :0, v_1^1 : 1, t_1^0: -1\}}$
			\end{tabular}
		};
		\node[minimum width=14.5em, minimum height=1.5em, ] (dots1b) at (0, -6) {\ldots};

		\node[minimum width=14.5em, minimum height=1.5em, draw, rounded corners=2mm,outer sep=0, label={[font=\normalsize]2:$\mathbf{t_n^0}$}] (rm1) at (0, -7.5) {
			\begin{tabular}{l}
			$\eselect(\lambda r: \mathtt{eq}(r, v_n^0, v_n^1, v_n^2))$ \\
			$\iinsert{\{v_n^0 :0, v_n^1 : 1, t_n^0 : 1\}}$

			\end{tabular}
		};

		\node[minimum width=14.5em, minimum height=1.5em, draw, rounded corners=2mm,outer sep=0, label={[font=\normalsize]3:$\mathbf{0_1}$}] (x2) at (6, 0.5) {$\iinsert{\{x_1 : 0, 0_1 : -1\}}$};	
		\node[minimum width=6em ] (dots2) at (6, -.25) {\ldots};	
		\node[minimum width=14.5em, minimum height=1.5em, draw, rounded corners=2mm,outer sep=0, label={[font=\normalsize]3:$\mathbf{0_k}$}] (xk2) at (6, -1) {$\iinsert{\{x_k : 0, 0_k : -1\}}$};
		\node[minimum width=6em ] (dots2bis) at (6, -1.75) {\ldots};	
		\node[minimum width=14.5em, minimum height=1.5em, draw, rounded corners=2mm,outer sep=0, label={[font=\normalsize]3:$\mathbf{0_m}$}] (xm2) at (6, -2.5) {$\iinsert{\{x_m : 0, 0_m: -1\}}$};
		\node[minimum width=14.5em, minimum height=1.5em, draw, rounded corners=2mm,outer sep=0, label={[font=\normalsize]2:$\mathbf{t_1^1}$}] (r12) at (6, -4.5) {
			\begin{tabular}{l}
			$\eselect(\lambda r: \mathtt{eq}(r, v_1^1, v_1^2, v_1^0))$
			\\
			$\iinsert{\{v_1^1 :0, v_1^2 : 1, t_1^1: -1\}}$
			\end{tabular}
		};
		\node[minimum width=14.5em] (dots2b) at (6, -6) {\ldots};
		
		\node[minimum width=14.5em, minimum height=1.5em, draw, rounded corners=2mm,outer sep=0, label={[font=\normalsize]2:$\mathbf{t_n^1}$}] (rm2) at (6, -7.5) {
			\begin{tabular}{l}
			$\eselect(\lambda r: \mathtt{eq}(r, v_n^1, v_n^2, v_n^0))$ \\
			$\iinsert{\{v_n^1 :0, v_n^2 : 1, t_n^1 : -1\}}$
			\end{tabular}
		};	

		\node[minimum width=14.5em, minimum height=1.5em, draw, rounded corners=2mm,outer sep=0, label={[font=\normalsize]2:$\mathbf{t_1^2}$}] (r13) at (12, -4.5) {
			\begin{tabular}{l}
			$\eselect(\lambda r: \mathtt{eq}(r, v_1^2, v_1^0, v_1^1))$ \\
			$\iinsert{\{v_1^2 :0, v_1^0 : 1, t_1^2: -1\}}$
			\end{tabular}
		};
		\node[minimum width=14.5em] (dots3b) at (12, -6) {\ldots};
		\node[minimum width=14.5em, minimum height=1.5em, draw, rounded corners=2mm,outer sep=0, label={[font=\normalsize]2:$\mathbf{t_n^2}$}] (rm3) at (12, -7.5) {
			\begin{tabular}{l}
			$\eselect(\lambda r: \mathtt{eq}(r, v_n^2, v_n^0, v_n^1))$ \\
			$\iinsert{\{v_n^2 :0, v_n^0 : 1, t_n^2\}}$

			\end{tabular}
		};
		\path (x1) edge[->, soColor, right] node [] {$\so$} (dots1);
		\path (dots1) edge[->, soColor, right] node [] {$\so$} (xk1);
		\path (xk1) edge[->, soColor, right] node [] {$\so$} (dots1bis);
		\path (dots1bis) edge[->, soColor, right] node [] {$\so$} (xm1);
		\path (xm1) edge[->, soColor, right] node [] {$\so$} (r11);
		\path (r11) edge[->, soColor, right] node [] {$\so$} (dots1b);
		\path (dots1b) edge[->, soColor, right] node [] {$\so$} (rm1);

		\path (x2) edge[->, soColor, right] node [] {$\so$} (dots2);
		\path (dots2) edge[->, soColor, right] node [] {$\so$} (xk2);
		\path (xk2) edge[->, soColor, right] node [] {$\so$} (dots2bis);
		\path (dots2bis) edge[->, soColor, right] node [] {$\so$} (xm2);
		\path (xm2) edge[->, soColor, right] node [] {$\so$} (r12);
		\path (r12) edge[->, soColor, right] node [] {$\so$} (dots2b);
		\path (dots2b) edge[->, soColor, right] node [] {$\so$} (rm2);

		\path (r13) edge[->, soColor, right] node [] {$\so$} (dots3b);
		\path (dots3b) edge[->, soColor, right] node [] {$\so$} (rm3);

		\path (xk2.south west) edge[->, wrColor] node [below , transform canvas={xshift=4, yshift=-8}] {$\wro_{0_k}$} (r11.north east);

		\path (xk1.south west) edge[->, wrColor, dashed, bend right=30] node [shift={(0,-0.7)}] {$\wro_{v_1^0}$} (r11.north west);
		\path (r11.east) edge[->, wrColor, dashed, transform canvas={yshift=-8}] node [below] {$\wro_{v_1^1}$} (r12.west);
		\path (r12.east) edge[->, wrColor, dashed, transform canvas={yshift=-8}] node [below] {$\wro_{v_1^2}$} (r13.west);
		\path ($(xk2.west) + (0, 0)$) edge[->, coColor, dashed, double equal sign distance] node [] {$\co$} ($(xk1.east) + (0, 0)$);

		\draw[B=mirror]
    	($(x1.north west)+(-0.75, 0)$) -- node[rotate=90, xshift=-2em, yshift=1em] {Round $0$}  ($(xm1.south west)+(-0.75, 0)$);

		\draw[B=mirror]
    	($(r11.north west)+(-0.75, 0)$) -- node[rotate=90, xshift=-2em, yshift=1em] {Round $1$}  ($(r11.south west)+(-0.75, 0)$);

		\draw[B=mirror]
		($(rm1.north west)+(-0.75, 0)$) 
		-- node[rotate=90, xshift=-2em, yshift=1em] {Round $n$}   ($(rm1.south west)+(-0.75, 0)$);

	\end{tikzpicture}  
}

\vspace{-2mm}
\caption{Description of the history $h_\varphi$ from Theorem~\ref{th:k-expressive-np}. Dashed edges only belong to a possible consistent witness of $h_\varphi$, where we assume $v_1^0 = x_k$. Transaction $t_1^0$ reads $v_1^0, v_1^1$ and $v_1^2$ from round $0$; imposing some constraints on the transactions that write them. Due to axiom $\RC$'s definition, transaction $t_1^1$ must read $v_1^1$ from $t_1^0$ while transaction $t_1^2$ must read $v_1^1$ from $t_1^1$.
}
\label{fig:k-expressive-np-code}
\end{figure}

%% file: appendices/newproofs/proofs-partial-observation.tex
\subsection{Proof of \Cref{th:extend-np-complete}}
\label{app:proofs-observable-np-complete}

\extendNpComplete*

The structure of the proof is divided in two parts: proving that the problem is NP and proving that it is NP-hard. The fist part, corresponding to \Cref{lemma:np-complete:problem-in-np}, is straightforward as, for any client history, we simply guess a suitable witness and a total order on its transactions for deducing its consistency applying \Cref{def:consistency-full}. The second part, corresponding to Lemmas~\ref{lemma:np-complete:sat-cons} and~\ref{lemma:np-complete:cons-sat} is more complicated. We use a novel reduction from $3$-SAT. We encode a boolean formula $\varphi$ in a history $h_\varphi$, s.t. $h_\varphi$ is consistent iff $\varphi$ is satisfiable. We first prove that the problem is indeed in NP (\Cref{lemma:np-complete:problem-in-np}).

\begin{lemma}
\label{lemma:np-complete:problem-in-np}
The problem of checking consistency for a client history with an isolation configuration stronger than $\RC$ is in NP.
\end{lemma}

\begin{proof}
Let $h = \tup{T, \so, \wro}$ a client history whose isolation configuration is stronger than $\RC$. Guessing a witness of $h$, $\overline{h}$, and an execution of $\overline{h}$, $\exec = \tup{\overline{h}, \co}$, can be done in $\mathcal{O}(|\Vars| \cdot |\events{h}|^2 + |T|^2) \subseteq \mathcal{O}(|h|^3)$. By \Cref{lemma:necessary-co:correctness}, checking if $\exec$ is consistent is equivalent as checking if $\computeCO{\overline{h}}{\co}$ is an acyclic relation. As by \Cref{lemma:necessary-co:polynomial-time}, \Cref{algorithm:necessary-co} requires polynomial time, we conclude the result.

\end{proof}

For showing NP-hardness, we reduce 3-SAT to the problem of checking consistency of a partial observation history. Note that the problem is NP-hard in the case where the isolation configuration is not saturable, as discussed in \Cref{ssec:consistency-full-histories}, using the results in \cite{DBLP:journals/pacmpl/BiswasE19}. Therefore, we only prove it for the case where the isolation configuration is saturable.

Let $\varphi = \bigwedge_{i = 1}^n C_i$ a CNF expression with at most 3 literals per clause (i.e. $C_i = l_i^1 \lor l_i^2 \lor l_i^3$). Without loss of generality we can assume that each clause contains exactly three literals and each literal in a clause refers to a different variable. We denote $\var{l_i^j}$ to the variable that the literal $l_i^j$ employs and $\VarsFomula{\varphi}$ the set of all variables of $\varphi$. %

We design a history $h_{\varphi}$ with an arbitrary saturable isolation configuration encoding $\varphi$. Thus, checking $\varphi$-satisfiability would reduce to checking $h_{\varphi}$'s consistency. Note that as $\isolationExecution{h_\varphi}$ is saturable, $h_\varphi$'s consistency is equivalent to checking $\pco$'s acyclicity; where $\pco = \computeCO{h_\varphi}{(\so \cup \wro)^+}$. We use the latter characterization of consistency for encoding the formula $\varphi$ in $h_\varphi$.

\input{figures-tex/np-code}

First of all, we consider every literal in $\varphi$ independently. This means that even if two literals $l_i^j$ and $l_{i'}^{j'}$ share its variable ($\var{l_i^j} = \var{l_{i'}^{j'}}$) we will reason independently about them. For that, we employ keys $\var{l_i^j}_i$ and $\var{l_{i'}^{j'}}_{i'}$. We later enforce additional constraints for ensuring that $\var{l_i^j}_i$ and $\var{l_{i'}^{j'}}_{i'}$ coordinate so assignments on $\var{l_i^j}_i$ coincide with assignments in $\var{l_{i'}^{j'}}_{i'}$. For simplicity in the explanation, whenever we talk about a literal $l$ that is negated (for example $l \coloneqq \lnot x$), we denote by $\lnot l$ to the literal $x$. Also, we use indistinguishably $x$ when referring to a variable in $\varphi$ or to a homonymous key in $h_\varphi$. In addition, with the aim of simplifying the explanation, we assume hereinafter that any occurrence of indices $i, i', j, j'$ satisfy that $1 \leq i, i' \leq n$ and $1 \leq j, j' \leq 3$.

For every clause $C_i = l_i^1 \lor l_i^2 \lor l_i^3$, we create nine transactions denoted by $t_i^j$, $\lnot t_i^j$ and $S_i^j$. Figure~\ref{fig:np-code} shows in detail their definition, which we explain and justify during the following lines. The transaction $t_i^j$ represents the state where $l_i^j$ is satisfied while $\lnot t_i^j$ represents the state where $l_i^j$ is unsatisfied. Transaction $S_i^j$ is in charge of selecting one of the two states. With this goal on mind, transactions $t_i^j$ and $\lnot t_i^j$ contain a $\edelete$ event that deletes the key $\var{l_i^j}_i$ while $S_i^j$ contains a $\eselect$ event that does \emph{not} read $\var{l_i^j}_i$ in $h_\varphi$.
By \Cref{def:partial-observation-history}, any witnesses $h'$ of $h_\varphi$ must read $\var{l_i^j}_i$ from a transaction that deletes it. As $h_\varphi$ contain only two transactions that deletes such key ($t_i^j$ and $\lnot t_i^j$), we can interpret that if $S_i^j$ reads $\var{l_i^j}_i$ from $t_i^j$ in $h'$, then $l_i^j$ is satisfied in $\varphi$ while otherwise it is not.

For simplifying notation, as transactions $t_i^j, \lnot t_i^j, S_i^j$ only have one read event, we define write-read dependencies directly from transactions instead of their read events. We also denote by $\var{t_i^j}$ and $\var{\lnot t_i^j}$ to the variable $\var{l_i^j}$, associating each transaction with the variable of its associated literal.

We divide the construction of the history $h_\varphi$ in two main parts. In the first part, we relate transactions $t_i^j, \lnot t_i^j$ and $S_i^j$ with the clause $C_i$, ensuring a satisfying valuation of clause $C_i$ corresponds to a consistent history when restricted to its associated transactions. In the second part, we link transactions associated to different clauses (for example $t_i^j$ with $t_{i'}^{j'}$, $i \neq i'$), for ensuring that valuations are consistent between clauses (i.e. a variable is not assigned $1$ in clause $C_i$ and $0$ in clause $C_{i'}$).

\input{figures-tex/cycles-clause}

For the first part of $h_\varphi$'s construction, we observe that ``at least one literal among $l_i^1$, $l_i^2$ or $l_i^3$ must be satisfied'' is equivalent to ``$\lnot l_i^1$, $\lnot l_i^2$ and $\lnot l_i^3$ cannot be satisfied at the same time''. 
Thus, we add write-read dependencies to the history in such a way that if the three values that do not satisfy the clause are read by a witness $h'$ of $h_\varphi$, axiom $\axrc$ forces $h'$ to be inconsistent. We use an auxiliary key $c_i^j$ written by transactions $t_i^j$, $\lnot t_i^j$ and $\lnot t_i^{(j+1) \bmod 3}$ and read by transaction $S_i^j$; enforcing $(\lnot t_i^{(j+1) \bmod 3}, S_i^j) \in \wro_{c_i^j}$. Thanks to variable $c_i^j$, if $(\lnot t_i^j, S_i^j) \in \wro_{\var{l_i^j}_i}$ in such witness $h'$, for any consistent execution of $h'$ with commit order $\co$, $(\lnot t_i^j, \lnot t_i^{j+1 \bmod 3}) \in \co$. Hence, if $h'$ is consistent, for every $i$ there must exist a $j$ s.t.$(t_i^j, S_i^j) \in \wro_{\var{l_i^j}_i}$. Otherwise, every commit order witnessing $h'$'s consistency would be cyclic; which is a contradiction.

In \Cref{fig:cycles-clause} we see how such $\co$-cycle arise on any commit order witnessing $h_\varphi$'s consistency; where $\varphi$ contains the clause $C_i = x_i \lor y_i \lor \lnot z_i$.

\begin{equation}
\label{eq:np-observable:sign}
\signTransaction{t} = \left\{\begin{array}{rl}
    + & \text{if } t = t_i^j \land l_i^j = \var{l_i^j} \\
    - & \text{if } t = t_i^j \land l_i^j = \lnot \var{l_i^j} \\
    - & \text{if } t = \lnot t_i^j \land l_i^j = \var{l_i^j} \\
    + & \text{if } t = \lnot t_i^j \land l_i^j = \lnot \var{l_i^j} \\
    \bot & \text{otherwise}
\end{array} \right. 
\opSignTransaction{t} = \left\{\begin{array}{rl}
    + & \text{if } \signTransaction{t} = - \\
    - & \text{if } \signTransaction{t} = + \\
    \bot & \text{otherwise}
\end{array} \right.
\end{equation}

For the second part of $h_\varphi$'s construction, we utilize the functions $\mathtt{sign}$ and $\mathtt{opsign}$ described in Equation~\ref{eq:np-observable:sign}. The function $\mathtt{sign}$ describes when a literal $l_i^j$ is \emph{positive} (i.e. $l_i^j = \var{l_i^j}$) or \emph{negative} (i.e. $l_i^j = \lnot \var{l_i^j}$). If $l_i^j$ is positive, it assigns to transaction $t_i^j$ the symbol $+$ and to $\lnot t_i^j$ the symbol $-$; while if $l_i^j$ is negative, the opposite. Such symbol is denoted the \emph{sign} of a transaction.
Hence, for each transaction $t_i$ s.t. $\signTransaction{t_i} \neq \bot$ (i.e. $t_i$ is either $t_i^j$ or $\lnot t_i^j$), we introduce $n-1$ $\einsert$ events, one per key $\var{l_i^j}_{(i, i')}^{\signTransaction{t_i}}$, $i' \neq i$, that write on that exact key.
Those $\einsert$ events are read by transaction $S_{i'}^{j'}$ (i.e. $(t_i, S_{i'}^{j'}) \in \wro_{x}$, where $x = \var{l_i^j}_{(i, i')}^{\signTransaction{t_i}}$). In addition, we introduce in $t_i$ another $n-1$ $\einsert$ events that writes the key $\var{t_i}_{(i', i)}^{\opSignTransaction{t_i}}$. Figure~\ref{fig:np-code} describe the full description of transactions $t_i^j$ and $\lnot t_i^j$. 

\input{figures-tex/gadget-consistency}

The auxiliary keys $\var{l_i^j}_{(i, i')}^{\signTransaction{t_i}}$ and $\var{l_i^j}_{(i, i')}^{\opSignTransaction{t_i}}$ are key to ensure consistency between clauses. Intuitively, if $S_i^j$ reads $\var{l_i^j}_i$ from a positive transaction $t$ in a consistent execution of $h_\varphi$, $\exec = \tup{\overline{h}, \co}$, and $t'$ is a negative transaction s.t. $\var{t} = \var{t'}$, then $(t, t') \in \co$; where $\overline{h}$ is a witness of $h_\varphi$. Hence, any other transaction $S_{i'}^{j'}$ must read $\var{l_{i'}^{j'}}_{i'}$ also from a positive transaction in $\overline{h}$; otherwise $\co$ would be cyclic, which is impossible as $\co$ must be a total order. This phenomenon, that is depicted in Figure~\ref{fig:gadget-consistency}, ensures that $\var{l_i^j}$ is always read from transactions with the same sign. In conclusion, we can establish consistent valuation of the variables of $\varphi$ based on the write-read dependencies of the witnesses of $h_\varphi$.

We introduce a succint final part on the construction of $h_\varphi$ for technical reasons. Indeed, any witness of $h_\varphi$ ensures that $\wro^{-1}_\key$ is a total function for any $\key \in \Vars$. We impose a few additional constraints on $h_\varphi$ so we can better characterize the witnesses of $h_\varphi$. First, we assume that there exists an initial transaction that inserts, for every key $x$, a dummy value different from $\del_{x}$ (for example $0$). Then, we impose that $t_i^j$ and $\lnot t_i^j$ read every key $x$ from the initial transaction. Finally, for the case of transactions $S_i^j$, we define the set of \emph{auxiliary keys} $\mathtt{A}_i^j$ that contain every key different from $c_i^j, \var{l_i^j}_i, \var{l_i^j}_{(i', i)}^+$ and $\var{l_i^j}_{(i', i)}^-$. We introduce on $S_i^j$ an $\einsert$ event that writes every key in $\mathtt{A}_i^j$ with an abritrary value (for example, $0$). Hence, $S_i^j$ reads every key in $A_i^j$ from its own insert and no extra write-read dependency is required.

With this technical addendum, we define $h_\varphi = \tup{T, \so, \wro}$ as the conjunction of all transactions and relations described above. In such case, the only information missing in $h_\varphi$ to be a full history is $\wro_{\var{l_i^j}_i}^{-1}(S_i^j)$. We assume that no more variables than the ones aforementioned belong to $\Vars$.

The proof of NP-hardness goes as follows: first, we prove in \Cref{lemma:np-complete:polynomial-history} that $h_\varphi$ is indeed a polynomial transformation of $\varphi$. Then, as $\isolation{h_\varphi}$ is saturable, by Theorem~\ref{th:poly-consistency-saturable} we observe that it suffices to prove that $\varphi$ is satisfiable if and only there is a witness $\overline{h}$ of $h_\varphi$ s.t. the relation $\pco_{\overline{h}} = \computeCO{\overline{h}}{(\so \cup \overline{\wro})^+}$ is acyclic. For simplifying the reasoning when $\isolation{h_\varphi}$ has an arbitrary isolation configuration, we rely on \Cref{lemma:np-complete:co-cc-co-rc} for reducing the proof at the case when $\isolation{h_\varphi} = \RC$. 

Hence, we prove on \Cref{lemma:np-complete:sat-cons} that on one hand, whenever $\varphi$ is satisfiable we can construct a witness $\overline{h}$ of $h_\varphi$ based on such assignment for which $\pco_{\overline{h}}$ is acyclic. For that, we require \Cref{lemma:np-complete:co-cycle-clause,lemma:np-complete:co-cycle-necessary-wr,lemma:np-complete:co-edge-impossible}.
On the other hand, whenever there is a consistent witness $\overline{h}$ of $h_\varphi$, we prove on \Cref{lemma:np-complete:cons-sat} that we can construct a satisfying assignment of $\varphi$ based on the write-read dependencies in $\overline{h}$. In this case, we require once more \Cref{lemma:np-complete:co-cycle-clause}.

\begin{lemma}
\label{lemma:np-complete:polynomial-history}
The history $h_\varphi$ has polynomial size on the length of $\varphi$.
\end{lemma}

\begin{proof}
Let $\varphi$ a CNF with $n$ clauses and $3$ literals per clause. Then, as $\varphi$ has $3n$ literals, $h_\varphi$ employs $9n$ transactions plus one additional one ($\init$). The number of keys, $|\Vars|$, is quadratic on $n$ as transactions $t_i^j$ and $\lnot t_i^j$ insert $\mathcal{O}(n)$ keys while $S_i^j$ only insert keys also inserted by other transactions. Moreover, the number of events in $h_\varphi$, $\events{h_\varphi}$ is in $\mathcal{O}(|\Vars|) = \mathcal{O}(n^2)$ as transactions $t_i^j, \lnot t_i^j$ have one $\einsert$ event per keys inserted (and they insert $\mathcal{O}(n)$ keys) and one $\edelete$ event and transactions $S_i^j$ only have two events. In addition, $\so \subseteq T \times T$ and $\wro \subseteq \Vars \times \events{h_\varphi} \times  \events{h_\varphi}$; so their size is also polynomial on $n$. Thus, $h_\varphi$ is a polynomial transformation of $\varphi$. 
\end{proof}

One caveat of $h_\varphi$ is that its isolation configuration is unknown. \Cref{lemma:np-complete:co-cc-co-rc} express that, in the particular case of $h_\varphi$, all saturable isolation levels stronger than $\RC$ are equivalent (they impose the same constraints). Hence, thereinafter we can assume without loss of generality that $\isolation{h_\varphi} = \RC$.

\begin{lemma}
\label{lemma:np-complete:co-cc-co-rc}
Under history $h_\varphi$, $\isolation{h_\varphi}$ is equivalent to $\RC$ (i.e. $\isolation{h_\varphi}$ is both weaker and stronger than $\RC$.)
\end{lemma}

\begin{proof}
Let $\overline{h} = \tup{T, \so, \overline{\wro}}$ be any witness of $h_\varphi$ and let $h^{\RC}$ be the history that only differ with $\overline{h}$ on its isolation configuration ($\isolation{h^{\RC}} = \RC$ instead of $\isolation{h}$). We prove that $\overline{h}$ and $h^\RC$ are simultaneously consistent or inconsistent.

As both $\isolation{h_\varphi}$ and $\RC$ are saturable, by Theorem~\ref{th:poly-consistency-saturable}, the proof is equivalent to prove that $\pco_{\overline{h}}$ and $\pco_\RC$ are simultaneously cyclic or acyclic; where $\pco_{\overline{h}} = \computeCO{\overline{h}}{(\so \cup \overline{\wro})^+}$ and $\pco_\RC = \computeCO{h^\RC}{(\so \cup \overline{\wro})^+}$. We prove that the two relations coincide, which allow us to conclude the result. 

We observe that as $\isolation{\overline{h}}$ is weaker than $\RC$, $\pco_\RC \subseteq \pco_{\overline{h}}$. Thus, it suffices to prove that $\pco_{\overline{h}} \subseteq \pco_\RC$. Let $t, t'$ be two transactions s.t. $(t, t') \in \pco_{\overline{h}}$ and let us prove that $(t, t') \in \pco_\RC$. As $(\so \cup \overline{\wro})^+ \subseteq \pco_\RC$; we can assume without loss of generality that $(t, t') \in \pco_{\overline{h}} \setminus (\so \cup \overline{\wro})^+$. In such case, there exists $r \in \readOp{\overline{h}}, x \in \Vars$ and $v \in \visibilitySet{\isolation{\overline{h}}(\trans{r})}$ s.t. $t' = {\overline{\wro}_x}^{-1}(r)$ and $\visibilityRelationInstanceApplied[\pco_{\overline{h}}]{t}{r}{x}$ holds in $\overline{h}$. As $\isolation{\overline{h}}$ is saturable, $(t, r) \in (\so \cup \overline{\wro})^+$.

First, we note that $\trans{r} \neq \init$ as it does not contain any read event. %
As $t'$ is a $(\so \cup \overline{\wro})^+$-predecessor of $\trans{r}$, and transactions $S_i^j$ are $(\so \cup \overline{\wro})$-maximal, $t'$ is not a $S_i^j$ transaction; so it must be a $t_i^j$ transaction. However, note that by construction of transactions $t_i^j$, their only $(\so \cup \overline{\wro})$-predecessor is $\init$. Thus, their only $(\so \cup \overline{\wro})$-succesors can be transactions $S_{i'}^{j'}$; transactions that do not have $(\so \cup \overline{\wro})$-successors. In conclusion, if $(t', \trans{r}) \in (\so \cup \overline{\wro})^+$, $(t', \trans{r}) \in \so \cup \overline{\wro}$, and therefore, $(t', r) \in (\so \cup \overline{\wro});\po^*$. 
\end{proof}

\Cref{lemma:np-complete:co-cycle-clause} states a characterization of all commit order cycles imposed by the axiom $\RC$ that only relate the nine transactions associated to a clause in $\varphi$.

\begin{lemma}
\label{lemma:np-complete:co-cycle-clause}
Let $\overline{h} = \tup{T, \so, \overline{\wro}}$ a witness of $h_\varphi$. For a fixed $i$, there is a $\pco_{\overline{h}}$-cycle relating $\init$, $t_i^j, \lnot t_i^j$ and $S_i^j$, $1 \leq j \leq 3$ in $h$ if and only if for all $1 \leq j \leq 3$, $(\lnot t_i^j, S_i^j) \in \overline{\wro}_{\var{l_i^j}}$ in $\overline{h}$.
\end{lemma}

\begin{proof}

A graphical description of the different cases of this proof can be seen in Figure~\ref{fig:cycles-clause}.

\underline{$\impliedby$}

Let us suppose that for every $j$, $1 \leq j \leq 3$, $(\lnot t_i^j, S_i^j) \in \overline{\wro}_{\var{l_i^j}_i}$. As $\writeVar{\lnot t_i^j}{c_i^j}$ and $(\lnot t_i^{(j+1) \bmod 3}, S_i^j) \in \overline{\wro}_{c_i^j}$, by axiom $\RC$ we deduce that, $(\lnot t_i^j, \lnot t_i^{(j + 1) \bmod 3}) \in \pco_{\overline{h}}$. Therefore, there is a $\pco_{\overline{h}}$-cycle between transactions $\lnot t_i^1, \lnot t_i^2$ and $\lnot t_i^3$.

\underline{$\implies$}

First, note that $\so \cup \overline{\wro}$ is acyclic, so any $\pco_{\overline{h}}$-cycle has to include at least one $\pco_{\overline{h}} \setminus (\so \cup \overline{\wro})^+$-dependency. Hence, let $t, t'$ be distinct transactions such that $(t', t) \in \pco_{\overline{h}} \setminus (\so \cup \overline{\wro})^+$ is an edge belonging to such cycle. By axiom $\axrc$, this implies that there exists a read event $r$ and a key $x$ s.t. $(t, r) \in \overline{\wro}_x$ and $(t, r) \in (\so \cup \overline{\wro}); \po^*$. Note that in particular this means that $t'$ and $t$ are two distinct $(\so \cup \overline{\wro})$-succesors of $\trans{r}$.

We observe that $\trans{r} \neq \init$ as $\init$ does not contain any read event. Moreover, $\trans{r} \neq t_i^j, \lnot t_i^j$ as those transactions have only one $(\so \cup \overline{\wro})$-predecessor, $\init$. Hence, there exists $j$ s.t. $\trans{r} = S_i^j$. In this case, every ke written by $t_i^j$ or $\lnot t_i^j$ besides $c_i^j$ and $\var{l_i^j}_i$ is read by $S_i^j$ from the $\einsert$ event in its own transaction. We distinguish between two cases:

\begin{itemize}
    \item \underline{$x = \var{l_i^j}_i$}: The only transactions that write $\var{l_i^j}_i$ are $t_i^j$, $\lnot t_i^j$, $\init$ and transactions $S_{i'}^{j'}$. However, transactions $S_{i'}^{j'}$ have only one $(\so \cup \overline{\wro})$-succesor, $\init$ in $h_\varphi$. As $\forall x \neq \var{l_i^j}_i, {\overline{\wro}_x}^{-1}(S_i^j) = {\wro_x}^{-1}(S_i^j)$, one of them, $\init$ must be either $t$ or $t'$. However, $t \neq \init$ as $\init$ does not delete $\var{l_i^j}_i$; so $t = \init$. But in such case, $(t', t) \in \so$; which contradicts that $(t', t) \in \pco_{\overline{h}} \setminus (\so \cup \overline{\wro})^+$. This proves that this case is impossible.
    
    \item \underline{$x = c_i^j$}: In such case, as $(\lnot t_i^{j+1 \bmod 3}, S_i^j) \in \wro_{c_i^j}$, $t = \lnot t_i^{j+1 \bmod 3}$. The only transactions that writes $c_i^j$ and are $(\so \cup \overline{\wro})$-predecessors of $S_i^j$ are $\init, t_i^j$ and $\lnot t_i^j$. As $(\init, t) \in \so$; $t \neq \init$. Thus, any of the other two transactions are candidates to be $t'$. Note that $(t',t) \in \pco_{\overline{h}}$ is part of a cycle; so let $t''$ be a transaction s.t. $(t'', t') \in \pco_{\overline{h}}$. 
    
    If $(t'', t') \in (\so \cup \overline{\wro})$ would hold, as for every key $x$, ${\overline{\wro}_x}^{-1}(t) = {\wro_x}^{-1}(t)$, $t'' = \init$. As $(t',t)$ is part of a $\pco_{\overline{h}}$ cycle and $t \neq \init$, there must exist a transaction $t''' \neq t'' = \init$ s.t. $(t''', t'') \in \pco_{\overline{h}}$ is part of such cycle. Note that $(t''', \init) \in \pco_{\overline{h}} \setminus (\so \cup \overline{\wro})$ by construction of $h_\varphi$. Hence, there exists a key $y$ and a read event $r'$ s.t. $(\init, r') \in \overline{\wro}_y$ and $(t''', r') \in (\so \cup \overline{\wro}); \po^*$. By construction of $h_\varphi$, if $(\init, r') \in \overline{\wro}_y$ then $\trans{r'}$ must be $t_i^{j'}$ for some $j'$. But as we mentioned earlier, such transactions only have one $(\so \cup \overline{\wro})$-predecessor, $\init$; so it is impossible that $(t'', t') \in (\so \cup \overline{\wro})$.

    Hence, $(t'', t') \in \pco_{\overline{h}} \setminus (\so \cup \overline{\wro})$. Replicating the same argument as before we can deduce that there exists a $j'$ s.t. $(t'', S_i^{j'}) \in \overline{\wro}_{c_i^{j'}}$, $t' = \lnot t_i^{j'+1 \bmod 3}$ and $t''$ is either $t_i^{j'}$ or $\lnot t_i^{j'}$. However, as discussed before, $t'$ could only be $\lnot t_i^j$ or $t_i^j$. Therefore, $t' = \lnot t_i^j$ and $j = j'+1 \bmod 3$. 
    
    Finally, as $t'' \neq t$, there must exist a transaction $t'''$ s.t. $(t''', t'') \in \pco_{\overline{h}}$. By the same argument once more, there exists an index $j''$ s.t. $t'' = \lnot t_i^{j''+1 \bmod 3}$, $(t''', S_i^{j''}) \in \overline{\wro}_{c_i^{j''}}$ and $t'''$ is either $t_i^{j''}$ or $\lnot t_i^{j''}$. Once more, as $t''$ could only be $\lnot t_i^{j'}$ or $t_i^{j'}$, we deduce that $j' = j'' + 1 \bmod 3$ and $t'' = \lnot t_i^{j'}$. Note that in this case $j = j'' + 2 \bmod 3$. Thus, $t = \lnot t_i^{j + 1 \bmod 3} = \lnot t_i^{j''} = t'''$. In conclusion, if such cycle exists it contain exactly the transactions $\lnot t_i^1, \lnot t_i^2$ and $\lnot t_i^3$ and for each of them, $(\lnot t_i^j, S_i^j) \in \overline{\wro}_{\var{l_i^j}_i}$.

\end{itemize}

\end{proof}

\Cref{lemma:np-complete:co-cycle-necessary-wr} states that any $\pco_{\overline{h}}$-dependencies imposed by the axiom $\RC$ on transactions $t, t'$ associated to diferent clauses in $\varphi$ are related to valuation choices of literals in $\varphi$.

\begin{lemma}
\label{lemma:np-complete:co-cycle-necessary-wr}
Let $\overline{h} = \tup{T, \so, \overline{\wro}}$ a witness of the history $h_\varphi$.
For every pair of transactions $t,t'$ and indices $i,j$, if $\var{t} = \var{t'}$, $\deleteVar{t'}{\var{l_i^j}_i}$, $t \neq \lnot t_i^{j+1 \bmod 3}$ and $(t', t) \in \pco_{\overline{h}} \setminus (\so \cup \overline{\wro})^+$ in $h$, then $(t', S_i^j) \in \overline{\wro}_{\var{l_i^j}_i}$.
\end{lemma}

\begin{proof}
Let $i, j$ be indices and $t, t'$ be distinct transactions such that $t \neq \lnot t_i^{j+1 \bmod 3}$ and $(t', t) \in \pco_{\overline{h}} \setminus (\so \cup \overline{\wro})^+$. Hence, $(t', t) \in \pco_{\overline{h}} \setminus (\so \cup \overline{\wro})^+$, by axiom $\axrc$, there must exist a key $x$ and a read event $r$ s.t. $(t, r) \in \overline{\wro}_x$, $\writeVar{t'}{x}$ and $(t', r) \in (\so \cup \overline{\wro}) ; \po^*$. We characterize the possible candidates of transactions $t, t', \trans{r}$ and key $x$.

First, as $\trans{r}$ has two different $(\so \cup \overline{\wro})$-predecessors, $\trans{r} \neq \init, t_{i'}^{j'}$; for any indices $i', j'$. Hence, there must exist indices $i', j'$ s.t. $\trans{r} = S_{i'}^{j'}$. 

Next we deduce that $t$ and $t'$ belongs to different clauses. As $\deleteVar{t'}{\var{l_i^j}_i}$, we deduce that $t'$ is either $t_i^j$ or $\lnot t_i^j$. Hence, as $t \neq \lnot t_i^{j+1 \bmod 3}$, neither $t_i^j$ nor $\lnot t_i^j$ are $(\so \cup \overline{\wro})$-predecessors of $S_i^j$ but both $t$ and $t'$ are $(\so \cup \overline{\wro})$-predecessors of $S_i^j$, we then deduce that $t$ and $t'$ belong to different clauses.

Finally, we deduce that $i' = i$ and $(t', S_i^j) \in \overline{\wro}_{\var{l_i^j}_i}$. As $x$ is written by $t$ and $t'$ and $x \not\in \mathtt{A}_{i'}^{j'}$; either $t$ or $t'$ are associated to the same clause as $S_{i'}^{j'}$. If $t$ would be associated to clause $C_{i'}$, then $t$ should be either $t_{i'}^{j'}$ or $\lnot t_{i'}^{j'}$ and $x = \var{l_{i'}^{j'}}_{i'}$. However, this contradicts that $t'$ writes $\var{l_{i'}^{j'}}_{i'}$ as $t'$ is either $t_i^j$ or $\lnot t_i^j$. Hence, as $t$ is not associated to clause $C_{i'}$, $i' = i$. As $(t', S_i^j) \not\in (\so \cup \wro)$ but $(t', S_i^j) \in (\so \cup \overline{\wro})$ and $\overline{\wro}_y = \wro_y$ for any key $y \neq \var{l_i^j}_i$, we conclude that $(t', S_i^j) \in \overline{\wro}_{\var{l_i^j}_i}$. 
\end{proof}

\Cref{lemma:np-complete:co-edge-impossible} states that $\pco_{\overline{h}}$ does not contain tuples of transactions associated to literals with equal variable and sign.%

\begin{lemma}
\label{lemma:np-complete:co-edge-impossible}
Let $\overline{h} = \tup{T, \so, \overline{\wro}}$ a witness of the history $h_\varphi$. For every pair of transactions $t,t'$ and indices $i,j$, if $\signTransaction{t} = \signTransaction{t'}$, $\var{t} = \var{t'} = \var{l_i^j}$, $(t, S_i^j) \in \overline{\wro}_{\var{l_i^j}_i}$ and $t' \neq \lnot t_i^{j-1 \bmod 3}$ then $(t', t) \not\in \pco_{\overline{h}} \setminus (\so \cup \overline{\wro})^+$.%
\end{lemma}

\begin{proof}
We reason by contradiction. Let us suppose that $t, t'$ are a pair of transactions such that $\signTransaction{t} = \signTransaction{t'}$, $\var{t} = \var{t'} = \var{l_i^j}$, $(t, S_i^j) \in \overline{\wro}_{\var{t}_i}$, $t' \neq \lnot t^{j-1 \bmod 3}$ and $(t', t) \in \pco_{\overline{h}} \setminus (\so \cup \overline{\wro})^+$, for some indices $i,j$. As $(t, S_i^j) \in \overline{\wro}_{\var{l_i^j}_i}$, $t$ is either $t_i^j$ or $\lnot t_i^j$. Moreover, as $(t', t) \in \pco_{\overline{h}} \setminus (\so \cup \overline{\wro})^+$, by axiom $\axrc$ we deduce that there exists a key $x$ and a read event $r$ s.t. $(t', r) \in (\so \cup \overline{\wro});\po^*$, $(t, r) \in \overline{\wro}_x$ and $\writeVar{t'}{x}$.

We first prove that $t'$ and $t$ are associated to different clauses. As $(\init, t) \in \so$, $t' \neq \init$. Next, as $(t', r) \in (\so \cup \overline{\wro});\po^*$ and transactions $S_{i'}^{j'}$ are $\so \cup \overline{\wro}$-maximal, we deduce that there must exist a pair of indices $i', j'$ s.t. $t' = t_{i'}^{j'}$ or $\lnot t_{i'}^{j'}$. Moreover, as $t' \neq \lnot t_i^{j-1 \bmod 3}$, $t$ is either $t_i^j$ or $\lnot t_i^j$. In addition, as in any witness of $h_\varphi$ both $t_i^j$ and $\lnot t_i^j$ cannot be $(\so \cup \overline{\wro})$-predecessors of $S_i^j$, we deduce that $i' \neq i$. 

Finally we contradict the hypothesis proving that $\signTransaction{t} \neq \signTransaction{t'}$. If $i' \neq i$, $t \neq \init$ but $t$ is a $\overline{\wro}$-predecessor of $\trans{r}$, there must exist indices $i'', j''$ s.t. $\trans{r} = S_{i''}^{j''}$. Hence, as $x \in \mathtt{A}_{i''}^{j''}$ and it is written by $t$ and $t'$, $i''$ must be either $i'$ or $i$. However, $i'' \neq i$ as in that case, $x = \var{l_i^j}_i$ and $t'$ does not write $\var{l_i^j}_i$. Hence, $i'' = i' \neq i$ and $x = \var{l_i^j}_{(i', i)}^{\signTransaction{t'}}$. However, as $\writeVar{t}{x}$, by construction of $h_\varphi$, we must conclude that $\signTransaction{t} \neq \signTransaction{t'}$. Thus, as we reached a contradiction, the lemma holds.
\end{proof}

\begin{lemma}
\label{lemma:np-complete:sat-cons}
For every boolean formula $\varphi$, if $\varphi$ is satisfiable then there is a consistent witness $\overline{h}$ of $h_\varphi$.
\end{lemma}

\begin{proof}
Let $\alpha: \VarsFomula{\varphi} \to \{0, 1\}$ an assignment that satisfies $\varphi$. Let $h_\varphi^\alpha = \tup{T, \so, \overline{\wro}}$ the extension of $h_\varphi$ s.t. for every $i,j$, $(t_i^j, S_i^j) \in \overline{\wro}_{\var{l_i^j}_i}$ if $l_i^j[\alpha(\var{l_i^j})/\var{l_i^j}] = \btrue$ and $(\lnot t_i^j, S_i^j) \in \overline{\wro}_{\var{l_i^j}_i}$ otherwise. Note that for every two transactions $t, t'$ s.t. $\var{t} = \var{t'}$, $\alpha(\var{t}) = \alpha(\var{t'})$. Hence, if $(t, S_i^j) \in \overline{\wro}_{\var{t}_i}$ and $(t', S_{i'}^{j'}) \in \overline{\wro}_{\var{t'}_{i'}}$ then $\signTransaction{t} = \signTransaction{t'}$. In addition, by construction of $h_\varphi$, for every transaction $S_i^j$, the only variable $\key$ such that $\wro_\key^{-1}(S_i^j) \uparrow$ is $\key = \var{l_i^j}$.
Thus, for every $\key \in \Vars$, ${\overline{\wro}_{\key}}^{-1}$ is defined for any read that does not read locally and therefore, $h^\alpha_\varphi$ is a full history that extends $h_\varphi$.

Let us prove that $h^\alpha_\varphi$ is consistent. As mentioned before, thanks to Theorem~\ref{th:poly-consistency-saturable}, we can reduce the problem of checking if $h^\alpha_\varphi$ is consistent to the problem of checking if $\pco_{h^\alpha_\varphi} = \computeCO{h^\alpha_\varphi}{(\so \cup \overline{\wro})^+}$ is acyclic.

We reason by contradiction, assuming there is a $\pco_{h^\alpha_\varphi}$-cycle and reaching a contradiction. Clearly $\so \cup \overline{\wro}$ is acyclic as $\so \cup \wro$ is acyclic, transactions $S_i^j$ are $(\so \cup \overline{\wro})$-maximal and $\overline{\wro} \setminus \wro$ only contains tuples $(t_i^j, S_i^j)$ or $(\lnot t_i^j, S_i^j)$. Thus, any $\pco_{h^\alpha_\varphi}$-cycle in $h^\alpha_\varphi$ contains at least one edge $(t', t) \in \pco_{h^\alpha_\varphi} \setminus (\so \cup \overline{\wro})^+$; so let be $t, t'$ such a pair of distinct transactions s.t. $(t', t) \in \pco_{h^\alpha_\varphi} \setminus (\so \cup \overline{\wro})^+$ and $(t',t)$ is part of the $\pco_{h^\alpha_\varphi}$-cycle.

First, we observe that by construction of $h_\varphi$, transactions $S_i^j$ are $(\so \cup \overline{\wro})$-maximal. Moreover, they also $\pco_{h^\alpha_\varphi}$-maximal: by contradiction, if there was a transaction $u_i^j$ s.t. $(S_i^j, u_i^j) \in \pco_{h^\alpha_\varphi} \setminus (\so \cup \overline{\wro})+$, by axiom $\axrc$, there would be a variable a read event $r$ s.t. $(S_i^j, r) \in (\so \cup \overline{\wro}); \po^*$; which is impossible. We also observe that $\init$ is not only $(\so \cup \overline{\wro})$-minimal but $\pco_{h^\alpha_\varphi}$-minimal. By the same argument, if there would be a transaction $u \neq \init$ s.t. $(u, \init) \in \pco_{h^\alpha_\varphi} \setminus (\so \cup \overline{\wro})^+$, by axiom $\axrc$, there should be a read event $r$ and a key $x$ s.t. $(\init, r) \in \overline{\wro}_x$ and $(u', r) \in (\so \cup \overline{\wro}); \po^*$. However, by construction of $h_\varphi$, the only transactions that read a variable from $\init$ are $t_i^j$; transactions with only one $(\so \cup \overline{\wro})$-predecessor. This shows that such transaction $u$ does not exist. Altogether, the $\pco_{h^\alpha_\varphi}$-cycle can only contain pairs of transactions $t_i^j$ and $\lnot t_i^j$. In particular, as transactions $t_i^j$ have only one $(\so \cup \overline{\wro})$-predecessor, $\init$, such $\pco_{h^\alpha_\varphi}$-cycle is in $\pco_{h^\alpha_\varphi} \setminus (\so \cup \overline{\wro})^+$.

Next, we note that, as every clause $C_i$ is satisfied by $\alpha$, there exists an index $j$ s.t. $(t_i^j, S_i^j) \in \var{l_i^j}$. By \Cref{lemma:np-complete:co-cycle-clause}, we know there is no $\pco_{h^\alpha_\varphi}$-cycle relating the nine transactions associated with clause $C_i$ and $\init$. Therefore, a $\pco_{h^\alpha_\varphi}$-cycle has to involve at least two transactions from different clauses. Hence, we can assume without loss of generality that $t$ and $t'$ belong to the same clause.

As $(t', t) \in \pco_{h^\alpha_\varphi} \setminus (\so \cup \overline{\wro})^+$, there must exist a key $x$ and a read event $r_x$ s.t. $\writeVar{t}{x}$, $(t, r_x) \in \overline{\wro}_x$ and $(t', r) \in (\so \cup \overline{\wro}); \po^*$.
By construction of $h_\varphi$, the only case when two transactions from different clauses write the same variable is when $\var{t} = \var{t'}$. In particular, as $t$ and $t'$ belong to different clauses, there must exist indices $i, i'$ s.t. $x = \var{t}^{\signTransaction{t}}_{(i', i)}$ and $\trans{r_x} = S_i^j$. Hence, there is only one candidate for transaction $t'$: $t_i^j$ if $\signTransaction{t_i^j} = \signTransaction{t'} = \opSignTransaction{t}$ and $\lnot t_i^j$ otherwise. Therefore, as $t', \trans{r_x}$ belong to the same clause and $t'$ is a $(\so \cup \overline{\wro})$-predecessor of $\trans{r_x}$, we conclude that $(t', S_i^j) \in \overline{\wro}_{\var{t'}_i}$.

To reach a contradiction, we find a pair of distinct transactions $\tilde{t}, \hat{t}$ in the $\pco_{h^\alpha_\varphi}$-cycle from different clauses but associated to the same variable.
First, as $(t',t)$ is part of the $\pco_{h^\alpha_\varphi}$-cycle, there exists a $\pco_{h^\alpha_\varphi}$-predecessor of $t'$, $t''$ s.t. $(t'', t') \in \pco_{h^\alpha_\varphi}$ is part of the $\pco_{h^\alpha_\varphi}$-cycle. As we mentioned before, $(t'', t') \in \pco_{h^\alpha_\varphi} \setminus (\so \cup \overline{\wro})^+$. Then, there must exist a key $y$ and a read event $r_y$ s.t. $\writeVar{t''}{y}$, $(t', r_y) \in \overline{\wro}_y$ and $(t'', r) \in (\so \cup \overline{\wro});\po^*$. Two cases arise:
\begin{itemize}
    \item \underline{$t''$ is not associated to clause $i$}: In this case, as both $t', t''$ write variable $y$, by construction of $h_\varphi$ we observe that $\var{t''} = \var{t'}$. Thus, we denote $\tilde{t} = t''$ and $\hat{t} = t'$.
    
    \item \underline{$t''$ is associated to clause $i$}: In this case, $t'' \neq \lnot t'$ as no transaction in $h$ have both $t'$ and $\lnot t'$ as $(\so \cup \overline{\wro})$-predecessors. Hence, as no clause has two literals referring to the same variable, $\var{t'} \neq \var{t''}$. Thus, as $t''$ and $t'$ have one common key, we deduce that $t'' = \lnot t_i^{j-1\bmod 3}$ and $y = c_i^{j-1 \bmod 3}$. Thus, as $(t', r) \in \overline{\wro}_{c_i^{j-1 \bmod 3}}$, we can conclude that $\trans{r_y} = S_i^{j-1\bmod 3}$ and $(t'',S_i^{j-1 \bmod 3}) \in \overline{\wro}_{\var{t''}_i}$. As $t'' \neq t$, there must exist a transaction $t'''$ s.t. $(t''', t'') \in \pco_{h^\alpha_\varphi}$ belongs to the $\pco_{h^\alpha_\varphi}$-cycle. Again, we observe two cases:
    \begin{itemize}
        \item \underline{$t'''$ is not associated to clause $i$}: In this case, by an analogous argument, we observe that $\var{t'''} = \var{t''}$. Thus, we denote $\tilde{t} = t'''$ and $\hat{t} = t''$.
        
        \item \underline{$t'''$ is associated to clause $i$}: By the same reasoning as before, $t''' = \lnot t_i^{j-2 \bmod 3}$ and $(t''', S_i^{j-2 \bmod 3}) \in \overline{\wro}_{\var{t'''}_i}$. Moreover, as $t''' \neq t$, there must exist a transaction $t''''$ s.t. $(t'''', t''') \in \pco_{h^\alpha_\varphi}$ belongs to the $\pco_{h^\alpha_\varphi}$-cycle. Moreover, $t''''$ is not associated to clause $i$, as, once more, we would deduce that $t'''' = \lnot t_i^{j-3 \bmod 3}$ and that $(t'''', S_i^{j-3\bmod 3}) \in \overline{\wro}_{\var{t''''}_i}$; which is impossible as by the construction of $h^\alpha_\varphi$ is satisfied. Hence, $t''''$ and $t'''$ belong to different clauses and $\var{t''''} = \var{t'''}$. We denote in this case $\tilde{t} = t''''$ and $\hat{t} = t'''$.    
    \end{itemize}
    
\end{itemize}

Finally, we reach a contradiction with the help of Lemmas~\ref{lemma:np-complete:co-edge-impossible} and ~\ref{lemma:np-complete:co-cycle-necessary-wr}. 
On one hand, by the choice of transactions $\hat{t}$ and $\tilde{t}$, we know that $\var{\hat{t}} = \var{\tilde{t}}$ and there exist indices $\tilde{i}, \tilde{j}$ s.t. $\deleteVar{\tilde{t}}{\var{l_{\tilde{i}}^{\tilde{j}}}}$. 
Moreover, $\hat{t} \neq \lnot \tilde{t}_{\tilde{i}}^{\tilde{j}+1 \bmod 3}$ as they belong to different clauses. Thus, as $(\tilde{t}, \hat{t}) \in \pco_{h^\alpha_\varphi} \setminus (\so \cup \overline{\wro})^+$, by \Cref{lemma:np-complete:co-cycle-necessary-wr} we deduce that $(\tilde{t}, S_{\tilde{i}}^{\tilde{j}}) \in \overline{\wro}_{\var{\tilde{t}}_i}$. 
On the other hand, we also know that there exist indices $\hat{i}, \hat{j}$ s.t. $\hat{t}$ is associated to the literal $l_{\hat{i}}^{\hat{j}}$ and $(\hat{t}, S_{\hat{i}}^{\hat{j}}) \in \overline{\wro}_{\var{t}_{\hat{i}}}$. 
Hence, by construction of $h^\alpha_\varphi$, as $\var{\hat{t}} = \var{\tilde{t}}$, $(\tilde{t}, S_{\tilde{i}}^{\tilde{j}}) \in \overline{\wro}_{\var{\tilde{t}}_i}$ and $(\hat{t}, _{\hat{i}}^{\hat{j}}) \in \overline{\wro}_{\var{\hat{t}}_{\hat{i}}}$, we deduce that $\signTransaction{\hat{t}} = \signTransaction{\tilde{t}}$. However, by \Cref{lemma:np-complete:co-edge-impossible}, we deduce that $(\tilde{t}, \hat{t}) \not\in \pco_{h^\alpha_\varphi} \setminus (\so \cup \overline{\wro})^+$. This contradicts that $(\tilde{t}, \hat{t}){h^\alpha_\varphi}$ is part of the $\pco_{h^\alpha_\varphi}$-cycle. Thus, the initial hypothesis, that $\pco_{h^\alpha_\varphi}$ is cyclic, is false. In conclusion, $\pco_{h^\alpha_\varphi}$ is acyclic, so $h_\varphi$ is consistent as $h^\alpha_\varphi$ is a consistent witness of $h_\varphi$.

\end{proof}

\begin{lemma}
\label{lemma:np-complete:cons-sat}
For every boolean formula $\varphi$, if there is a consistent witness of $h$, then $\varphi$ is satisfiable.
\end{lemma}

\begin{proof}
Let $\overline{h} = \tup{T, \so, \overline{\wro}}$ be a consistent witness of $h_\varphi$. Hence, by Theorem~\ref{th:poly-consistency-saturable}, the relation $\pco_{\overline{h}} = \computeCO{\overline{h}}{(\so \cup \overline{\wro})^+}$ is acyclic. We use this fact to construct a satisfying assignment of $\varphi$. Let us call $u_i^j$ to the transaction s.t. $(u_i^j, S_i^j) \in \overline{\wro}_{\var{l_i^j}_i}$. Note that by construction of $h_\varphi$, $\deleteVar{u_i^j}{\var{l_i^j}_i}$, so $u_i^j$ is either $t_i^j$ or $\lnot t_i^j$.

We first prove that for every pair of pairs of indices $i, i', j, j'$, if $\var{u_i^j} = \var{u_{i'}^{j'}}$ then $\signTransaction{u_i^j} = \signTransaction{u_{i'}^{j'}}$. By contradiction, let $u_i^j, u_{i'}^{j'}$ be a pair of transactions s.t. $\var{u_i^j} = \var{u_{i'}^{j'}}$ and $\signTransaction{u_i^j} \neq \signTransaction{u_{i'}^{j'}}$. In such case, $\opSignTransaction{u_i^j} = \signTransaction{u_{i'}^{j'}}$. Thus, both transactions write $\var{u_i^j}_{(i', i)}^{\opSignTransaction{u_i^j}}$ and $\var{u_i^j}_{(i, i')}^{\signTransaction{u_i^j}}$. By axiom $\axrc$, as $(u_{i'}^{j'}, S_i^j) \in \overline{\wro}_{\var{u_i^j}_{(i', i)}^{\opSignTransaction{u_i^j}}}$ and $(u_i^j, S_i^j) \in \overline{\wro}$, we conclude that $(u_i^j, u_{i'}^{j'}) \in \pco_{\overline{h}}$. By a symmetric argument using $\var{u_i^j}_{(i, i')}^{\signTransaction{u_i^j}}$ we deduce that $(u_{i'}^{j'}, u_i^j) \in \pco_{\overline{h}}$. However, this is impossible as $\pco_{\overline{h}}$ is acylclic; so we conclude that indeed $\signTransaction{u_i^j} = \signTransaction{u_{i'}^{j'}}$.

Next, we construct a map that assign at each variable in $\varphi$ a value $0$ or $1$. Let $\alpha_h: \VarsFomula{\varphi} \to \{0,1\}$ be the map that assigns for each variable $\var{l_i^j}$ the value $1$ if $\signTransaction{u_i^j} = +$ and $0$ if $\signTransaction{u_i^j} = -$. Note that this map is well defined as, by the previous paragraph, if two literals $l_i^j, l_{i'}^{j'}$ share variable, then their respective transactions $u_i^j, u_{i'}^{j'}$ have the same sign.

Finally, we prove that $\varphi$ is satisfied with this assignment. By construction of $\alpha_h$, for every pair of indices $i,j$, $l_i^j[\alpha_h(\var{l_i^j})/\var{l_i^j}]$ is $\btrue$ if and only if $(t_i^j, S_i^j) \in \overline{\wro}_{\var{l_i^j}}$. Moreover, as $\pco_{\overline{h}}$ is acyclic, by \Cref{lemma:np-complete:co-cycle-clause}, we know that for each $i$ there exists a $j$ s.t. $u_i^j \neq \lnot t_i^j$. Hence, for this $j$, $u_i^j$ must be $t_i^j$ as $u_i^j$ is either $t_i^j$ or $\lnot t_i^j$. Therefore, every clause is satisfied using $\alpha_h$ as assignment; so $\varphi$ is satisfiable.

\end{proof}

%% file: figures-tex/np-code.tex
\begin{figure}[t]
	\centering
	\begin{subfigure}[b]{.32\textwidth}
		\resizebox{\textwidth}{!}{
			\begin{tikzpicture}[->,>=stealth',shorten >=1pt,auto,node distance=3cm,
				semithick, transform shape]
				\node[draw, rounded corners=2mm,outer sep=0] (t) at (0, 0) {\begin{tabular}{l}
                    $\iinsert{\{\var{l_{i}}_{(i,1)}^+ : 1\}}$ \\
					\ldots \\
					$\iinsert{\{\var{l_{i}}_{(i, i-1)}^+ : 1\}}$ \\
					$\iinsert{\{\var{l_{i}}_{(i, i+1)}^+ : 1\}}$ \\
					\ldots \\ 
					$\iinsert{\{\var{l_{i}}_{(i, n)}^+ : 1\}}$ \\
                    $\iinsert{\{\var{l_{i}}_{(1,i)}^- : 1\}}$ \\
					\ldots \\
					$\iinsert{\{\var{l_{i}}_{(i-1,i)}^- : 1\}}$ \\
					$\iinsert{\{\var{l_{i}}_{(i+1,i)}^- : 1\}}$ \\
					\ldots \\
					$\iinsert{\{\var{l_{i}}_{(n,i)}^- : 1\}}$ \\
					$\iinsert{\{\var{c_i^j} : 1\}}$ \\ 
					$\idelete{\lambda r: \keyof{r} = \var{l_i^j}_i} $
				\end{tabular}};		
			\end{tikzpicture}  
			
		}
		\caption{Transaction $t_i^j$}
		\label{fig:np-code:1}
	\end{subfigure}
	\centering
	\begin{subfigure}[b]{.32\textwidth}
		\resizebox{\textwidth}{!}{
			\begin{tikzpicture}[->,>=stealth',shorten >=1pt,auto,node distance=3cm,
				semithick, transform shape]
				\node[draw, rounded corners=2mm,outer sep=0] (t) at (0, 0) {\begin{tabular}{l}
                    $\iinsert{\{\var{l_{i}}_{(i, 1)}^- : 1\}}$ \\
					\ldots \\
					$\iinsert{\{\var{l_{i}}_{(i, i-1)}^- : 1\}}$ \\
					$\iinsert{\{\var{l_{i}}_{(i, i+1)}^- : 1\}}$ \\
					\ldots \\
					$\iinsert{\{\var{l_{i}}_{(i, n)}^- : 1\}}$ \\
					$\iinsert{\{\var{l_{i}}_{(1,i)}^+ : 1\}}$ \\
					\ldots \\
					$\iinsert{\{\var{l_{i}}_{(i-1,i)}^+ : 1\}}$ \\
                    $\iinsert{\{\var{l_{i}}_{(i+1, i)}^+ : 1\}}$ \\
					\ldots \\
                    $\iinsert{\{\var{l_{i}}_{(n, i)}^+ : 1\}}$ \\
					$\iinsert{\{\var{c_i^j} : 1\}}$ \\ 
                    $\iinsert{\{\var{c_i^{j-1 \bmod 3}} : 1\}}$ \\ 
					$\idelete{\lambda r: \keyof{r} = \var{l_i^j}_i} $
				\end{tabular}};		
			\end{tikzpicture}  
			
		}
		\caption{Transaction $\lnot t_i^j$}
		\label{fig:np-code:2}
	\end{subfigure}
	\centering
	\begin{subfigure}[b]{.24\textwidth}
		\resizebox{\textwidth}{!}{
			\begin{tikzpicture}[->,>=stealth',shorten >=1pt,auto,node distance=3cm,
				semithick, transform shape]
				\node[draw, rounded corners=2mm,outer sep=0] (t) at (0, 0) {\begin{tabular}{l} 
					$\iinsert{\{ x \in \mathtt{A}_i^j : 0\}} $ \\
					$\eselect(\lambda r: \btrue)$ 
				\end{tabular}};		
			\end{tikzpicture}  
			
		}
		\caption{Transaction $S_i^j$}
		\label{fig:np-code:3}
	\end{subfigure}
	\centering
\vspace{-2mm}
	\caption{Description in full detail of the transactions $t_i^j$, $\lnot t_i^j$ and $S_i^j$ described in the proof of theorem \ref{th:extend-np-complete} assuming $\signTransaction{t_i^j} = +$; where $\mathtt{A}_i^j$ is the set of auxiliary variables for $S_i^j$. The case where $\signTransaction{t_i^j} = -$ is analogous replacing in the first two instructions of both $t_i^j$ and $\lnot t_i^j$ $+$ by $-$ and vice versa. %
	}
	\label{fig:np-code}
\end{figure}

%% file: figures-tex/cycles-clause.tex
\begin{figure}[t]
    \resizebox{\textwidth}{!}{
        \footnotesize
        \begin{tabular}{|c|c|c|c|}
        \hline & & &\\ [-2mm]
        \begin{subfigure}[b]{.24\textwidth}
            \centering
            \resizebox{\textwidth}{!}{
    
            \begin{tikzpicture}[->,>=stealth,shorten >=1pt,auto,node distance=4cm,
                semithick, transform shape]
                \def \R {2}
                \def \r {1.2}
                \node[transaction state, text=black] (sx) at (0,0) {$S_i^{x}$};
                \node[transaction state, text=black] (nox) at ([shift=({30:\R})]sx) {$\lnot x_i$};
                \node[transaction state, text=black] (sz) at ([shift=({-30:\R})]nox) {$S_i^z$};
                \node[transaction state, text=black] (z) at ([shift=({-90:\R})]sz) {$z_i$};
                \node[transaction state, text=black] (sy) at ([shift=({210:\R})]z) {$S_i^{y}$};
                \node[transaction state, text=black] (noy) at ([shift=({-90:\R})]sx) {$\lnot y_i$};
                \node[transaction state, text=black] (x) at ([shift=({-90:\r})]nox) {$x_i$};
                \node[transaction state, text=black] (y) at ([shift=({30:\r})]noy) {$y_i$};
                \node[transaction state, text=black] (noz) at ([shift=({150:\r})]z) {$\lnot z_i$};
                \path (noy) edge[wrColor, left] node {$\wro_{c_i^1}$} (sx);
                \path (z) edge[wrColor, below] node [shift={(0.2,0.0)}] {$\wro_{c_i^2}$} (sy);
                \path (nox) edge[wrColor, above] node {$\wro_{c_i^3}$} (sz);
    
                \path (x) edge[wrColor, above, dashed] node {$\wro_{x_i}$} (sx);
                \path (y) edge[wrColor, left, dashed] node {$\wro_{y_i}$} (sy);
                \path (noz) edge[wrColor, left, pos=0.6, dashed] node {$\wro_{z_i}$} (sz);
    
                \path (x) edge[double equal sign distance, coColor, left] node {$\co$} (noy);
                \path (noz) edge[double equal sign distance, coColor, right] node {$\co$} (nox);
    
                \path (y) edge[double equal sign distance, coColor, below] node {$\co$} (z);

            \end{tikzpicture}
            }
            
            \caption{No $\co$-cycle when all literals are satisfied}
            \label{fig:cycles-clause:three-top}
        \end{subfigure}
    
        &  \begin{subfigure}[b]{.24\textwidth}
            \centering
            \resizebox{\textwidth}{!}{
            \begin{tikzpicture}[->,>=stealth,shorten >=1pt,auto,node distance=4cm,
                semithick, transform shape]
                \def \R {2}
                \def \r {1.2}
                \node[transaction state, text=black] (sx) at (0,0) {$S_i^{x}$};
                \node[transaction state, text=black] (nox) at ([shift=({30:\R})]sx) {$\lnot x_i$};
                \node[transaction state, text=black] (sz) at ([shift=({-30:\R})]nox) {$S_i^z$};
                \node[transaction state, text=black] (z) at ([shift=({-90:\R})]sz) {$z_i$};
                \node[transaction state, text=black] (sy) at ([shift=({210:\R})]z) {$S_i^{y}$};
                \node[transaction state, text=black] (noy) at ([shift=({-90:\R})]sx) {$\lnot y_i$};
                \node[transaction state, text=black] (x) at ([shift=({-90:\r})]nox) {$x_i$};
                \node[transaction state, text=black] (y) at ([shift=({30:\r})]noy) {$y_i$};
                \node[transaction state, text=black] (noz) at ([shift=({150:\r})]z) {$\lnot z_i$};
                \path (noy) edge[wrColor, left] node {$\wro_{c_i^1}$} (sx);
                \path (z) edge[wrColor, below] node [shift={(0.2,0.0)}] {$\wro_{c_i^2}$} (sy);
                \path (nox) edge[wrColor, above] node {$\wro_{c_i^3}$} (sz);
    
                \path (nox) edge[wrColor, left, pos=0.2, dashed] node [shift={(-0.2,0.0)}] {$\wro_{x_i}$} (sx);
                \path (y) edge[wrColor, left, dashed] node {$\wro_{y_i}$} (sy);
    
                \path (noz) edge[wrColor, left, pos=0.6, dashed] node {$\wro_{z_i}$} (sz);

                \path (nox) edge[double equal sign distance, coColor, right] node {$\co$} (noy);

                \path (noz) edge[double equal sign distance, coColor, right] node {$\co$} (nox);
                \path (y) edge[double equal sign distance, coColor, below] node {$\co$} (z);

            \end{tikzpicture}
            }
            
            \caption{No $\co$-cycle when two literals are satisfied}
            \label{fig:cycles-clause:two-top}
        \end{subfigure}
    
        &
        \begin{subfigure}[b]{.24\textwidth}
            \centering
            \resizebox{\textwidth}{!}{
    
            \begin{tikzpicture}[->,>=stealth,shorten >=1pt,auto,node distance=4cm,
                semithick, transform shape]
                \def \R {2}
                \def \r {1.2}
                \node[transaction state, text=black] (sx) at (0,0) {$S_i^{x}$};
                \node[transaction state, text=black] (nox) at ([shift=({30:\R})]sx) {$\lnot x_i$};
                \node[transaction state, text=black] (sz) at ([shift=({-30:\R})]nox) {$S_i^z$};
                \node[transaction state, text=black] (z) at ([shift=({-90:\R})]sz) {$z_i$};
                \node[transaction state, text=black] (sy) at ([shift=({210:\R})]z) {$S_i^{y}$};
                \node[transaction state, text=black] (noy) at ([shift=({-90:\R})]sx) {$\lnot y_i$};
                \node[transaction state, text=black] (x) at ([shift=({-90:\r})]nox) {$x_i$};
                \node[transaction state, text=black] (y) at ([shift=({30:\r})]noy) {$y_i$};
                \node[transaction state, text=black] (noz) at ([shift=({150:\r})]z) {$\lnot z_i$};
                \path (noy) edge[wrColor, left] node {$\wro_{c_i^1}$} (sx);
                \path (z) edge[wrColor, below] node [shift={(0.2,0.0)}] {$\wro_{c_i^2}$} (sy);
                \path (nox) edge[wrColor, above] node {$\wro_{c_i^3}$} (sz);
    
                \path (nox) edge[wrColor, left, pos=0.2, dashed] node [shift={(-0.2,0.0)}] {$\wro_{x_i}$} (sx);
                \path (noy) edge[wrColor, below, pos=0.5, dashed] node [shift={(-0.1,0.0)}] {$\wro_{y_i}$} (sy);
                \path (noz) edge[wrColor, left, pos=0.6, dashed] node {$\wro_{z_i}$} (sz);
    
                \path (nox) edge[double equal sign distance, coColor, right] node {$\co$} (noy);
                \path (noy) edge[double equal sign distance, coColor, above] node {$\co$} (z);
                \path (noz) edge[double equal sign distance, coColor, right] node {$\co$} (nox);

            \end{tikzpicture}
            }
            
            \caption{No $\co$-cycle when only one literal is satisfied}
            \label{fig:cycles-clause:one-top}
        \end{subfigure}
    
        &  \begin{subfigure}[b]{.24\textwidth}
            \centering
            \resizebox{\textwidth}{!}{
            \begin{tikzpicture}[->,>=stealth,shorten >=1pt,auto,node distance=4cm,
                semithick, transform shape]
                \def \R {2}
                \def \r {1.2}
                \node[transaction state, text=black] (sx) at (0,0) {$S_i^{x}$};
                \node[transaction state, text=black] (nox) at ([shift=({30:\R})]sx) {$\lnot x_i$};
                \node[transaction state, text=black] (sz) at ([shift=({-30:\R})]nox) {$S_i^z$};
                \node[transaction state, text=black] (z) at ([shift=({-90:\R})]sz) {$z_i$};
                \node[transaction state, text=black] (sy) at ([shift=({210:\R})]z) {$S_i^{y}$};
                \node[transaction state, text=black] (noy) at ([shift=({-90:\R})]sx) {$\lnot y_i$};
                \node[transaction state, text=black] (x) at ([shift=({-90:\r})]nox) {$x_i$};
                \node[transaction state, text=black] (y) at ([shift=({30:\r})]noy) {$y_i$};
                \node[transaction state, text=black] (noz) at ([shift=({150:\r})]z) {$\lnot z_i$};
                \path (noy) edge[wrColor, left] node {$\wro_{c_i^1}$} (sx);
                \path (z) edge[wrColor, below] node [shift={(0.2,0.0)}] {$\wro_{c_i^2}$} (sy);
                \path (nox) edge[wrColor, above] node {$\wro_{c_i^3}$} (sz);
    
                \path (nox) edge[wrColor, left, pos=0.2, dashed] node [shift={(-0.2,0.0)}] {$\wro_{x_i}$} (sx);
                \path (noy) edge[wrColor, below, pos=0.5, dashed] node [shift={(-0.1,0.0)}] {$\wro_{y_i}$} (sy);
                \path (z) edge[wrColor, pos=0.8, left, dashed] node {$\wro_{z_i}$} (sz);
    
                \path (z) edge[double equal sign distance, coColor, left] node {$\co$} (nox);
                \path (nox) edge[double equal sign distance, coColor, right] node {$\co$} (noy);
                \path (noy) edge[double equal sign distance, coColor, above] node {$\co$} (z);

            \end{tikzpicture}
            }
            
            \caption{One $\co$-cycle appears when no literal is satisfied}
            \label{fig:cycles-clause:all-bot}
        \end{subfigure}

        \\ \hline			  
        \end{tabular}
    }
    \caption{Transformation of the clause $C_i = x_i \lor y_i \lor \lnot z_i$ into part of the history. Solid $\wro$-edges in $h_\varphi$ represent the constraints of the clause while dashed $\wro$-edges, belonging only to the witnesses of $h_\varphi$, reflect the literals satisfied.
    }
    \label{fig:cycles-clause}
\end{figure}

%% file: figures-tex/gadget-consistency.tex
\begin{figure}[t]
    \resizebox{\textwidth}{!}{
        \footnotesize
        \begin{tabular}{|c|}
        \hline \\ [-2mm]
        \begin{subfigure}[b]{\textwidth}
            \centering
            \resizebox{\textwidth}{!}{
    
            \begin{tikzpicture}[->,>=stealth,shorten >=1pt,auto,node distance=4cm,
                semithick, transform shape]
                \def \R {2}
                \def \r {1.2}
                \node[transaction state, text=black] (nox1) at (-5.5,2.25) {$\lnot t_1$};
                \node[transaction state, text=black] (x1) at (-4.5,2) {$t_1$};
                \node[transaction state, text=black] (dots) at (-3.5,1.75) {$\ldots$};
                \node[transaction state, text=black] (noxim1) at (-2.5,1.5) {$\lnot t_{i-1}$};
                \node[transaction state, text=black] (xim1) at (-1.5,1.25) {$t_{i-1}$};
                \node[transaction state, text=black] (noxi) at (-0.5,1) {$\lnot t_i$};
                \node[transaction state, text=black] (xi) at (0.5,1) {$t_i$};
                \node[transaction state, text=black] (noxi1) at (1.5,1.25) {$\lnot t_{i+1}$};
                \node[transaction state, text=black] (xi1) at (2.5,1.5) {$t_{i+1}$};
                \node[transaction state, text=black] (dotsn) at (3.5,1.75) {$\ldots$};
                \node[transaction state, text=black] (noxn) at (4.5,2) {$\lnot t_n$};
                \node[transaction state, text=black] (xn) at (5.5,2.25) {$t_n$};
                \node[transaction state, text=black] (sx) at (0,-0.5) {$S_i^j$};
    
                \path (x1) edge[wrColor, bend right=31, below] node [shift={(-0.35,0.05)}]{$\wro_{x_{(1,i)}^+}$} (sx);
                \path (xim1) edge[wrColor,  bend left=10, below] node [shift={(-0.55,0.185)}] {$\wro_{x_{(i-1,i)}^+}$} (sx);
                \path (xi1) edge[wrColor, bend left=25, below] node [shift={(0.72,0.3)}]{$\wro_{x_{(i+1,i)}^+}$} (sx);
                \path (xn) edge[wrColor, bend left=44, below] node [shift={(0.25,0.0)}]{$\wro_{x_{(n,i)}^+}$} (sx);

                \path (nox1) edge[wrColor, bend right=44, below] node [shift={(-0.25,0.0)}] {$\wro_{x_{(1,i)}^-}$} (sx);
                \path (noxim1) edge[wrColor, bend right=25, below] node [shift={(-0.72,0.3)}] {$\wro_{x_{(i-1,i)}^-}$} (sx);
                \path (noxi1) edge[wrColor, bend right=10, below] node [shift={(0.55,0.185)}] {$\wro_{x_{(i+1,i)}^-}$} (sx);
                \path (noxn) edge[wrColor, bend left=31, below] node[shift={(0.35,0.05)}] {$\wro_{x_{(n,i)}^-}$} (sx);

                \path (xi) edge[wrColor, dashed, left] node [shift={(0.1,0.25)}] {$\wro_{x_i}$} (sx);
                \path (xi) edge[coColor, double equal sign distance, below] node {$\co$} (noxi1);
                \path (xi) edge[coColor, double equal sign distance, bend right=25, below] node {$\co$} (noxim1);
                \path (xi) edge[coColor, double equal sign distance, bend right=25, below] node {$\co$} (nox1);
                \path (xi) edge[coColor, double equal sign distance, bend left=25, below] node {$\co$} (noxn);

            \end{tikzpicture}
            }
            
        \end{subfigure}

        \\ \hline			  
        \end{tabular}
    }
\caption{Commit edges between transactions of different sign associated to variable $x = \var{l_i^j}$. Superindices $j$ are omitted for legibility. For simplicity on the Figure, we assume that $\signTransaction{t_k} = +$ and $\signTransaction{\lnot t_k} = -$; the situation generalizes for any other setting. If $S_i^j$ would read $x_i$ from $t_i$ in a witness $h'$ of $h_\varphi$ (respectively $\lnot t_i$), for every $i' \neq i$ $(t_i, \lnot t_{i'}) \in \co$, (resp. $(\lnot t_{i}, t_{i'}) \in \co$).}
\label{fig:gadget-consistency}
\end{figure}

%% file: appendices/newproofs/proofs-algorithm.tex
\subsection{Proof of \Cref{th:csob}.}
\label{ssec:proof-algorithm}

\csobTheorem*
The proof of \Cref{th:csob} is a consequence of \Cref{lemma:csob:consistent-true,lemma:csob:true-consistent}.

\begin{lemma}
\label{lemma:csob:seen-false}
Let $h = \tup{T, \so, \wro}$ be a client history, $P = \tup{T_P, M_P}$ be a consistent prefix of $h$ and $t \in T \setminus T_P$. If $(P \cup \{t\}) \in \mathtt{seen}$ then \csobAlgorithm$(h, P \cup \{t\})$ returns $\bfalse$.
\end{lemma}

\begin{proof}
If $(P \cup \{t\}) \in \mathtt{seen}$, then $P \cup \{t\}$ has been to $\mathtt{seen}$ added at line~\ref{algorithm:csob:add-prefix-seen} of \Cref{algorithm:csob}. To execute such instruction, the condition at line~\ref{algorithm:csob:no-conflict-invalid}, \csobAlgorithm$(\hist, P \cup \{t\})$ returns $\btrue$, does not hold; which let us conclude the result.
\end{proof}

\begin{lemma}
\label{lemma:csob:consistent-true}
Let $h = \tup{T, \so, \wro}$ be a client history whose isolation configuration is stronger than $\RC$. If $h$ is consistent, Algorithm~\ref{algorithm:checksobound} returns $\btrue$.
\end{lemma}

\begin{proof}
Let $h$ be a consistent history that satisfies the hypothesis of the Lemma. As $h$ is consistent, let $\overline{h} = \tup{T, \so, \overline{\wro}}$ be a witness of $h$ and let $\exec = \tup{\overline{h}, \co}$ be a consistent execution of $\overline{h}$. We first reduce the problem to prove that \Cref{algorithm:csob} returns true on a particular witness of $h$, a history $\hat{h}$ s.t. $h \subseteq \hat{h} \subseteq \overline{h}$. 

First, let $\pco, E_h$ and $X_h$ be defined as in \Cref{algorithm:checksobound} at lines~\ref{algorithm:checksobound:co}-\ref{algorithm:checksobound:xh}. As $\overline{h}$ is consistent, for every $\iread$ event $r$ and a variable $x$ s.t. $\wro_x^{-1}(r) \uparrow$, $\overline{\wro}_x^{-1}(r) \downarrow$ and $\where{r}(\valuewr{t_x^r}{x}) = 0$; where $t_x^r = \overline{\wro}_x^{-1}(r)$.

On one hand, if $E_h$ is empty, $X_h$ is empty as well. In such case, we denote $\hat{h} = h$. On the other hand, if $E_h \neq \emptyset$, for every $(r,x ) \in E_h$ we know that the transaction $t_x^r$ belongs to $\mathtt{0}_x^r$. Therefore, $X_h \neq \emptyset$. Thus, let $f$ be the map that assigns for every pair $(r, x) \in E_h$ the transaction $t_x^r$; and let $\hat{h} = \tup{T, \so, \hat{\wro}}$ be the history s.t. $\hat{h} = h \bigoplus_{(r,x) \in E_h} \wro_x(f(r,x), r)$. We observe that the fact that $\hat{h}$ is a history and $\overline{h}$ witnesses $\hat{h}$'s consistency using $\co$ is immediate as $\wro \subseteq \hat{\wro} \subseteq \overline{\wro}$.
Note that in both cases, the condition at line~\ref{algorithm:checksobound:inconsistent-case} does not hold. Therefore, to prove that \Cref{algorithm:checksobound} returns $\btrue$ it suffices to prove that $\csobAlgorithm(\hat{h}, \emptyset)$ returns $\btrue$.

We define an inductive sequence of prefixes based on $\co$ and show that they represent recursive calls to \Cref{algorithm:csob}. As a base case, let $P_0$ be the prefix with only $\init$ as transaction. Assuming that for every $j, 0 \leq j \leq i$, $P_i$ is defined, let $P_{i+1} = P_i \cup \{t_i\}$; where $t_i$ is the $i$-th transaction of $T$ according to $\co$. By construction of $\co$, $\pco \subseteq \co$. Hence, Property~\ref{def:consistent-extension:2} immediately holds. Moreover, as $\co$ witnesses $\overline{h}$'s consistency, Property~\ref{def:consistent-extension:3} also holds; so $P_i \triangleright_{t_{i+1}} P_{i+1}$.

We conclude showing by induction on the number of transactions that are not in the prefix that for every $i, 0 \leq i \leq |T|$, \csobAlgorithm$(\hat{h}, P_i)$ returns $\btrue$.

\begin{itemize}
    \item \underline{Base case:} The base case is $i = |T|$. In such case, $P_{|T|}$ contains all transactions in $T$. Therefore, the condition at line~\ref{algorithm:csob:base-case} in \Cref{algorithm:csob} holds and the algorithm returns $\btrue$.
    
    \item \underline{Inductive case:} The inductive hypothesis guarantees that for every $k, i \leq k \leq |T|$, \csobAlgorithm$(\hat{h}, P_{i})$ returns $\btrue$ and we show that \csobAlgorithm$(\hat{h}, P_{i-1})$ also returns $\btrue$. By definition of $P_{i}$, $T_{P_i} = T_{P_{i-1}} \cup \{t_i\}$. In particular, $|P_i| \neq |T|$ and $P_{i-1} \triangleright_{t_{i+1}} P_i$. In addition, by induction hypothesis, we know that \csobAlgorithm$(\hat{h}, P_{i})$ returns $\btrue$. Hence, by \Cref{lemma:csob:seen-false}, $P_{i} \not\in \mathtt{seen}$.
    Altogether, we deduce that \csobAlgorithm$(\hat{h}, P_{i -1 })$ returns $\btrue$.
\end{itemize}

\end{proof}

\begin{lemma}
\label{lemma:csob:chain-prefixes}
Let $\hat{h} = \tup{T, \so, \hat{\wro}}$ be a client history and $P = \tup{T_P, M_P}$ be a consistent prefix. If \csobAlgorithm$(\hat{h},P)$ returns $\btrue$, there exist distinct transactions $t_i \in T \ T_P$ and a collection of consistent prefixes $P_i = \tup{T_P, M_P}$ s.t. $P_{i} = P_{i-1} \cup \{t_{i}\}$, $P_{i-1} \triangleright_{t_i} P_i$ and \csobAlgorithm$(\hat{h},P_i)$ returns $\btrue$; where $|T_P| < i \leq |T|$ and $P_{|T_P|} = P$.
\end{lemma}

\begin{proof}
Let $\hat{h}$ be a client history and $P = \tup{T_P, M_P}$ be a consistent prefix s.t. \csobAlgorithm$(\hat{h},P)$ returns $\btrue$. We prove the result by induction on the number of transactions not present in $T_P$.
The base case, when $|T_P| = |T|$, immediately holds as $T \setminus T_P = \emptyset$. Let us assume that the inductive hypothesis holds for any prefix containing $k$ transactions and let us show that it also holds for every consistent prefix with $k-1$ transactions. Let us thus assume that $|T_P| = k - 1$. As \csobAlgorithm$(\hat{h}, P)$ returns $\btrue$, it must reach line~\ref{algorithm:csob:no-conflict-valid} in \Cref{algorithm:csob}. Hence, there must exist a transaction $t_{k} \in T \setminus T_P$ s.t. $P \triangleright_{t_k} (P \cup \{t_k\})$ and \csobAlgorithm$(\hat{h}, P \cup \{t_k\})$ returns $\btrue$. By induction hypothesis on $P \cup \{t_k\} = \tup{T_k, M_k}$, there exist a distinct transactions $t_i \in T \setminus T_k$ and a collection consistent prefixes $P_i$ s.t. $P_{i} = P_{i-1} \cup \{t_i\}$, $P_{i-1} \triangleright_{t_i} P_i$ and \csobAlgorithm$(\hat{h},P_i)$ returns $\btrue$; where $k < i \leq |T|$ and $P_k = P \cup \{t_k\}$. Thus, the inductive step holds thanks to prefix $P_k$. 
\end{proof}

\begin{lemma}
\label{lemma:csob:checksobound-returns-true}
Let $h = \tup{T, \so, \wro}$ be a client history and let $\pco$ be the relation defined as at line~\ref{algorithm:checksobound:co} in \Cref{algorithm:csob}. If \checksobound$(h)$ returns $\btrue$, there exists an extension $\hat{h} = \tup{T, \so, \hat{\wro}}$ of $h$ s.t. for every $\iread$ event $r$, variable $x$ and transaction $t$, (1) if $(t,r) \in \hat{\wro}_x \setminus \wro_x$ then $t \in \zeroWhere$, (2) if $\hat{\wro}_x^{-1}(r) \uparrow$ then $\oneWhere = \emptyset$, and (3) \csobAlgorithm$(\hat{h}, \emptyset)$ returns $\btrue$.
\end{lemma}

\begin{proof}
Let $h = \tup{T, \so, \wro}$ be a client history s.t. \checksobound$(h)$ returns $\btrue$ and let $\pco, E_h, X_h$ be the objects described in lines~\ref{algorithm:checksobound:co}-\ref{algorithm:checksobound:xh} in \Cref{algorithm:checksobound}. If there exists a pair $(r,x) \in E_h$ for which $\zeroWhere = \emptyset$, \checksobound$(h)$ returns $\bfalse$. Hence, $E_h$ is empty if and only if $X_h$ is empty. If $E_h = \emptyset$, \Cref{algorithm:checksobound} executes line~\ref{algorithm:checksobound:call-csob-h}. Thus, taking $\hat{h} = h$, conditions (1), (2) and (3) trivially hold. Otherwise, \Cref{algorithm:checksobound} executes line~\ref{algorithm:checksobound:else-case}. Once again, as \csobAlgorithm$(h, \emptyset)$ returns $\btrue$, there must exists $f \in X_h$ s.t. \csobAlgorithm$(\hat{h}, \emptyset)$ returns $\btrue$; where $\hat{h} = \bigoplus_{(r,x) \in E_h}\wro_x(f(r, x), r)$. Thanks to the definition of $f$ and $\hat{h}$ conditions (1), (2) and (3) are satisfied.
\end{proof}

\begin{lemma}
\label{lemma:csob:true-consistent}
Let $h = \tup{T, \so, \wro}$ be a client history whose isolation configuration is composed of $\{\SER, \SI, \PRE, \RC\}$ isolation levels. If \Cref{algorithm:csob} returns $\btrue$, $h$ is consistent.
\end{lemma}

\begin{proof} 
Let $h = \tup{T, \so, \wro}$ be a client history s.t. \checksobound$(h)$ returns $\btrue$ and let $\pco, E_h$ and $X_h$ be defined as at lines~\ref{algorithm:checksobound:co}-\ref{algorithm:checksobound:xh} in \Cref{algorithm:checksobound}. By Lemma~\ref{lemma:csob:checksobound-returns-true}, there exists an extension of $h$ $\hat{h} = \tup{T, \so, \hat{\wro}}$ s.t. for every $\iread$ event $r$, variable $x$ and transaction $t$, (1) if $\hat{\wro}_x^{-1}(r) \uparrow$ then $\oneWhere = \emptyset$, (2) if $(t,r) \in \hat{\wro}_x \setminus \wro_x$ then $t \in \zeroWhere$ and (3) \csobAlgorithm$(\hat{h}, \emptyset)$ returns $\btrue$. By \Cref{lemma:csob:chain-prefixes} applied on $\hat{h}$ and $\emptyset$, there exist distinct transactions $t_i \in T$ and a collection of prefixes of $h$, $P_i = \tup{T_i, M_i}$, s.t. $P_i = P_{i-1} \cup \{t_i\}$, $P_{i-1} \triangleright_{t_i} P_i$ and $\csobAlgorithm(\hat{h}, P_i)$ returns $\btrue$; where $P_0 = \emptyset$ and $0 < i \leq |T|$. Let $\co$ be the total order based on the aforementioned transactions $t_i$, i.e. $\co= \{(t_i, t_j) \ | \ i < j\}$. We construct a full history that extends $\hat{h}$ employing $\co$ and taking into account the isolation level of each transaction. 

For every $\iread$ event $r$, key $x$ and visibility relation $\mathsf{v} \in \visibilitySet{\isolation{h}(\trans{r})}$, let $t_\mathsf{v}, t_x^r$ be the transactions defined as follows:
\begin{align}
\label{eq:transactions-extension-algorithm}
t_\mathsf{v}^x = \max_{\co}\{ t' \in T \ | \ \writeVar{t'}{x} \ \land \ \visibilityRelationInstanceApplied[\co]{t'}{r}{x}\} \nonumber\\
t_x^r = \max_{\co}\{t_\mathsf{v}^x \ | \ \mathsf{v} \in \visibilitySet{\isolation{h}(\trans{r})}\}
\end{align}
Note that if $\mathsf{v}$ is a visibility relation associated to an axiom from $\SER, \SI, \PRE,  \RA$ and $\RC$ isolation levels, transactions $t_{\mathsf{v}}^x$ and $t_x^r$ are well-defined as $\mathsf{v}(\init, r, x)$ holds.
Thus, let $\overline{\wro}_x = \hat{\wro}_x \cup \{(t_x^r, r) \ | \ \hat{\wro}_x^{-1}(r) \uparrow\}$ and $\overline{\wro} = \bigcup_{x \in \Vars} \overline{\wro}_x$. As $\overline{\wro}_x^{-1}$ is a total function and $\writeVar{\overline{\wro}_x^{-1}(r)}{x}$ we can conclude that $\overline{h} = \tup{T, \so, \overline{\wro}}$ is a full history. 

We prove that $\overline{h}$ is also a witness of $h$. For that, we show that for every read event $r$, every key $x$ and every transaction $t$, if $(t, r) \in \overline{\wro}_x \setminus \wro_x$, $t \in \zeroWhere$. Two cases arise: $(t, r) \in \hat{\wro}_x \setminus \wro_x$ and $(t, r) \in \overline{\wro}_x \setminus \hat{\wro}_x$. The first case is quite straightforward, as if $(t, r) \in \hat{\wro}_x \setminus \wro_x$, by Property (1) of \Cref{lemma:csob:checksobound-returns-true}, $t \in \zeroWhere$. The second case, $(t, r) \in \overline{\wro}_x \setminus \hat{\wro}_x$, is slightly more subtle. First, for every isolation level considered, if $(t, r) \in \overline{\wro}_x$ then $(t, \trans{r}) \in \pco$. Next, as \checksobound$(h)$ returns $\btrue$, the condition at line~\ref{algorithm:checksobound:co-cyclic} does not hold. Hence, as $\pco$ is acyclic, we deduce that $(\trans{r}, t) \not\in \pco$. In addition, as $(t, r) \in \overline{\wro}_x \setminus \hat{\wro}_x$, $\hat{\wro}_x^{-1}(r) \uparrow$. By Property (2) of \Cref{lemma:csob:checksobound-returns-true} employed during $\hat{h}$'s construction, we deduce that $\oneWhere = \emptyset$. In conclusion, as $(\trans{r}, t) \not\in \pco$ and $\oneWhere = \emptyset$, we conclude that $t \in \zeroWhere$.

Finally, we prove that $\co$ witnesses that $\overline{h}$ is consistent. Let $r$ be a $\iread$ event, $x$ be a key and $t_1, t_2$ be transactions s.t. $(t_1, r) \in \overline{\wro}_x$ and $\writeVar{t_2}{x}$. We prove that if there exists $\mathsf{v} \in \visibilitySet{\isolation{\overline{h}}(\trans{r})}$ s.t. $\visibilityRelationInstanceApplied[\co]{t_2}{r}{x}$ holds in $\overline{h}$ then $(t_2, t_1) \in \co$; which by \Cref{def:consistency-full}, we know it implies that $\overline{h}$ is consistent. Note that if $(t_1, r) \in \overline{\wro} \setminus \hat{\wro}$, by definition of $t_x^r$ the statement immediately holds; so we can assume without loss of generality that $(t_1, r) \in \hat{\wro}_x$. 

First, we note that proving that whenever $\visibilityRelationInstanceApplied[\co]{t_2}{r}{x}$ holds in $\overline{h}$, then $(t_2, t_1) \in \co$ is equivalent to prove that whenever $\visibilityRelationInstanceApplied[\co]{t_2}{r}{x}$ holds in $\overline{h}$, then $t_1 \not\in T_{i-1}$; where $i$ is the index of the transaction in $T$ s.t. $t_2 = t_i$.  

For every $i, 1 \leq i \leq |T|$ $P_{i-1} \triangleright_{t_i} P_i$, $\so \cup \hat{\wro} \subseteq \co$. Thus, by \Cref{def:consistency-full}, it suffices to show that for every read event $r$, $C_{\isolation{h}(\trans{r})}(\pco)(r)$ holds. For that, let $\hat{\pco} = \mathsf{FIX}(\lambda \mathit{R}: \computeCO{\hat{h}}{ \mathit{R}})(\so \cup \hat{\wro})^+$ be the partial commit order implied by $\hat{h}$.

As $\isolation{h}$ is composed of $\{\SER, \SI, \PRE, \RA, \RC\}$ isolation levels and $P_{i-1} \triangleright_{t_2} P_i$, by Property~\ref{def:consistent-extension:3} of \Cref{def:consistent-extension}, it suffices to prove that whenever $\visibilityRelationInstanceApplied[\co]{t_2}{r}{x}$, if $v \neq \axconf$ then $v(\hat{\pco}_{t_2}^{P_i})(t, r, x)$ holds in $\hat{h}$, while if $v = \axconf$, that there exists $t' \in T_{i-1}$ s.t. $v(\hat{\pco}_{t_2}^{P_i})(t', r, x)$ holds in $\hat{h}$; where $\hat{\pco}_{t_2}^{P_i}$ is obtained by applying \Cref{table:visibility-prefix} on $\hat{\pco}$. We analyze five different cases:

\begin{itemize}
    \item \underline{$\isolation{\overline{h}}(\trans{r}) = \SER$}: In this case, $\mathsf{\axser}(\co)(t_2, r, x)$ holds in $\overline{h}$ if and only if $(t_2, \trans{r}) \in \co$. As $\hat{\pco}_{t_2}^{P_i}$ totally orders $t_2$ and every other transaction in $T$ and $\hat{\pco}_{t_2}^{P_i} \subseteq \co$, we deduce that $(t_2, \trans{r}) \in \hat{\pco}_{t_2}^{P_i}$. Hence, $\mathsf{\axser}(\hat{\pco}_{t_2}^{P_i})(t_2, r, x)$ holds in $\hat{h}$.

    \item \underline{$\isolation{\overline{h}}(\trans{r}) = \SI$}: Two disjoint sub-cases arise: 
    \begin{itemize}
        \item \underline{$\mathsf{\axconf}(\co)(t_2, r, x)$ holds in $\overline{h}$}: This happens if and only if there exists a transaction $t_3$ and a key $y \in \Vars$ s.t. $\writeVar{t_3}{y}$, $\writeVar{\trans{r}}{y}$, $(t_2, t_3) \in \co^*$ and $(t_3, \trans{r}) \in \co$. Let $j$ be the index s.t. $t_3 = t_j$. Then, as $\hat{\pco}_{t_3}^{P_j}$ totally orders $t_3$ and every other transaction and $\hat{\pco}_{t_3}^{P_j} \subseteq \co$, $(t_2, t_3) \in (\hat{\pco}_{t_3}^{P_j})^*$ and $(t_3, \trans{r}) \in \hat{\pco}_{t_3}^{P_j}$. Thus, $\axconf(\hat{\pco}_{t_3}^{P_j})(t_2, r, x)$ holds in $\hat{h}$.
        
        \item \underline{$\mathsf{\axpre}(\co)(t_2, r, x)$ holds in $\overline{h}$ but $\mathsf{\axconf}(\co)(t_2, r, x)$ does not}: 
        We observe that $\mathsf{\axpre}(\co)(t_2, r, x)$ holds in $\overline{h}$ if there exists a transaction $t_3$ s.t. $(t_2, t_3) \in \co^*$ and $(t_3, \trans{r}) \in \so \cup \overline{\wro}$. If $(t_3, \trans{r}) \in \overline{\wro} \setminus (\so \cup \hat{\wro})$, by \Cref{eq:transactions-extension-algorithm} there exist $y \in \Vars$ and $v \in \visibilitySet{\SI}$ s.t. $v(\co)(t_3, r, y)$ holds in $\hat{h}$. Note that $v \neq \axconf$ as otherwise $\mathsf{\axconf}(\co)(t_2, r, x)$ would hold in $\overline{h}$. Hence, $v = \axpre$ and by transitivity of $\co$, we conclude that $\mathsf{\axpre}(\co)(t_2, r, x)$ holds in $\hat{h}$. As $\hat{\pco}_{t_2}^{P_i}$ totally orders $t_2$ with respect every other transaction in $t_2$ and $\hat{\pco}_{t_2}^{P_i} \subseteq \co$, we conclude that $\axpre(\hat{\pco}_{t_2}^{P_i})(t_2, r, x)$ holds in $\hat{h}$.
    \end{itemize}

    \item \underline{$\isolation{\overline{h}}(\trans{r}) = \PRE$}: In this case, $\mathsf{\axpre}(\co)(t_2, r, x)$ holds in $\overline{h}$ if and only if there exists a transaction $t_3$ s.t. $(t_2, t_3) \in \co^*$ and $(t_3, \trans{r}) \in \so \cup \overline{\wro}$. If $(t_3, \trans{r}) \in \overline{\wro} \setminus (\so \cup \hat{\wro})$, by \Cref{eq:transactions-extension-algorithm} there exist $y \in \Vars$ and $v \in \visibilitySet{\SI}$ s.t. $v(\co)(t_3, r, y)$ holds in $\hat{h}$. Hence, by transitivity of $\co$, we conclude that $\mathsf{\axpre}(\co)(t_2, r, x)$ holds in $\hat{h}$. As $\hat{\pco}_{t_2}^{P_i}$ totally orders $t_2$ with respect every other transaction in $t_2$ and $\hat{\pco}_{t_2}^{P_i} \subseteq \co$, we conclude that $\axpre(\hat{\pco}_{t_2}^{P_i})(t_2, r, x)$ holds in $\hat{h}$.
    
    \sloppy \item \underline{$\isolation{\overline{h}}(\trans{r}) = \RA$}: In this case,
    $\mathsf{\axra}(\co)(t_2, r, x)$ holds in $\overline{h}$ if and only if $(t_2, \trans{r}) \in \so \cup \overline{\wro}$. We observe that by \Cref{eq:transactions-extension-algorithm}, if $(t_2, \trans{r}) \in \overline{\wro} \setminus (\so \cup \hat{\wro})$, then $t_2 = t_x^r$ and $\mathsf{\axra}(\hat{\pco}_{t_2}^{P_i})(t_2, r, x)$ holds in $\hat{h}$. Hence, $(t_2, \trans{r}) \in \so \cup \hat{\wro}$; which is a contradiction. Thus, as $(t_2, \trans{r}) \in \so \cup \hat{\wro}$, $\mathsf{\axra}(\hat{\pco}_{t_2}^{P_i})(t_2, r, x)$ holds in $\hat{h}$.
    
    \item \underline{$\isolation{\overline{h}}(\trans{r}) = \RC$}: Similarly to the previous case, we observe that the formula $\mathsf{\axrc}(\co)(t_2, r, x)$ holds in $\overline{h}$ iff $(t_2, \trans{r}) \in (\so \cup \overline{\wro}); \po^*$. We observe that by \Cref{eq:transactions-extension-algorithm}, if $(t_2, \trans{r}) \in \overline{\wro} \setminus (\so \cup \hat{\wro}) ; \po^*$, then $t_2 = t_x^r$ and $\mathsf{\axrc}(\hat{\pco}_{t_2}^{P_i})(t_2, r, x)$ holds in $\hat{h}$. Therefore, $(t_2,r) \in \so \cup \hat{\wro} ; \po^*$; which is a contradiction. Thus, as $(t_2, \trans{r}) \in \so \cup \hat{\wro} ; \po^*$, $\mathsf{\axrc}(\hat{\pco}_{t_2}^{P_i})(t_2, r, x)$ holds in $\hat{h}$.

\end{itemize}

\end{proof}

%% file: appendices/newproofs/proofs-algorithm-complexity.tex
\subsection{Proof of \Cref{th:csob-poly}}
\label{ssec:algorithm-complexity}

\csobPolyTheorem*

The proof of Theorem~\ref{th:csob-poly} is split in two Lemmas: Lemma~\ref{lemma:csob:poly-csob-width} analyzes the complexity of Algorithm~\ref{algorithm:csob} while Lemma~\ref{lemma:csob:poly-csob-all} relies on the previous result to conclude the complexity of Algorithm~\ref{algorithm:checksobound}.

\input{algorithms-tex/check-consistent-extension}

\begin{lemma}
\label{lemma:csob:consistent-extension-correctness}
Let $h = \tup{T, \so, \wro}$ be a history, $P = \tup{T_P, M_P}$ be a consistent prefix of $h$ and $t \in T \setminus T_P$ be a transaction. \Cref{algorithm:consistent-extension} returns $\btrue$ if and only if $P \triangleright_t (P \cup \{t\})$.
\end{lemma}

\begin{proof}
Clearly, $P \cup \{t\}$ is an extension of $P$.$\isConsistentExtension{h}{P}{t}$ returns $\btrue$ if and only if conditions at lines~\ref{algorithm:consistent-extension:condition-1} and lines~\ref{algorithm:consistent-extension:condition-2} in \Cref{algorithm:consistent-extension} hold. This is equivalent to respectively satisfy Properties~\ref{def:consistent-extension:2} and \ref{def:consistent-extension:3} of \Cref{def:consistent-extension}. By \Cref{def:consistent-extension}, this is equivalent to $P \triangleright_t (P \cup \{t\})$.
\end{proof}

\begin{lemma}
\label{lemma:csob:consistent-extension-poly}
Let $h = \tup{T, \so, \wro}$ be a history and $k\in \mathbb{N}$ be a bound in $\isolation{h}$. For any consistent prefix $P = \tup{T_P, M_P}$ of $h$ and any transaction $t \in T \setminus T_P$, Algorithm~\ref{algorithm:consistent-extension} runs in $\mathcal{O}(|h|^{k+3})$.
\end{lemma}

\begin{proof}
We analyze the cost of \Cref{algorithm:consistent-extension}. First, as $\pco \subseteq T \times T$, by \Cref{lemma:necessary-co:polynomial-time}, line~\ref{algorithm:consistent-extension:co} runs in $\mathcal{O}(|h|^2 \cdot |h|^{k+1})$. Next, the condition at line~\ref{algorithm:consistent-extension:condition-1} can be checked in $\mathcal{O}(|T|)$. Finally, the condition at line~\ref{algorithm:consistent-extension:condition-2} can be checked in $\mathcal{O}(|T| \cdot k \cdot U)$; where $U$ is an upper-bound on the complexity of checking $\visConsPrefix{v}{P}{t}{r}$. With the aid of \Cref{lemma:evaluate-v:polynomial-time}, we deduce that $U \in \mathcal{O}(|h|^{k-2})$. Altogether, we conclude that \Cref{algorithm:consistent-extension} runs in $\mathcal{O}(|h|^{k+3})$.
\end{proof}

\begin{lemma}
\label{lemma:csob:bound-isolation}
Let $h = \tup{T, \so, \wro}$ be a client history. If $\isolationExecution{h}$ is composed of $\{\SER, \SI,  \PRE, \RA, \RC\}$ isolation levels, then $5$ is a bound of $\isolation{h}$. Moreover, if no transaction has $\SI$ as isolation, $4$ is a bound on $\isolation{h}$.
\end{lemma}

\begin{proof}
Let $h$ be a history as described in the hypothesis. First, all isolation levels in the set $\{\SER, \SI,  \PRE, \RA, \RC\}$ employ at most two axioms. Moreover, every axiom described employs at most $5$ quantifiers: three universal quantifiers and at most two existential quantifiers. Hence, $5$ is a bound on $\isolation{h}$. Note that $\axconf$ is the only axiom employing two existential quantifiers; so if no transaction employs $\SI$, $4$ bounds $\isolation{h}$.
\end{proof}

\begin{lemma}
\label{lemma:csob:poly-csob-width}
Let $h = \tup{T, \so, \wro}$ be a client history whose isolation configuration is composed of $\{\SER, \SI, \PRE, \RA, \RC\}$ isolation levels. Algorithm~\ref{algorithm:csob} runs in $\mathcal{O}(|h|^{\widthHistory{h}+9} \cdot \widthHistory{h}^{|\Keys|})$. Moreover, if no transaction has $\SI$ as isolation level, Algorithm~\ref{algorithm:csob} runs in $\mathcal{O}(|h|^{\widthHistory{h}+8})$.
\end{lemma}

\begin{proof}
For proving the result, we focus only on prefixes that are computable by Algorithm~\ref{algorithm:csob}. Let $h = \tup{T, \so, \wro}$ be a history. A prefix $P$ of $h$ is \emph{computable} if either $P= \emptyset$ or there exist a transaction $t$ and a prefix $P'$ s.t. $P = P'\cup \{t\}$ and $P'$ is computable. 

Intuitively, computable prefixes represent recursive calls of Algorithm~\ref{algorithm:csob} when employed by Algorithm~\ref{algorithm:checksobound}. Indeed, Algorithm~\ref{algorithm:checksobound} only employs Algorithm~\ref{algorithm:csob} at lines~\ref{algorithm:checksobound:call-csob-h} and \ref{algorithm:checksobound:call-csob-h-prime}. In both cases, $P' = \emptyset$ is the initial call to Algorithm~\ref{algorithm:csob}. Moreover, the condition at line~\ref{algorithm:csob:for-condition} justifies the recursive definition.

On one hand, we observe that any call to Algorithm~\ref{algorithm:csob} is associated to a computable prefix and on the other hand, Algorithm~\ref{algorithm:csob} does not explore two equivalent computable prefix thanks to the global variable $\mathtt{seen}$ (line~\ref{algorithm:csob:no-conflict-invalid}). Therefore, Algorithm~\ref{algorithm:csob} runs in $\mathcal{O}(N \cdot U)$; where $N$ is the number of distinct equivalence class of prefixes of $h$ and $U$ is an upper-bound on the running time of Algorithm~\ref{algorithm:csob} on a fixed prefix without doing any recursive call. 

We first compute an upper-bound of $N$. For any computable prefix $P$, we can deduce by induction on the length of $P$ that there exists transactions $t_i \in T_P$ and a collection of computable prefixes of $h$, $P_i = \tup{T_i, M_i}$ and transactions $t_i$ s.t. $P_{|T_P|} = P$, $P_i = P_{i-1} \cup \{t_i\}$ and $P_{i-1} \triangleright_{t_i} P_i$; where $P_0 = \emptyset$ and $ 0 < i \leq |T_P|$. The base case is immediate as $|T_P| = 0$ implies that $T' = \emptyset$ while the inductive step can be simply obtained by applying the recursive definition of computable prefix.

Let $P = \tup{T_P, M_P}$ be in what follows a computable prefix of $h$. We observe that both $T_P$ and $M_P$ are determined by its $\so$-maximal transactions. Let $t, t' \in T$ be a pair of transactions s.t. $(t, t') \in \so$ and $t' \in T_P$. As $t' \in T_P$ there must exist an index $i, 1 \leq i \leq |T_P|$ s.t. $P_i = P_{i-1} \cup \{t'\}$. Therefore, as $P_{i-1} \triangleright_{t'} P_{i-1}$, $t \in T_{i-1} \subseteq T_P$. In particular, if $t'$ is a $\so$-maximal transaction in $T_P$, all its $\so$-predecessors are also contained in $T_P$; and hence, $T_P$ can be characterized by its $\so$-maximal transactions. Moreover, by induction on the length of $P$ we can prove that for every key $x$, $M_P(x)$ is a $\so$-maximal transaction: the base case, $|T_P| = 0$ is immediate while the inductive step is obtained by the definition of $P \cup \{t''\}, t'' \not\in T_P$. Hence, the number of computable prefixes of a history is in $\mathcal{O}(|T|^{\widthHistory{h}} \cdot {\widthHistory{h}}^{|\Keys|})$. Thus, $N \in \mathcal{O}(|h|^{\widthHistory{h}} \cdot \widthHistory{h}^{|\Keys|})$. Moreover, if no transaction employs $\SI$ as isolation level, prefixes with identical transaction set coincide. Hence, in such case, $N \in \mathcal{O}(|h|^{\widthHistory{h}} )$.

We conclude the proof bounding $U$. If $|T_P| = |T|$, Algorithm~\ref{algorithm:csob} runs in $\mathcal{O}(1)$; so we can assume without loss of generality that $|T_P| \neq |T|$. In such case, $U$ represent the cost of executing lines~\ref{algorithm:csob:for-condition}-\ref{algorithm:csob:return-false} in \Cref{algorithm:csob}. Thus, $U \in \mathcal{O}((|T|-|T_P|)\cdot V)$; where $V$ is the cost of checking $P \triangleright_t (P \cup \{t\})$ for a transaction $t \in T \setminus T_P$. By \Cref{lemma:csob:consistent-extension-correctness}, Algorithm~\ref{algorithm:consistent-extension} can check if $P \triangleright_t (P \cup \{t\})$ and thanks to \Cref{lemma:csob:consistent-extension-poly}, Algorithm~\ref{algorithm:consistent-extension} runs in $\mathcal{O}(|h|^{k+3})$; where $k$ is a bound on $\isolationExecution{h}$. Thus, $U \in \mathcal{O}(|h|^{k+4})$.

Thanks to Lemma~\ref{lemma:csob:bound-isolation}, we conclude that \Cref{algorithm:csob} runs in $\mathcal{O}(|h|^{\widthHistory{h}+9} \cdot \widthHistory{h}^{|\Keys|})$ and, if no transaction employs $\SI$ as isolation level, then it runs in $\mathcal{O}(|h|^{\widthHistory{h}+8})$.
\end{proof}

\begin{lemma}
\label{lemma:csob:poly-csob-all}
\sloppy
Let $h = \tup{T, \so, \wro}$ be a client history whose isolation configuration is composed of $\{\SER, \SI,  \PRE, \RA, \RC\}$ isolation levels. Algorithm~\ref{algorithm:checksobound} runs in $\mathcal{O}(|h|^{\conflictsHistory{h} +\widthHistory{h}+9} \cdot \widthHistory{h}^{|\Keys|})$. Moreover, if no transaction has $\SI$ as isolation level, Algorithm~\ref{algorithm:checksobound} runs in $\mathcal{O}(|h|^{\conflictsHistory{h} + \widthHistory{h}+8})$.
\end{lemma}

\begin{proof}
Let $h = \tup{T, \so, \wro}$ be a history satisfying the hypothesis of the Lemma. We decompose our analysis in two sections, the first one where we analyze the complexity of executing lines~\ref{algorithm:checksobound:co}-\ref{algorithm:checksobound:xh} and second one where we analyze the complexity of executing lines~\ref{algorithm:checksobound:co-cyclic}-\ref{algorithm:checksobound:end}. We observe that by \Cref{lemma:csob:bound-isolation}, $5$ is a bound on $\isolation{h}$.

In line~\ref{algorithm:checksobound:co}, \Cref{algorithm:checksobound} computes $\pco$. On one hand, computing $(\so \cup \wro)^+$ is in $\mathcal{O}(|T|^3)$. On the other hand, as $\pco \subseteq T \times T$ and by Lemma~\ref{lemma:necessary-co:polynomial-time}, executing $\computeCO{h}{(\so \cup \wro)^+}$ is in $\mathcal{O}(|h|^6)$; we deduce that computing $\pco$ after compute  $(\so \cup \wro)^+$ is in $\mathcal{O}(|h|^8)$.

In line~\ref{algorithm:checksobound:eh}, \Cref{algorithm:checksobound} computes $E_h$. As $\wro$ is acyclic, for a given key $x$ and transaction $t$, $\valuewr{t}{x} \in \mathcal{O}(|T|)$. Therefore, computing $\oneWhere$ is in $\mathcal{O}(|T|)$ as we assume that for every $r \in \Vals$, computing $\where{r}(v) \in \mathcal{O}(1)$. Thus, computing $E_h$ is in $\mathcal{O}(|h|^3)$.

Finally, in line~\ref{algorithm:checksobound:xh}, \Cref{algorithm:checksobound} computes $X_h$. Note that $X_h$ can be seen is a $\bigtimes_{(r, x) \in E_h} \zeroWhere$. Computing each $\zeroWhere$ set is in $\mathcal{O}(|T|)$; so computing all of them is in $\mathcal{O}(|T| \cdot |E_h|)$. As each set $\zeroWhere$ is a subset of $T$, applying the cartesian-product definition of $X_h$ we can compute $X_h$ in $\mathcal{O}(|T|^{|E_h|})$. Therefore, as $|E_h| = \conflictsHistory{h}$, we conclude that computing $X_h$ is in $\mathcal{O}(|h| \cdot \conflictsHistory{h}+ |h|^{\conflictsHistory{h}})$ and that $|X_h| \in \mathcal{O}(|h|^{\conflictsHistory{h}})$. Altogether, as $\conflictsHistory{h} \leq |T|^2$, we deduce that computing lines~\ref{algorithm:checksobound:co}-\ref{algorithm:checksobound:xh} of Algorithm~\ref{algorithm:checksobound} is in $\mathcal{O}(|h|^8+|h|^{\conflictsHistory{h}})$.

Next, we analyze the complexity of executing lines~\ref{algorithm:checksobound:co-cyclic}-\ref{algorithm:checksobound:end}. Four disjoint cases arise, one per boolean condition in \Cref{algorithm:checksobound}. The first one, checking if $\pco$ is cyclic (line~\ref{algorithm:checksobound:co-cyclic}), is in $\mathcal{O}(|h|)$. The second one, checking if $\exists (r,x)\in E_h$ s.t. $\zeroWhere = \emptyset$ (line~\ref{algorithm:checksobound:inconsistent-case}), clearly runs in $\mathcal{O}(\conflictsHistory{h} \cdot |h|)$. The third one, checking if $E_h = \emptyset$ and executing \Cref{algorithm:csob} is in $\mathcal{O}(\conflictsHistory{h}+|h|^{\widthHistory{h}+9} \cdot \widthHistory{h}^{|\Keys|})$ thanks to \Cref{lemma:csob:poly-csob-width}. 

Finally we analyze the last case, computing an extension of $h$ for each mapping in $X_h$ and then executing \Cref{algorithm:checksobound} (lines~\ref{algorithm:checksobound:else-case}-\ref{algorithm:checksobound:end}). On one hand, computing each history is in $\mathcal{O}(|h|^3)$ as we require to define both $\so \subseteq T \times T$ and $\wro \subseteq \Vars \times T \times T$. On the other hand, as the size of each extension of $h$ is in $\mathcal{O}(|h|)$, executing \Cref{algorithm:checksobound} for a given history is in $\mathcal{O}(|X_h| \cdot |h|^{\widthHistory{h}+9} \cdot \widthHistory{h}^{|\Keys|})$ thanks again to \Cref{lemma:csob:poly-csob-width}. Altogether, for each mapping $f \in X_h$, executing lines~\ref{algorithm:checksobound:h-prime-definition}-\ref{algorithm:checksobound:call-csob-h-prime} is in $\mathcal{O}(|h|^{\widthHistory{h}+9} \cdot \widthHistory{h}^{|\Keys|})$. As $|X_h| \in \mathcal{O}(|h|^{\conflictsHistory{h}})$, we conclude that executing this last case is in $\mathcal{O}(|h|^{\conflictsHistory{h}+\widthHistory{h}+9}  \cdot \widthHistory{h}^{|\Keys|})$. 

We then conclude that Algorithm~\ref{algorithm:checksobound} runs in $\mathcal{O}(|h|^{\conflictsHistory{h} + \widthHistory{h} + 9} \cdot \widthHistory{h}^{|\Keys|})$. Moreover, if no transaction employs $\SI$ as isolation level, \Cref{lemma:csob:poly-csob-width} allows us to deduce that in such case \Cref{algorithm:checksobound} runs in $\mathcal{O}(|h|^{\conflictsHistory{h} + \widthHistory{h} + 9})$.
\end{proof}

%% file: algorithms-tex/check-consistent-extension.tex
\begin{algorithm}[t]
\caption{Checking if $P \triangleright_t (P \cup \{t\})$ holds in $h$}
\begin{algorithmic}[1]
\small
\Function {\textsc{isConsistentExtension}}{$\hist = \tup{T, \so, \wro}$, $P = \tup{T_P, M_P}$, $t$}
\Statex
\Comment{We assume $t \not\in T_P$.}

\Let $\pco = \mathsf{FIX}(\lambda \mathit{R}: \computeCO{h}{\mathit{R}})(\so \cup \wro)^+$
\label{algorithm:consistent-extension:co}%

\If{$\exists t' \in T \setminus T_P$ s.t. $(t', t) \in \pco$}
\label{algorithm:consistent-extension:condition-1}
\Comment{Condition \ref{def:consistent-extension:2}}
    \State \Return $\bfalse$
\EndIf

\ForAll{$r \in \readOp{h}$ s.t. $\trans{r} \not\in T_P \cup \{t\}, v \in \visibilitySet{\isolation{h}}(\trans{r})$}
\label{algorithm:consistent-extension:condition-2}

    \If{$\visConsPrefix{v}{P}{t}{r}$ does not hold in $h$}
    \label{algorithm:consistent-extension:vis-check}
    \Comment{Condition \ref{def:consistent-extension:3}}
    \State \Return $\bfalse$
    \EndIf

\EndFor

\State \Return $\btrue$
\label{algorithm:consistent-extension:return-true}

\EndFunction
\label{algorithm:consistent-extension:end}

\end{algorithmic}
\label{algorithm:consistent-extension}
\end{algorithm}